\documentclass[aps,prd,twocolumn,superscriptaddress,floatfix,letterpaper,nofootinbib]{revtex4-1}

\usepackage{amsmath, amsthm, amssymb,slashed}

\usepackage[usenames, dvipsnames]{color}
\usepackage[colorlinks,citecolor=blue,linkcolor=blue,urlcolor=blue]{hyperref}

\usepackage[normalem]{ulem}

\usepackage{booktabs}

\definecolor{mygray}{gray}{0.6}

\usepackage{upgreek}
\usepackage{bbm}

\newenvironment{myfont}[2][]{\csname#2\endcsname[#1]}{}

\usepackage{slashed}
\usepackage[makeroom]{cancel}
\usepackage[normalem]{ulem}
\usepackage{soul}
\newcommand{\stkout}[1]{\ifmmode\text{\sout{\ensuremath{#1}}}\else\sout{#1}\fi}

\usepackage{sseq}
\usepackage[all,cmtip]{xy}
\usepackage{tikz-cd}
\usepackage{tikz}
\usetikzlibrary{matrix}
\usetikzlibrary{decorations.markings}
\usetikzlibrary{tikzmark,decorations.pathreplacing,positioning}
\newcommand{\bea}{\begin{eqnarray}}
\newcommand{\eea}{\end{eqnarray}}
\def\be{\begin{equation}}
\def\ee{\end{equation}}

\newcommand{\e}{\hspace{1pt}\mathrm{e}}
\newcommand{\imth}{\hspace{1pt}\mathrm{i}\hspace{1pt}}
\newcommand{\im}{\hspace{1pt}\mathrm{i}\hspace{1pt}}
\newcommand{\ii}{\hspace{1pt}\mathrm{i}\hspace{1pt}}

\def\RP{{\mathbb{RP}}}
\def\CP{{\mathbb{CP}}}

\definecolor{red}{rgb}{1,0,0}
\definecolor{blue}{rgb}{0,0,1}
\definecolor{dblue}{rgb}{0,0,0.4}
\definecolor{green}{rgb}{0,1,0}
\definecolor{black}{rgb}{0,0,0}
\definecolor{white}{rgb}{1,1,1}

\definecolor{brn}{rgb}{.8,.4,.0}
\definecolor{redo}{rgb}{1,.5,.0}
\definecolor{ddgrn}{rgb}{0,0.4,0}
\definecolor{dgrn}{rgb}{0,0.55,0}
\definecolor{dbl}{rgb}{0,0,0.5}

\usepackage[bbgreekl]{mathbbol}
\usepackage{amscd}

\newcommand{\Z}{\mathbb{Z}}
\newcommand{\C}{\mathbb{C}}
\newcommand{\R}{\mathbb{R}}

\renewcommand{\t}[1]{\tilde{#1}} 
\newcommand{\dd}{\hspace{1pt}\mathrm{d}}

\newcommand{\Ref}[1]{Ref.~[\onlinecite{#1}]}
\newcommand{\Eq}[1]{(\ref{#1})} 
\newcommand{\eq}[1]{eq.~(\ref{#1})} 
\newcommand{\eqn}[1]{eqn.~(\ref{#1})}

\newcommand{\Tr}{{\rm Tr}}

\newcommand{\up}{\uparrow} 
\newcommand{\down}{\downarrow}

\newcommand{\bpm}{\begin{pmatrix}}
\newcommand{\epm}{\end{pmatrix}}
\newcommand{\bmm}{\begin{matrix}}
\newcommand{\emm}{\end{matrix}}

\newcommand{\cC}{ {\cal C} }

\newcommand{\cP}{ {\cal P} } 
\newcommand{\cR}{ {\cal R} } 
 
\newcommand{\cT}{ {\cal T} }

\newcommand{\al}{\alpha}

\def\Z{{\mathbb{Z}}}

\def\R{{\mathbb{R}}}
\def\C{{\mathbb{C}}}

\def\Tr{{\mathrm{Tr}}}

\def \Hom{\operatorname{Hom}}
\def \Tor{\operatorname{Tor}}

\def \H{\operatorname{H}}

\def \Z{\mathbb{Z}}

\def \RP{\mathbb{RP}}
\def \CP{\mathbb{CP}}

\newcommand\hcup[1]{\underset{{\scriptscriptstyle #1}}{\cup}}

\newcommand {\emptycomment}[1]{}

\def\TP{\mathrm{TP}}
\def\Sq{\mathrm{Sq}}
\def\B{\mathrm{B}}

\newcommand{\SO}{{\rm SO}}
\newcommand{\Spin}{{\rm Spin}}
\newcommand{\U}{{\rm U}}
\newcommand{\SU}{{\rm SU}}
\newcommand{\PSU}{{\rm PSU}}
\newcommand{\Pin}{{\rm Pin}}
\newcommand{\tO}{{\rm O}}
\newcommand{\tA}{{\rm A}}

\newcommand{\tE}{{\rm E}}
\newcommand{\rN}{{\rm N}}
\def \N{\mathrm{N}}

\def\bZ{{\mathbf{Z}}}
\usepackage{centernot}
\newcommand{\nn}{{\nonumber}}
\newcommand{\Sec}[1]{Sec.~\ref{#1}} 

\newcommand{\W}{{\rm W}}
\newcommand{\diag}{{\rm diag}}
\usepackage{enumitem} 
\usepackage{mathtools,amssymb,varwidth}

\newcommand{\Fig}[1]{Fig.~\ref{#1}}

\usepackage{datetime}

\newtheorem{theorem}{Theorem}[section]

\newtheorem{lemma}[theorem]{Lemma}

\begin{document}

\begin{titlepage}

\title{
New  Higher Anomalies,
SU(N) Yang-Mills Gauge
Theory and $\mathbb{CP}^{\mathrm{N}-1}$ Sigma Model
}

\author{Zheyan Wan {\includegraphics[height=3.85ex]{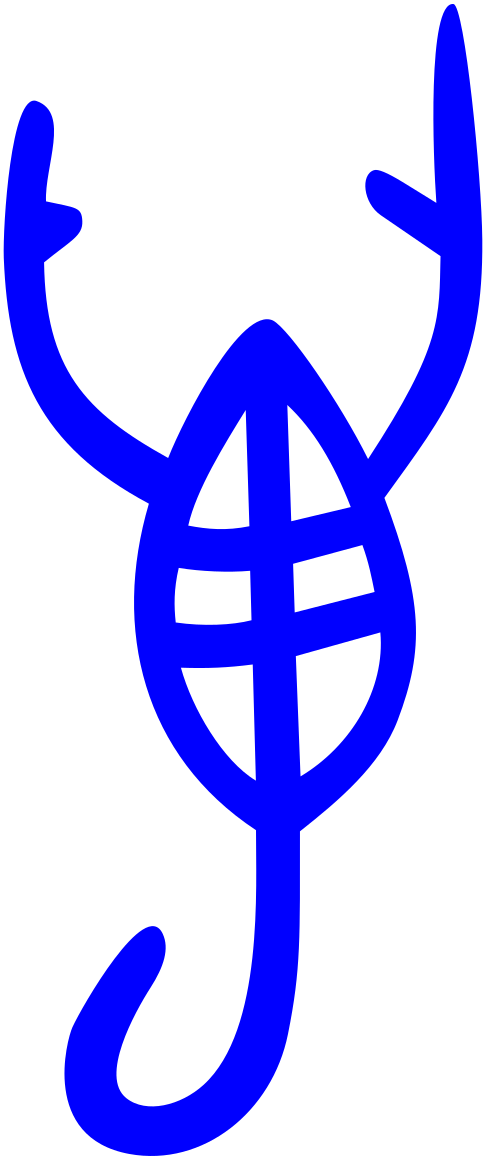}}
}
\email{wanzheyan@mail.tsinghua.edu.cn}
\affiliation{Yau Mathematical Sciences Center, Tsinghua University, Beijing 100084, China}
\affiliation{School of Mathematical Sciences, USTC, Hefei 230026, China}

\author{Juven Wang {\includegraphics[height=3.35ex]{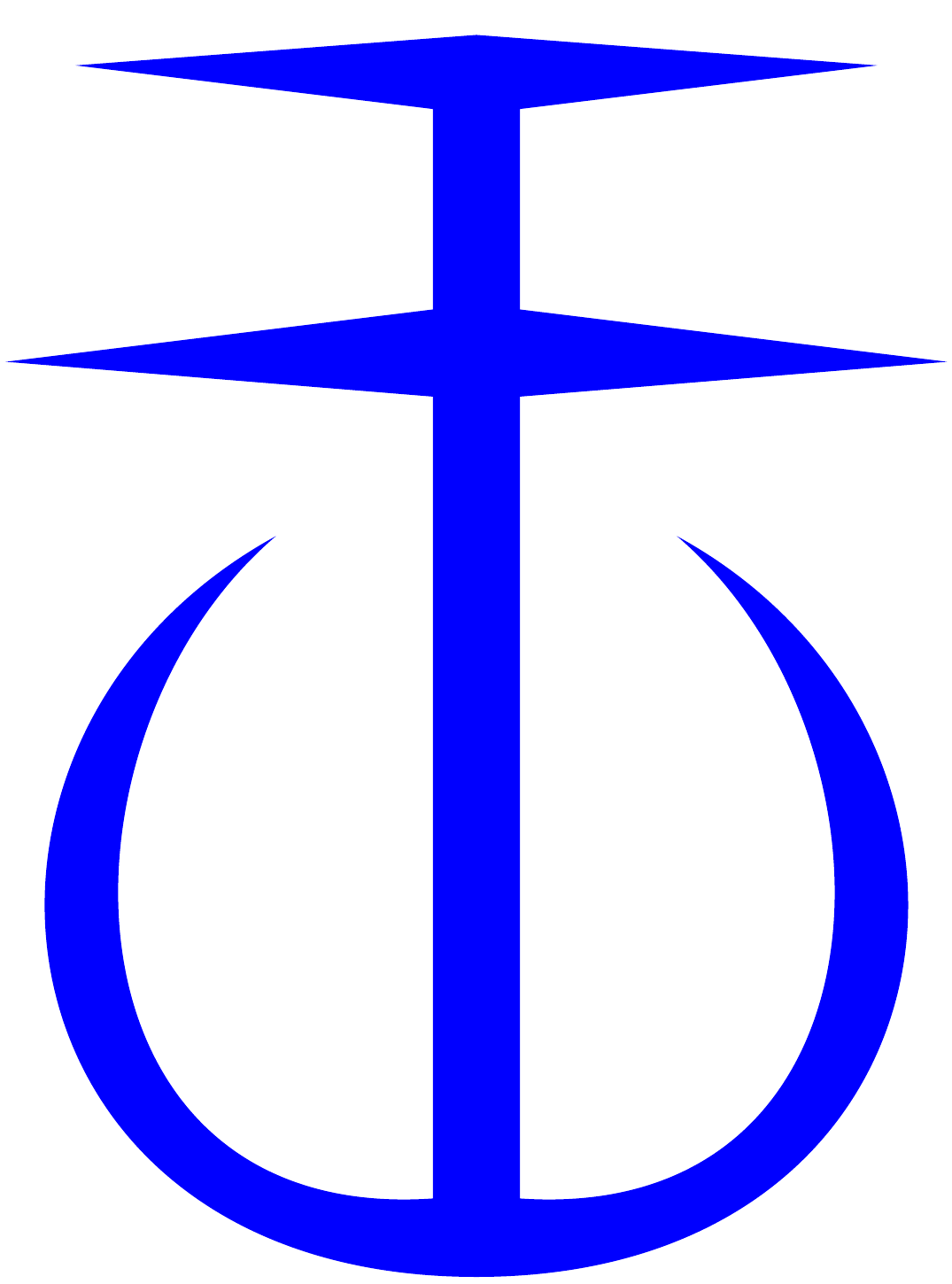}}
}
\email{jw@cmsa.fas.harvard.edu}
\affiliation{School of Natural Sciences, Institute for Advanced Study,  Einstein Drive, Princeton, NJ 08540, USA }
\affiliation{Center of Mathematical Sciences and Applications, Harvard University, MA 02138, USA}

\author{Yunqin Zheng {\includegraphics[height=3.55ex]{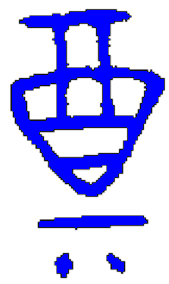}}
}
\email{yunqinz@princeton.edu}
\affiliation{Physics Department, Princeton University, Princeton, New Jersey 08544, USA}


\begin{abstract}

We hypothesize a new and more complete set 
of anomalies of certain
quantum field theories (QFTs) and then give an eclectic verification.  
First, 
 we propose a set of 't Hooft higher anomalies of 
4d time-reversal symmetric pure SU(N)-Yang-Mills (YM) gauge theory with a second-Chern-class topological term at $\theta=\pi$, 
via 5d cobordism invariants (higher symmetry-protected topological states), with N = $2, 3, 4$ and others.
Second,
we propose a set of 't Hooft anomalies of 
2d $\mathbb{CP}^{\mathrm{N}-1}$-sigma models with a first-Chern-class topological term at $\theta=\pi$, 
by enlisting all possible 
3d cobordism invariants and selecting the matched terms.
Based on algebraic/geometric topology, QFT analysis, manifold generator correspondence, condensed matter inputs 
such as stacking PSU(N)-generalized 
Haldane quantum spin chains,
and additional physics criteria, 
we derive a correspondence between 5d and 3d new invariants. Thus we broadly prove a potentially complete 
anomaly-matching
between 4d SU(N) YM and 2d $\mathbb{CP}^{\mathrm{N}-1}$ models at N = 2, and suggest new (but maybe incomplete) anomalies at N = 4.
We formulate a higher-symmetry analog of ``Lieb-Schultz-Mattis theorem''
to constrain the low-energy dynamics.

\end{abstract}


\maketitle

\end{titlepage}

\tableofcontents

\section{Introduction and Summary}
Determining the dynamics and phase structures of strongly coupled quantum field theories (QFTs) is a challenging but important problem. 
For example, one of {the} Millennium Problems is partly on showing the existence of quantum Yang-Mills (YM) gauge theory \cite{PhysRev.96.191-YM} and the mass gap.
The fate of a pure YM theory with an SU(N) gauge group (i.e. we simply denote it as an SU(N)-YM), without additional matter fields,
without topological term ($\theta=0$), is confined and trivially gapped in an Euclidean spacetime $\R^4$ \cite{YMMP-Jaffe-Witten}.
A powerful tool to constrain the dynamics of QFTs is based on non-perturbative methods such as the 't Hooft anomaly-matching \cite{tHooft:1980xss}.
Although anomaly-matching may not uniquely determine the quantum dynamics, it can rule out some impossible quantum phases with mismatched anomalies, 
thus guiding us to focus only on favorable anomaly-matched phases for low energy phase structures of QFTs.
The importance of dynamics and anomalies is not merely for a formal QFT side, 
but also on a more practical application to high-energy ultraviolet (UV) completion of QFTs, such as on a lattice regularization or  
condensed matter systems. (See, for instance [\onlinecite{2017arXiv171111587GPW}] and references therein, 
a recent application of the anomalies, topological terms and dynamical constraints of SU(N)-YM gauge theories on
UV-regulated condensed matter systems, obtained from dynamically gauging the SU($\rN$)-symmetric \emph{interacting} generalized 
topological superconductors/insulators \cite{2010RMP_HasanKane,2011_RMP_Qi_Zhang},
or more generally Symmetry-Protected Topological state (SPTs) \cite{Chen2011pg1106.4772, Senthil1405.4015, Wen1610.03911}).

In this work, we attempt to identify the potentially complete 't Hooft anomalies of 4d pure SU($\rN$)-YM gauge theory
with a ${\theta=\pi}$ topological term (a second-Chern-class topological term)
and 2d $\mathbb{CP}^{\rN-1}$-sigma model with a ${\theta=\pi}$ topological term (a first-Chern-class topological term) 
in an Euclidean spacetime. Here $d$d denotes a $d$-dimensional spacetime.
For the convenience of readers, our main result is summarized in \Fig{Fig-reduce-1} and \Fig{Fig-reduce-2}.  
By completing 't Hooft anomalies of QFTs, we need to first identify the relevant (if not all of) global symmetry $G$ of QFTs.
Then we couple the QFTs to classical background-symmetric gauge field of $G$, and try to detect the possible obstructions of such
coupling \cite{tHooft:1980xss}. Such obstructions, known as the obstruction of gauging the global symmetry, are named 
{`` 't Hooft anomalies.''}
In the literature, when people refer to {``anomalies,'' }however, they can means different things. 
To fix our terminology, we refer {``anomalies''} to one of the followings:
\begin{enumerate}
\item Classical global symmetry is violated at the quantum theory, such that the 
classical global symmetry fails to be a quantum global symmetry, e.g. the original Adler-Bell-Jackiw anomaly \cite{Adler1969gkABJ,Bell1969tsABJ}.

\item Quantum global symmetry is well-defined and preserved. 
(Global symmetry is sensible, not only at a classical theory [if there is any classical description], but also for a quantum theory.) 
However, there is an obstruction to gauge the global symmetry.
Specifically, we can detect a certain obstruction to even \emph{weakly gauge}
the symmetry or couple the symmetry to a \emph{non-dynamical background probed gauge field}.
(We will refer this as a background field, abbreviated as {``bgd.field.''})
This is known as `{`'t Hooft anomaly,''} or sometimes
regarded as a ``weakly gauged anomaly'' in condensed matter.
Namely, the partition function $\bZ$ does not sum over background gauge connections, but only fix a background gauge connection
and only depend on the background gauge connection as a classical field (as a classical coupling constant).

\item Quantum global symmetry is well-defined and preserved. 
However, once we promote the global symmetry to a gauge symmetry of the dynamical gauge theory,
then the gauge theory becomes ill-defined. Some people call this as 
a ``dynamical gauge anomaly'' which makes a quantum theory ill-defined.
Namely, the partition function $\bZ$ after summing over dynamical gauge connections becomes ill-defined.

\end{enumerate}

Now {``'t Hooft anomalies''}  (for simplicity, from now on, we may abbreviate them as {``anomalies''})
 have at least three intertwined interpretations:\\
 

\noindent
Interpretation (1): In condensed matter physics, ``'t Hooft anomalies''  are known as the obstruction to lattice-regularize the global symmetry's quantum operator in a local on-site manner at UV due to symmetry-twists. 
(See [\onlinecite{Wen2013oza1303.1803, 1405.7689, Wang2017locWWW1705.06728}] for QFT-oriented discussion and references therein.)
This ``non-onsite symmetry'' viewpoint is generically applicable to both, \emph{perturbative} anomalies, and  \emph{non-perturbative} anomalies:\\
$\bullet$ \emph{perturbative} anomalies --- Computable from perturbative Feynman diagram calculations.\\ 
$\bullet$ \emph{non-perturbative or global} anomalies --- 
Examples of global anomalies include the old and the new SU(2) anomalies \cite{Witten:1982fp, Wang:2018qoyWWW} 
(a caveat: here we mean their 't Hooft anomaly analogs if we view the SU(2) gauge field as a non-dynamical classical background, instead of dynamical field)
and the global gravitational anomalies \cite{Witten1985xe}.


The occurrence of these anomalies are 
sensitive to the underlying UV-completion not only of fermionic systems, but also of bosonic systems \cite{WangSantosWen1403.5256, 1405.7689, Kapustin1404.3230, JWangthesis}.
We call the anomalies of QFT whose UV-completion requires only the bosonic degrees of freedom as bosonic anomalies \cite{WangSantosWen1403.5256};
while those must require fermionic degrees of freedom as fermionic anomalies.\\

\noindent
Interpretation (2): In QFTs, the obstruction is on the impossibility of adding any counter term 
in its own dimension ($d$-d) in order to absorb a one-higher-dimensional counter term (e.g. $(d+1)$d topological term) due to background $G$-field \cite{Kapustin:2014gua}.
This is named the ``anomaly-inflow \cite{Callan:1984sa}.''
The  $(d+1)$d topological term is known as the $(d+1)$d SPTs in condensed matter physics \cite{Chen2011pg1106.4772, Senthil1405.4015}.\\

\noindent
Interpretation (3): In math, the $d$d anomalies can be systematically captured by $(d+1)$d topological invariants \cite{Witten:1982fp} known as 
cobordism invariants \cite{DaiFreed1994kq9405012,Kapustin2014tfa1403.1467,Kapustin1406.7329,Freed2016}.\\

There is a long history of relating these two particular 4d SU($\rN$)-YM and 2d $\mathbb{CP}^{\rN-1}$ theories, 
since the work of Atiyah \cite{atiyah1984}, Donaldson \cite{Donaldson1984tm} and others, in the interplay of QFTs in physics and mathematics.
Recently three key progresses shed new lights on their relations further:\\

\noindent
($i$) Higher symmetries and higher anomalies:
The familiar 0-form global symmetry has a \emph{charged} object of 0d measured by the \emph{charge} operator of $(d-1)$d.
The generalized $q$-form global symmetry, introduced by [\onlinecite{Gaiotto:2014kfa}], demands a \emph{charged} object of $q$d measured by the \emph{charge} operator of $(d-q-1)$d (i.e. codimension-$(q+1)$). This concept turns out to be powerful to
detect new anomalies, e.g. the pure SU($\rN$)-YM at $\theta=\pi$ (See \eq{eq:YM-pi})
has a mixed anomaly between 0-form time-reversal symmetry $\Z_2^T$
and 1-form center symmetry $\Z_{\rN,[1]}$ at an even integer $\rN$, firstly discovered in a remarkable work [\onlinecite{Gaiotto2017yupZoharTTT}].
We review this result in Sec.~\ref{sec:rev}, then we will introduce new anomalies (to our best understanding, these have not yet been identified in the previous literature)
in later sections (\Fig{Fig-reduce-1} and \Fig{Fig-reduce-2}).\\

\noindent
($ii$) Relate (higher)-SPTs to (higher)-topological invariants: 
Follow the condensed matter literature, based on the earlier discussion on the \emph{symmetry twist},
it has been recognized that the classical background-field partition function under the {symmetry twist}, called $\bZ_{{\text{sym.twist}}}$
in $(d+1)$d
can be regarded as the partition function of  $(d+1)$d SPTs $\bZ_{{\text{SPTs}}}$.
 These descriptions are applicable to both low-energy infrared (IR) field theory, but also to the UV-regulated SPTs on a lattice,
see [\onlinecite{Wen2013oza1303.1803, 1405.7689, Kapustin2014tfa1403.1467}]
and Refs.~therein for a systematic set-up. Schematically, we follow the framework of \cite{1405.7689}, 
\begin{multline}
\label{eq:Zpart}
\bZ^{\text{$(d+1)$d}}_{{\text{sym.twist}}}=\bZ^{\text{$(d+1)$d}}_{{\text{SPTs}}}=\bZ^{\text{$(d+1)$d}}_{{\text{topo.inv}}}=\bZ^{\text{$(d+1)$d}}_{{\text{Cobordism.inv}}}  \\
\longleftrightarrow \text{$d$d-(higher) 't Hooft anomaly}.\;\;\;\;\;
\end{multline}
 In general,
the partition function $\bZ_{{\text{sym.twist}}}=\bZ_{{\text{SPTs}}}[A_1,B_2, w_i, \dots]$
is a functional containing background gauge fields of 1-form $A_1$, 2-form $B_2$ or higher forms;
and can contain characteristic classes \cite{milnor1974characteristic} such as the $i$-th Stiefel-Whitney class ($w_i$)
and other geometric probes such as gravitational background fields, e.g. a gravitational Chern-Simons 3-form CS$_3(\Gamma)$ involving the Levi-Civita connection or
the spin connection $\Gamma$. 
For convention, we use the capital letters ($A,B,...$) to denote \emph{non-dynamical background gauge} fields (which, however, later they may or may not be dynamically gauged),  while 
the little letters ($a,b,...$) to denote \emph{dynamical gauge} fields.

More generally, \\
$\bullet$ For the ordinary 0-form symmetry, we can couple the charged 0d point operator to 1-form background gauge field (so the symmetry-twist occurs
in the Poincar\'e dual codimension-1 sub-spacetime [$d$d] of SPTs).\\
$\bullet$ For the 1-form symmetry, we can couple the charged 1d line operator to 2-form background gauge field (so the symmetry-twist occurs in the Poincar\'e dual codimension-2 sub-spacetime [$(d-1)$d] of SPTs).\\
$\bullet$ For the $q$-form symmetry, we can couple the \emph{charged} $q$d extended operator to $(q+1)$-form background gauge field. 
The  \emph{charged} $q$d extended operator can be measured by another \emph{charge} operator of codimension-$(q+1)$ [i.e. $(d-q)$d].
So the symmetry-twist can be interpreted as the occurrence of the codimension-$(q+1)$ \emph{charge} operator.
Namely, the symmetry-twist happens at a Poincar\'e dual codimension-$(q+1)$ sub-spacetime [$(d-q)$d] of SPTs. 
We can view the measurement of a \emph{charged} $q$d extended object, happening at any $q$-dimensional intersection 
between the $(q+1)$d form background gauge field
and the codimension-$(q+1)$ symmetry-twist or \emph{charge} operator of this SPT vacua.\\

For SPTs protected by higher symmetries (for generic $q$, especially for any SPTs with at least a symmetry of $q>0$), we refer them as {higher-SPTs}. 
So our principle above is applicable to higher-SPTs \cite{Thorngren:2015gtw1511.02929, Delcamp2018wlb1802.10104, W2}.
In the following of this article, thanks to \eq{eq:Zpart}, we can interchange the usages and interpretations of
``higher SPTs $\bZ_{{\text{SPTs}}}$,'' ``higher topological terms due to symmetry-twist $\bZ^{\text{$(d+1)$d}}_{{\text{sym.twist}}}$,'' ``higher topological invariants $\bZ^{\text{$(d+1)$d}}_{{\text{topo.inv}}}$'' or ``cobordism invariants 
$\bZ^{\text{$(d+1)$d}}_{{\text{Cobordism.inv}}} $'' in $(d+1)$d.
They are all physically equivalent, and can uniquely determine a $d$d higher anomaly, when we study the anomaly of any
boundary theory of the $(d+1)$d higher SPTs living on a manifold with $d$d boundary.
Thus, we regard all of them as physically tightly-related given by \eq{eq:Zpart}.
In short, by turning on the classical background probed field (denoted as ``{bgd.field}'' in \eq{eq:QFT-SPT}) coupled to $d$d QFT, under the symmetry transformation (i.e. symmetry twist),
its partition function $ \bZ^{\text{$d$d}}_{{\text{QFT}}}$
can be \emph{shifted} 
\begin{multline} \label{eq:QFT-SPT}
\left.
 \bZ^{\text{$d$d}}_{{\text{QFT}}}   \right|_{\text{bgd.field}=0} \\
 \longrightarrow
 \left. \bZ^{\text{$d$d}}_{{\text{QFT}}}
   \right|_{\text{bgd.field}\neq 0} \cdot
   \bZ^{\text{$(d+1)$d}}_{{\text{SPTs}}}(\text{bgd.field}) , 
\end{multline}
to detect the underlying {$(d+1)$d} topological terms/counter term/SPTs, namely the $(d+1)$d partition function  $\bZ^{\text{$(d+1)$d}}_{{\text{SPTs}}}$.
To check whether the underlying {$(d+1)$d} SPTs really specifies a true $d$d 't Hooft anomaly unremovable from $d$d \emph{counter term},
it means that $\bZ^{\text{$(d+1)$d}}_{{\text{SPTs}}}(\text{bgd.field})$ cannot be absorbed by a lower-dimensional SPTs 
$\bZ^{\text{$d$d}}_{{\text{SPTs}}}(\text{bgd.field})$, namely 
\begin{multline} 
 \left. \bZ^{\text{$d$d}}_{{\text{QFT}}}
   \right|_{\text{bgd.field}} \cdot
   \bZ^{\text{$(d+1)$d}}_{{\text{SPTs}}}(\text{bgd.field})\\
\neq 
  \left. \bZ^{\text{$d$d}}_{{\text{QFT}}}
   \right|_{\text{bgd.field}} \cdot
      \bZ^{\text{$d$d}}_{{\text{SPTs}}}(\text{bgd.field}).
\end{multline}

\noindent
($iii$) Dimensional reduction:
A very recent progress shows that a certain
anomaly of 4d SU($\rN$)-YM theory can be matched with another anomaly of 2d $\mathbb{CP}^{\rN-1}$ model
under a 2-torus $T^2$ reduction in [\onlinecite{Yamazaki:2017dra}],
built upon previous investigations \cite{Yamazaki:2017ulc,Tanizaki:2017qhf1710.08923}.
This development, together with the mathematical rigorous constraint from 4d and 2d instantons \cite{atiyah1984, Donaldson1984tm},
provides the evidence that the complete set of (higher) anomalies 
of 4d YM should be fully matched with 2d $\mathbb{CP}^{\rN-1}$ model
under a $T^2$  reduction.\footnote{ \label{footnote:FS}
{The complex projective space $\CP^{\rN-1}$
is obtained from the moduli space of flat connections of SU(N) YM theory. (See \cite{Yamazaki:2017ulc} and Fig.~\ref{fig:YM-reduce}.) 
This moduli space of flat connections do not have
a canonical Fubini-Study metric and may have singularities. However, this subtle issue, between
the $\CP^{\rN-1}$ target  and the moduli space of flat connections, only affects the geometry issue, and
should not affect the topological issue concerning 
non-perturbative global discrete anomalies that we focus on in this work.}
}
\\

In this work, we draw a wide range of knowledges, tools, comprehensions, and intuitions from:\\
$\bullet$ Condensed matter physics and lattice regularizations. Simplicial-complex regularized triangulable manifolds and smooth manifolds. This approach is related to our earlier Interpretation (1),
and the progress (ii). \\
$\bullet$ QFT (continuum) methods: Path integral, higher symmetries associated with extended operators, etc.
This is related to our earlier Interpretation (2), and the progress (i), (ii) and (iii). \\
$\bullet$ Mathematics:  Algebraic topology methods include cobordism, cohomology and group cohomology theory. 
Geometric topology methods include the embedding of manifolds, and Poincar\'e duality, etc.
This is related to our earlier Interpretation (3), and the progress (ii) and (iii).\\
Built upon previous results, 
we are able to derive a consistent story, which identifies, previously missing, thus, new higher anomalies in YM theory and in $\mathbb{CP}^{\rN-1}$ model.
A sublimed version of our result may count  as an eclectic proof between the anomaly-matching between two theories under a 2-torus $T^2$ reduction from the 4d theory reduced to a 2d theory.

{Earlier we stated that our aim is to provide  potentially complete 't Hooft anomalies of 4d pure SU($\rN$)-YM gauge theory with a $\theta=\pi$ topological term
and 2d $\mathbb{CP}^{\rN-1}$-sigma model with a $\theta=\pi$ topological term. 
It turns out that our recent work \Ref{Wan2019oyr1904.00994}}
suggests there are indeed \emph{different} versions of 4d pure SU($\rN$)-YM gauge theory with a $\theta=\pi$ topological term.
What happened is that \Ref{Wan2019oyr1904.00994} founds the 
\emph{different} versions of YM theories
can be characterized at least partially by the
\emph{different}  quantum numbers associated with the extended operator (Wilson line) of SU(N) YM.  
In simple words, 
Wilson line of YM can have:
\begin{enumerate}
\item time-reversal $\cT$ quantum number, say labeled by $K_1 \in \{0,1\}= \Z_2$ \cite{Wan2019oyr1904.00994}, under
$\cT$-symmetry tranformation, as:
\begin{itemize}
\item Kramers singlet ($\cT^2=+1$)
\item Kramers doublet ($\cT^2=-1$)
\end{itemize}
\item spin-statistics quantum number, say labeled by $K_2 \in \{0,1\}= \Z_2$  \cite{Wan2019oyr1904.00994}, as:
\begin{itemize}
\item bosonic (integer spin-statistics)
\item fermionic (half-integer spin-statistics).
\end{itemize}
\end{enumerate}
More physically intuitively, 
imagine in the {ultraviolet} lattice cut-off energy scale, 
the closed Wilson line can be opened up as an
open Wilson line with two open ends.
Such that each open end can host very highly-energetically massive
0D particle. This 0D particle can be
Kramers singlet $\cT^2=+1$ or 
Kramers doublet $\cT^2=-1$
under time-reversal. 
Under self-spinning by $2\pi$,
this 0D particle can also be
{bosonic (getting a +1 sign) or
fermionic (getting a $-1$ sign).}
\Ref{Wan2019oyr1904.00994} focuses on SU(2) YM and gives mathematical interpretations of the $K_1, K_2$ term,
based on the gauge bundle constraint,
\be
  \hspace*{-.8cm} 
w_2(V_{\PSU(2)})=w_2(V_{\SO(3)})=B+K_1w_1(TM^5)^2+K_2w_2(TM^5).
\ee
Thus $K_1$ and $K_2$ are the choices of the gauge bundle constraint, with $K_1, K_2 \in\Z_2$.
The $w_j(TM)$ is $j$-th Stiefel-Whitney (SW) classes of tangent bundle $TM$.
 \Ref{Wan2019oyr1904.00994} shows that putting \emph{different} siblings of 
 4d YM on unorientable manifolds and turning on background $B$ fields,
 give us the access to \emph{different} versions of 't Hooft anomalies.
 \Ref{Wan2019oyr1904.00994} suggests the Wilson line quantum numbers 
 are related to the $(K_1, K_2)$ via:
 \bea
 \left\{\begin{array}{cc}
(K_1, K_2)=(0,0), & \text{Kramers singlet bosonic,}\\
(K_1, K_2)=(1,0), & \text{Kramers doublet bosonic,}\\
(K_1, K_2)=(0,1), & \text{Kramers doublet fermionic,}\\
(K_1, K_2)=(1,1), & \text{Kramers singlet fermionic.}
\end{array}\right.
 \eea  

In this article, we \emph{do not} use the approach of {\Ref{Wan2019oyr1904.00994} }. Instead,
we like to relate \emph{different} {versions} (four siblings of \Ref{Wan2019oyr1904.00994}) of 4d YM simply based on possible 4d 't Hooft anomalies
(5d topological terms) satisfying physical constraints (given in \Sec{sec:rule}). 
Amusingly, we can find out at least two versions of YM with two different 't Hooft anomalies.
We also relate \emph{different} versions of 4d SU(N) YM
to  \emph{different} versions of 2d $\mathbb{CP}^{\rN-1}$-sigma model with a $\theta=\pi$ topological term,
via a 2-torus dimensional reduction. 
{We also consider a slight generalization of the above gauge-bundle constraint when $\rN >2$, e.g. see \Sec{N=4reduction}.
The details of a further generalized gauge-bundle constraint
including the \emph{charge conjugation} $C$ quantum number for \emph{different} siblings of YM theories is reported in an upcoming future work \cite{WWZ2019-2}.
}

{The outline of our article goes as follows.} 

In Sec.~\ref{sec:remark}, we comment and review on QFTs (relevant to YM theory and $\mathbb{CP}^{\rN-1}$ model), 
their global symmetries, anomalies and topological invariants.
This section can serve as an invitation for condensed matter colleagues, while we also review the relevant new concepts and notations to
high energy/QFT theorists  and mathematicians. 

In Sec.~\ref{sec:cobo-top}, we provide the concrete explicit results on the {cobordisms, SPTs/topological terms, and manifold generators.
This is relevant to our classification of all possible higher 't Hooft anomalies.} Also it is relevant to our later eclectic proof on the anomalies of
YM theory and $\mathbb{CP}^{\rN-1}$ model.

In Sec.~\ref{sec:rev}, we review the known anomalies in 4d YM theory and 2d $\mathbb{CP}^{\rN-1}$ model, and explain their physical meanings, or re-derive them,
in terms of mathematically precise cobordism invariants.

In \Sec{sec:rule}, \Sec{sec:5d3d} and \Fig{fig:YM-reduce}, we should cautiously remark that how 4d SU($\rN$)-YM theory is related to 2d $\mathbb{CP}^{\rN-1}$ model.

In \Sec{sec:rule}, in particular, we give our rules to constrain the anomalies for 4d YM theory and 2d $\mathbb{CP}^{\rN-1}$ model,
and for 5d and 3d invariants.

In \Sec{sec:5d3d}, we present mathematical formulations of dimensional reduction, from
5d to 3d of cobordism/SPTs/topological term, and from 4d to 2d of anomaly reduction.

In \Sec{sec:newSUn}, we present new higher anomalies for 4d SU(N) YM theory.

In \Sec{sec:newCPn}, we present new anomalies for 2d $\mathbb{CP}^{\rN-1}$ model.

In \Sec{sec:sTQFT}, with the list of potentially complete 't Hooft anomalies of the above 
4d SU($\rN$)-YM and 2d $\CP^{\rN-1}$-model at ${\theta=\pi}$,
we constrain their low-energy dynamics further, based on the anomaly-matching.
We discuss the higher-symmetry analog Lieb-Schultz-Mattis theorem.
In particular, we check whether the  't Hooft anomalies of the above 
4d SU($\rN$)-YM and 2d $\CP^{\rN-1}$-model can be saturated by
a symmetry-extended TQFT of their own dimensions, 
by the (higher-)symmetry-extension method generalized from the method of Ref.~\cite{Wang2017locWWW1705.06728}. (See also our companion work \cite{Wan2018djlW2.1812.11955})
We also discuss their dynamical fates which become spontaneously symmetry-breaking (SSB) phases.

In \Sec{sec:tables}, we summarize our main results of
a more complete set of 't Hooft anomalies of the 4d SU($\rN$)-YM and 2d $\CP^{\rN-1}$-model and their dimensional reduction
in \Sec{sec:tables-N=2} for N = 2 and \Sec{sec:tables-N=4} for N = 4. 

We conclude in \Sec{sec:con}.


\section{Comments on QFTs: Global Symmetries and Topological Invariants}
\label{sec:remark}


\subsection{4d Yang-Mills Gauge Theory}  

{Now we consider a 4d pure SU($\rN$)-Yang-Mills gauge theory with $\theta$-term, with a positive integer $\rN \geq 2$, for a Euclidean partition function (such as an $\R^4$ spacetime)}
The path integral (or partition function) $\bZ^{\text{$4$d}}_{{\text{YM}}}$ is formally written as
\begin{widetext}
\begin{multline}
 \label{eq:YM-pi}
\bZ^{\text{$4$d}}_{{\text{YM}}}
\equiv\int [{\cal D} {a}] \exp\big( - S_{\text{YM}+\theta}[a] \big)\equiv\\
\int [{\cal D} {a}] \exp\big( - S_{\text{YM}}[a] \big)\exp\big( - S_{\theta}[a] \big)
\equiv
\int [{\cal D} {a}] \exp\big(   \Big(- \int\limits_{M^4} (\frac{1}{g^2}\text{Tr}\,F_a\wedge \star F_a)
+  \int\limits_{M^4} (\frac{\ii  \theta}{8 \pi^2}  \text{Tr}\,F_a\wedge F_a) \Big)\big).
\end{multline}
\end{widetext}
%
$\bullet$ $a$ is the 1-form SU($\rN$)-gauge field connection obtained from parallel transporting the principal-SU($\rN$) bundle over the spacetime manifold $M^4$.
The $a = a_\mu \dd x^\mu= a_\mu^\alpha T^\alpha \dd x^\mu$;
here $T^\alpha$ is the generator of Lie algebra {\bf{g}} for the gauge group (SU($\rN$)), {with the commutator $[T^\alpha,T^\beta]=\ii f^{\alpha \beta \gamma} T^\gamma$, 
where $f^{\alpha \beta \gamma}$ is a fully anti-symmetric structure constant.}
Locally $\dd x^\mu$ is a differential 1-form, the $\mu$ runs through the indices of coordinate of $M^4$.
Then $a_\mu=a_\mu^\alpha T^\alpha$ is the Lie algebra valued gauge field, which is in the adjoint representation of the Lie algebra.
(In physics, $a_\mu$ is the gluon vector field for quantum chromodynamics.)
The $[{\cal D} {a}]$ is the path integral measure, for a certain configuration of the gauge field $a$.
 All allowed gauge inequivalent configurations
 are integrated over within the path integral measures $\int [{\cal D} {a}]$, where gauge redundancy is removed or mod out.
The integration is under a weight factor $\exp\big( \ii S_{\text{YM}+\theta}[a] \big)$.\\
$\bullet$ The $F_a= \dd a-\ii a \wedge a $ is the $\SU(\rN)$ field strength, while $\dd$ is the exterior derivative and $\wedge$ is the wedge product; 
the $\star F_a$ is $F_a$'s Hodge dual. 
The $g$ is YM coupling constant. 
 \\
$\bullet$ The $\text{Tr}\,(F_a\wedge \star F_a)$ is the Yang-Mills Lagrangian \cite{PhysRev.96.191-YM} (a non-abelian generalization of Maxwell Lagrangian of U(1) gauge theory).
The Tr denotes the trace as an invariant quadratic form of the Lie algebra of gauge group (here SU($\rN$)).
Note that $\Tr[F_a]=\Tr[ \dd a-\ii a \wedge a]=0$ is traceless for a $\SU(\rN)$ field strength.
Under the variational principle, YM theory's classical equation of motion (EOM), in contrast to the linearity of U(1) Maxwell theory, is non-linear.\\
$\bullet$  The $(\frac{  \theta}{8 \pi^2}  \text{Tr}\,F_a\wedge F_a)$ term is named the $\theta$-topological term, which does not contribute to the classical EOM.\\
$\bullet$  This path integral is physically sensibly well-defined, but not precisely mathematically well-defined, because the gauge field can be freely chosen due to the gauge freedom. 
This problem occurs already for quantum U(1) Maxwell theory, but now becomes more troublesome due to the YM's non-abelian gauge group.
One way to deal with the path integral and the quantization is the method by Faddeev-Popov \cite{Faddeev1967fc} and De Witt \cite{DeWitt1967ub}.
However,
in this work, we actually do not need to worry about of the subtlety of the gauging fixing and the details of the running coupling $g$ for the full \emph{quantum} theory part of this 
path integral.
The reason is that we only aim to capture the 5d \emph{classical} background field partition function
$\bZ^{\text{$(d+1)$d}}_{{\text{sym.twist}}}=\bZ^{\text{$(d+1)$d}}_{{\text{SPTs}}}$ in \eq{eq:Zpart} 
that 4d YM theory must couple with in order to match the 't Hooft anomaly.
Schematically, by coupling YM to background field, under the symmetry transformation, we expect that
\be \label{eq:SPT-YM}
\left.
 \bZ^{\text{$4$d}}_{{\text{YM}}}   \right|_{\text{bgd.field}=0} \to \bZ^{\text{$5$d}}_{{\text{SPTs}}}(\text{bgd.field}) \cdot 
 \left. \bZ^{\text{$4$d}}_{{\text{YM}}}
   \right|_{\text{bgd.field}\neq 0}. 
\ee
For example, when {the} bgd.field is $B$,
\bea
 \bZ^{\text{$4$d}}_{{\text{YM}}} (B=0) \to \bZ^{\text{$5$d}}_{{\text{SPTs}}}(B \neq 0)  \cdot 
  \bZ^{\text{$4$d}}_{{\text{YM}}}(B \neq 0).   
\eea
Our goal will be identifying the
5d topological term (5d SPTs) \eq{eq:SPT-YM} under coupling to background fields.
We will focus on the Euclidean path integral of \eq{eq:YM-pi}.

\subsection{{SU($\rN$)-YM theory: Mixed higher-anomalies}}  
\label{sec:SUN-YM-mix-higher}

Below we warm up by re-deriving the result on the {mixed higher-anomaly of time-reversal $\Z_2^T$ and 1-form center $\Z_\rN$-symmetry of SU($\rN$)-YM theory},
firstly obtained in [\onlinecite{Gaiotto2017yupZoharTTT}], from scratch. 
Our derivation will be as self-contained as possible, meanwhile we introduce useful notations.

\subsubsection{Global symmetry and preliminary}
For 4d SU($\rN$)-Yang-Mills (YM) theory at $\theta=0$ and $\pi$ mod $2 \pi$, on an Euclidean $\R^4$ spacetime,  
we can identify its global symmetries: the 0-form time-reversal $\Z_2^T$ symmetry with time reversal $\cT$ (see more details in \Sec{sec:TR}), 
and 0-form charge conjugation $\Z_2^C$ with symmetry transformation $\cC$ (see more details in \Sec{sec:C-P-R}). 
Since the parity $\cP$ is guaranteed to be a symmetry due to the $\cC \cP \cT$ theorem (see more details in \Sec{sec:C-P-R}, or a version for Euclidean \cite{Witten2016cio1605.02391}),
we can denote the full 0-form symmetry as $G_{[0]}=\Z_2^T \times \Z_2^C$.
We also have the 1-form electric $G_{[1]}=\Z_{\rN,[1]}^e$ center symmetry \cite{Gaiotto:2014kfa}.

{So we find that the full global symmetry group ``schematically'' as}
\bea \label{eq:YM-Sym}
G=   \Z_{\rN,[1]}^e \rtimes (\Z_2^T \times \Z_2^C),
\eea
here {we focus on the discrete part, but we intentionally omit the continuous part (shown later in \Sec{sec:cobo-top}) 
of the spacetime symmetry group.}\footnote{One may wonder the role of 
parity $\cP$ (details in \Sec{sec:C-P-R}), and a potential larger symmetry group $(\Z_2^T \times \Z_2^C \times \Z_2^P)$ for $G_{[0]}$.
As we know that $\cC \cP \cT$ transformation is almost {trivial and tightly related} to the spacetime symmetry group.
It is at most a complex conjugation and anti-unitary operation in Minkowski signature. 
It is trivial in the Euclidean signature.
It will become clear later, 
when we show our cobordism calculation for the full global symmetry group including the spacetime and internal symmetries in \Sec{sec:cobo-top}.
More discussions on discrete symmetries in various YM gauge theories can be found 
in \cite{2017arXiv171111587GPW}.}

For N = 2, we actually have the semi-direct product ``$\rtimes$'' reduced to a direct product ``$\times$,'' so we write
\bea \label{eq:YM-SU2-Sym-T}
G=   \Z_{2,[1]}^e \times \Z_2^T,
\eea
{here we also do not have the $\Z_2^C$ charge conjugation global symmetry, due to that now becomes part of the SU(2) gauge group of YM theory.}
The non-commutative nature (the semi-direct product ``$\rtimes$'') of \eq{eq:YM-Sym} between 0-form and 1-form symmetries will
 be explained in the end of \Sec{sec:C-P-R}, after we first derive some preliminary knowledge below:\\

\noindent
$\bullet$ The 0-form $\Z_2^T$ symmetry can be probed by ``background symmetry-twist'' if placing the system on non-orientable manifolds.
The details of time-reversal symmetry transformation will be discussed in Sec.~\ref{sec:TR}.
\\
$\bullet$ The 1-form electric $\Z_{\rN,[1]}^e$-center symmetry (or simply 1-form $\Z_{\rN}$-symmetry) can be coupled to 2-form background field $B_2$.
The \emph{charged} object of the 1-form $\Z_{\rN,[1]}^e$-symmetry is the gauge-invariant Wilson line 
\be
W_e=\Tr_{\text{R}}( \text{P} \exp(\ii \oint a)).
\ee

The Wilson line $W_e$ has the $a$ viewed as a connection over a principal Lie group bundle (here SU(N)), 
which is parallel transported around the integrated closed loop resulting an element of the Lie group.
P is the path ordering.
The Tr is again the trace in the Lie algebra valued, over the irreducible representation R of the Lie group (here SU(N)).
The spectrum of Wilson line $W_e$ includes all representations of the given Lie group (here SU(N)). 
Specifying the local Lie algebra {\bf{g}} is not enough, we need to also specify the gauge Lie group (here SU(N)) and other data, such
as the set of extended operators and the topological terms,
in order to learn the global structure and non-perturbative physics of gauge theory (See \cite{AharonyASY2013hda1305.0318}, and \cite{2017arXiv171111587GPW} for many examples).

For the SU(N) gauge theory we concern, the spectrum of purely electric Wilson line
$W_e$ includes the fundamental representation with a $\Z_{\rN}$ class, 
which can be regarded as the $\Z_{\rN}$ charge label of 1-form $\Z_{\rN,[1]}^e$-symmetry. 

The 2-surface \emph{charge} operator that measures the 1-form $\Z_{\rN,[1]}^e$-symmetry of the  \emph{charged}  Wilson line
is the electric 2-surface operator that we denoted as $U_e$.
The higher $q$-form symmetry ($q>0$) needs to be abelian \cite{Gaiotto:2014kfa}, thus the
1-from electric symmetry is associated with the $\Z_{\rN}$ center subgroup part of SU($\rN$),
known as the 1-form $\Z_{\rN,[1]}^e$-symmetry.

If we place the Wilson line along the $S^1$ circle of the time or thermal circle,
it is known as the Polyakov loop, which nonzero expectation value (i.e. breaking of the 0-form center dimensionally-reduced 
from 1-form center symmetry) serves as the \emph{order parameter} of confinement-deconfinement transition.

Below we illuminate our understanding in details for the SU(2) YM  theory (so we set N=2), which the discussion can be generalized to SU(N) YM.

\begin{enumerate}
\item We write the SU(2)-YM theory with a background $B$ field (more precisely the $B_2$ 2-cochain field) coupling to 1-form $\Z_{\rN,[1]}^e$-symmetry as: 
\bea
 \label{eq:SU2YM}
&&\bZ^{\text{$4$d}}_{{{\SU(2)}\text{YM}}}[B]\\
&&= \int [D \lambda] \; \bZ^{\text{$4$d}}_{{\text{YM}}}
  \exp\big( \ii \pi \int \lambda \cup (w_2(E) - B_2) + \ii \frac{\pi}{2} p \mathcal{P}_2(B_2)\big), \nn
\eea
where $w_2(E)$ is the Stiefel-Whitney (SW) class of gauge bundle $E$, and $B_2$ is 
a $\Z_2$-valued 2-cochain, both are non-dynamical probes.
We see that integrating out $\lambda$, set $(w_2(E) - B_2)=0 \mod 2$, thus $B_2=w_2(E)$ is related.
For $B_2=0$, there is no symmetry twist $w_2(E)=0$.\\
For $B_2=w_2(E) \neq 0$, there is a twisted bundle or a so called symmetry twist. So
we have an additional $ \ii \frac{\pi}{2} p \mathcal{P}_2(w_2(E))$ depending on $p  \in \Z_4$.
The Pontryagin square term
$\mathcal{P}_2:\H^2(-,\Z_{2^k})\to\H^4(-,\Z_{2^{k+1}})$, here is given by
\be \label{eq:PB}
\mathcal{P}_2(B_2)=B_2\cup B_2+B_2\hcup{1}\delta B_2 
=B_2\cup B_2+B_2\hcup{1} 2 \Sq^1 B_2,
\ee
see Sec.~\ref{sec:theta}. With $\cup$ is a normal cup product
and $\hcup{1}$ is a higher cup product. For readers who are not familiar with
the mathematical details,  see the introduction to mathematical background in \cite{W2}.
The physical interpretation of adding  $\frac{\pi}{2} p \mathcal{P}_2(B_2)$ with $p \in \Z_4$,
is related to the fact of the YM vacua can be shifting by a higher-SPTs protected by 1-form symmetry, see Sec.~\ref{sec:theta}.

\item The electric Wilson line $W_e$ in the fundamental representation is dynamical and a genuine line operator.
Wilson line $W_e$ can live on the boundary of a magnetic 2-surface $U_m=\exp(\ii  {{ \pi \int}} w_2(E))=\exp(\ii  {{ \pi \int}} B_2)$. However, we can set $B_2=0$ since it is a probed field.
So $W_e$ is a genuine line operator, i.e. without the need to be at the boundary of 2-surface \cite{Gaiotto:2014kfa}.

\item The magnetic 't Hooft line $T_m$ is on the boundary of an electric 2-surface $U_e=\exp(\ii  { \pi \int} \lambda)$. Since $\lambda$ is dynamical,
't Hooft line is not genuine thus not in the line spectrum.

\item The electric 2-surface $U_e=\exp(\ii { \pi \int} \lambda)$ measures 1-form $e$-symmetry, and it is dynamical.
This can be seen from the fact that the 2-surface $w_2(E)$ is defined as a 2-surface defect 
({where each small 1-loop of 't Hooft line 
linked with this $w_2(E)$ getting a nontrivial $\pi$-phase $e^{\ii \pi}$}). 
The $w_2(E)$ has its boundary with Wilson loop $W_e$, such that
$U_e  U_m \sim \exp(\ii \pi {\int} \lambda \cup w_2(E) )$
specifies that when a 2-surface $\lambda$ links with (i.e. wraps around) a 
1-Wilson loop $W_e$, 
there is a nontrivial statistical $\pi$-phase $e^{\ii \pi}=-1$.
{This type of a link of 2-surface and 1-loop in a 4d spacetime is widely known as the generalized Aharonov-Bohm type of linking,
captured by a topological link invariant,
see e.g. \cite{1602.05951, Putrov2016qdo1612.09298} and references therein.}

\end{enumerate}

\subsubsection{YM theory coupled to background fields} 

First we {make a 2-form $\Z_\rN$ field out of 2-form and 1-form $\U(1)$ fields.}
The 1-form global symmetry $G_{[1]}$ can be coupled to a 2-form background ${\Z_{\rN} }$-gauge field $B_2$. 
In the continuum field theory, consider firstly a 2-form $\U(1)$-gauge field $B_2$ and 1-form $\U(1)$-gauge field $C_1$ such that
\bea
&& B_2 \text{ as a 2-form $\U(1)$ gauge field}, \\
&& C_1 \text{ as a 1-form $\U(1)$ gauge field}, \\
&& \rN B_2=\dd C_1, \text{ $B_2$ as a 2-form $\Z_\rN$ gauge field}.
\eea
that satisfactorily makes the continuum formulation of $B_2$ field as a 2-form $\Z_\rN$-gauge field when we constrain an enclosed surface integral
\bea
 \oint B_2=\frac{1}{\rN}\oint \dd C_1 \in \frac{1}{\rN} 2\pi \Z.
\eea


Now based on the relation {${\PSU(\rN)}=\frac{\SU(\rN)}{\Z_\rN}=\frac{\U(\rN)}{\U(1)}$},
we aim to have an $\SU(\rN)$ gauge theory coupled to a background 2-form ${\Z_\rN}$ field. Here
\bea
&&a \text{ as an SU(N) 1-form gauge field}, \nn\\
&&F_a= \dd a-\ii a \wedge a, \text{ as an SU(N) field strength},\nn\\
&&\Tr[F_a]=\Tr[ \dd a-\ii a \wedge a]=0  \text{ traceless for SU(N)}.\;\;
\eea
We then promote the U(N) gauge theory with 1-form U(N) gauge field $a'$, such that its normal subgroup U(1) is coupled to the background 1-form probed field $C_1$.
Here we can identify the U(N) gauge field to a combination of SU(N) and U(1) gauge fields
via,  {up to details of gauge transformations} \cite{Gaiotto2017yupZoharTTT},
\bea
&&a' \text{ as an U(N) 1-form gauge field}, \nn \\
&& a' \simeq a+ \mathbb{I}\frac{1}{\rN}C_1, \nn \\ 
&& \Tr a' \simeq \Tr a+ C_1 =C_1,  
\eea
\bea
&&F_{a'}= \dd a'-\ii a' \wedge a', \text{ as a U(N) field strength}, \nn \\
&&\Tr[F_{a'}]=\Tr[ \dd a'-\ii a' \wedge a']=\Tr[ \dd a']  =  \dd C_1,  \nn \\
&& \text{its trace is a $\U(1)$ field strength.}  
\eea
To associate the U(1) field strength $\Tr[F_{a'}]=\Tr[ \dd a']=\dd C_1$ to the background U(1) field strength, we can impose a Lagrange multiplier 2-form $u$,
\bea
&&\int [{\cal D} {u}] \exp\big( \ii  \int\limits_{M^4} \frac{1}{2 \pi} u \wedge (\Tr F_{a'} -\dd C_1)\big)\nn\\
&&=\int [{\cal D} {u}] \exp\big( \ii  \int\limits_{M^4} \frac{1}{2 \pi} u \wedge \dd(\Tr {a'} - C_1)\big).
\eea
We also have $\rN B_2=\dd C_1$, so we can impose another  Lagrange multiplier 2-form $u'$,
\bea
\int [{\cal D} {u'}] \exp\big( \ii  \int\limits_{M^4} \frac{1}{2 \pi} u' \wedge (\rN B_2 -\dd C_1)\big)
\eea
From now we will make the YM kinetic term {\emph{implicit}}, we focus on the $\theta$-topological term associated with the symmetry transformation.
The YM kinetic term does not contribute to the anomaly (in QFT language) and is not affected under the symmetry twist (in condensed matter language \cite{1405.7689}).
Overall, with only a pair $(B_2,C_1)$ as background fields (or sometimes simply written as $(B,C)$), we have, 
\begin{widetext}
\begin{multline} 
\int [{\cal D} {a}] [{\cal D} {u}]  [{\cal D} {u'}] \exp\big(  \ii 
  \int\limits_{M^4} (\frac{  \theta}{8 \pi^2}  \text{Tr}\,F_a\wedge F_a) +
 \ii  \int\limits_{M^4} \frac{1}{2 \pi} u \wedge \dd(\Tr {a'} - C_1) + \ii  \int\limits_{M^4} \frac{1}{2 \pi} u' \wedge (\rN B_2 -\dd C_1)\big)\;\;\\
=\int [{\cal D} {a}] [{\cal D} {u}]   \exp\big(  \ii 
  \int\limits_{M^4} (\frac{  \theta}{8 \pi^2}  \text{Tr}\,F_a\wedge F_a) \big)
\exp\big( \ii  \int\limits_{M^4} \frac{1}{2 \pi} u \wedge \dd(\Tr {a'} - C_1)\big) \vert_{\rN B_2 =\dd C_1}\\
=\int [{\cal D} {a}]    \exp\big(  \ii 
  \int\limits_{M^4} (\frac{  \theta}{8 \pi^2}  \text{Tr}\,F_a\wedge F_a) \big)
\vert_{ \Tr(F_{a'})= \Tr {\dd a'} =\dd C_1=\rN B_2 = B_2 (\Tr \; \mathbb{I})},
\end{multline}
\end{widetext}
here $\mathbb{I}$ is a rank-$\rN$ identity matrix, thus $(\Tr \;\mathbb{I})=\rN$.

Next we rewrite the above path integral in terms of $\U(\rN)$ gauge field, again 
{up to details of gauge transformations} \cite{Gaiotto2017yupZoharTTT},
\bea
&& a' \simeq a+ \mathbb{I}\frac{1}{\rN}C_1,\nn \\
&&F_{a'}= \dd a' -\ii a' \wedge a' \nn\\
&&= (\dd a+ \mathbb{I}\frac{1}{\rN} \dd C_1) -\ii (a+ \mathbb{I}\frac{1}{\rN}C_1) \wedge (a+ \mathbb{I}\frac{1}{\rN}C_1) \nn\\
&&=(\dd a+ B_2 \mathbb{I}) - \ii a  \wedge a + 0 =F_{a} +B_2 \mathbb{I}
\eea
Now, to fill in the details of gauge transformations,
\bea
&&B_2 \to B_2+ \dd \lambda, \\
&&C_1 \to C_1+ \dd \eta+ \rN \lambda,
\eea
\bea
&& a' \to a' - \lambda \mathbb{I} + \dd \eta_a,\\
&& a \to a + \dd \eta_a,
\eea
The infinitesimal and finite gauge transformations are:
\bea
&& a_{\mu}^\alpha \to a_{\mu}^\alpha +  \partial_{\mu} \eta_a^\alpha + f^{\alpha \beta \gamma}  a_{\mu}^\beta \eta_a^\gamma,\\
&& a \to V (a + {\ii} \dd) V^\dagger \equiv  e^{\ii \eta_a^\alpha T^\alpha} (a + \ii \dd)  e^{-\ii \eta_a^\alpha T^\alpha}, 
\eea
where we denote 1-form $\lambda$ and 0-form $\eta, \eta_a$ for gauge transformation parameters.
Here $\eta_a$ with {subindex $a$ is merely an internal label for the gauge field $a$'s transformation $\eta_a$}.
Here ${\alpha \beta \gamma}$ are the color indices in physics, and also the indices for the adjoint representation of Lie algebra in math,
which runs from $1,2, \dots, \dd(G_{\text{gauge}})$ with the dimension  $\dd (G_{\text{gauge}})$ of Lie group $G_{\text{gauge}}$ (YM gauge group), especially here $\dd(G_{\text{gauge}})=\dd(\SU(\rN))=\rN^2-1$. 
By coupling $\bZ^{\text{$4$d}}_{{\text{YM}}}$ to 2-form background field $B$, we obtain a modified partition function 
\begin{widetext}
\bea
\label{eq:SUNYMN}
&&\bZ^{\text{$4$d}}_{{\text{YM}}}[B]
  =
  \int [{\cal D} {a}]
  \exp\big(  \ii 
  \int\limits_{M^4} (\frac{  \theta}{8 \pi^2}  \text{Tr}\, (F_{a'} -B_2 \mathbb{I}) \wedge (F_{a'} -B_2 \mathbb{I})) \big)\vert_{ \Tr(F_{a'})= \Tr {\dd a'} =\dd C_1= \rN B_2 = B_2 (\Tr \mathbb{I})}
   \\
&&  =
    \int [{\cal D} {a}]
  \exp\big( 
   \ii 
  \int\limits_{M^4} (\frac{  \theta}{8 \pi^2} \text{Tr}\, (F_{a'} \wedge F_{a'})- \frac{2\theta \rN}{8 \pi^2} B_2 \wedge B_2 + \frac{\theta \rN}{8 \pi^2} B_2 \wedge B_2 ) \big) 
  =
    \int [{\cal D} {a}]
  \exp\big( 
   \ii 
  \int\limits_{M^4} (\frac{  \theta}{8 \pi^2} \text{Tr}\, (F_{a'} \wedge F_{a'})-  \frac{\theta \rN}{8 \pi^2} B_2 \wedge B_2 ) \big). \nn
  \eea
\end{widetext}

\subsubsection{$\theta$ periodicity and the vacua-shifting of higher SPTs} 
\label{sec:theta}

Normally, people say {$\theta$} has the $2\pi$-periodicity, 
\bea
\theta \simeq \theta+ 2\pi.
\eea
However, this identification is imprecise. 
Even though the dynamics of the vacua $\theta$ and $\theta+ 2\pi$ is the same,
the $\theta$ and $\theta+ 2\pi$ can be differed by a short-ranged entangled gapped phase of SPTs of condensed matter physics.
In \cite{Gaiotto2017yupZoharTTT}'s language, the vacua of
$\theta$ and $\theta+ 2\pi$ are differed by a counter term (which is the 4d higher-SPTs in condensed matter physics language).
We can see the two vacua are differed by
  $ \exp (- 
   \ii 
  \int_{M^4}  \frac{\Delta \theta \rN}{8 \pi^2} B_2 \wedge B_2 )  \mid_{\Delta\theta =2 \pi}$,
which is
\be \label{eq:4dSPT-1}
\exp( 
  \ii 
\int_{M^4} \frac{ -\rN}{4 \pi} B_2 \wedge B_2 )  =
\exp( 
   \ii 
\int_{M^4}  \frac{-\pi}{ \rN} B_2 \cup B_2 ) ,
\ee
where on the right-hand-side (rhs), we switch the notation from the wedge product ($\wedge$) of differential forms
to the cup product ($\cup$) of cochain field, such that
$B_2 \to \frac{2 \pi}{N} B_2$ and $\wedge \to \cup$.  \footnote{We use the same notation $B_2$ for the differential form (which is $2\pi$ periodic) and the cocycle (which is $\rN$ periodic).}  More precisely, 
when 
$\rN={2^k}$ as a power of 2,
the vacua is differed by
\be \label{eq:4dSPT-2}
\exp( 
  \ii 
\int_{M^4}  \frac{-\pi}{ \rN} \mathcal{P}_2(B_2) ),
\ee
where a Pontryagin square term $\mathcal{P}_2:\H^2(-,\Z_{2^k})\to\H^4(-,\Z_{2^{k+1}})$
 is given by  \eq{eq:PB}
$
\mathcal{P}_2(B_2)=B_2\cup B_2+B_2\hcup{1}\delta B_2 =B_2\cup B_2+B_2\hcup{1} 2 \Sq^1 B_2
$.
{
This term is related to the generator of group cohomology $\H^4(\B^2 \Z_2,\U(1))=\Z_4$ when N=2,
and $\H^4(\B^2 \Z_{\N},\U(1))$ for general N.
This term is also related to the generator of cobordism group $\Omega^4_{\SO}(\B^2 \Z_2,\U(1))\equiv \Tor  \Omega_4^{\SO}(\B^2 \Z_2) =\Z_4$ when N=2,
and $\Omega^4_{\SO}(\B^2 \Z_\rN,\U(1))\equiv \Tor  \Omega_4^{\SO}(\B^2 \Z_\rN)$ for general N.
For the even integer N $={2^k}$, we have  $\Omega^4_{\SO}(\B^2 \Z_{\rN={2^k}},\U(1)) =\Z_{2 \rN={2^{k+1}}}$
via a $\Z_{2^k}$-valued 2-cochain in 2d to $\Z_{2^{k+1}}$ in 4d.
For our concern (e.g. N = 2, 4, etc.), we have $\Omega^4_{\SO}(\B^2 \Z_\rN,\U(1)) =\Z_{2 \rN}$,
and the Pontryagin square is well-defined.
For the odd integer N  that we concern (e.g. N = 3 or say N $=\text{p}$ an odd prime), 
Pontryagin square still can be defined, but it is $\H^{2n}(-,\Z_{\text{p}^k}) \to\H^{2\text{p}n}(-,\Z_{\text{p}^{k+1}})$.
So we do \emph{not} have Pontryagin square at $\rN = 3$ in 4d.
See more details on the introduction to mathematical background in \cite{W2}.
}
Since we know that the probed-field topological term characterizes SPTs \cite{1405.7689},
which are classified by group cohomology  \cite{Chen2011pg1106.4772, Wen1610.03911} or cobordism theory \cite{Kapustin2014tfa1403.1467,Kapustin1406.7329,Freed2016}; 
we had identified the precise SPTs (\eq{eq:4dSPT-1}, \eq{eq:4dSPT-2}) differed between 
the vacua of
$\theta$ and $\theta+ 2\pi$. 

\subsubsection{Time reversal $\cT$ transformation} \label{sec:TR}
As mentioned in \eq{eq:YM-Sym}, the global symmetry of YM theory (at $\theta=0$ and $\theta=\pi$) contains a time-reversal symmetry $\cT$. 
We denote the spacetime coordinates $\mu$ for
2-form $B \equiv B_2$ and 1-form $C_1$ gauge fields as
$B_{2,\mu\nu}$ and $C_{1,\mu}$ respectively. 
We denote $a_\mu = a_\mu^{\al} T^{\al}$.
Then, 
time reversal acts as:
\bea \label{eq:T}
\cT: &&a_0 \to - a_0,\;\;\;\; a_i \to + a_i, \;\;\;\;  (t, x_i) \to (-t, x_i). \nn\\
&& C_{1,0}\to  - C_{1,0},\;\;\;\; C_{1,i} \to + C_{1,i}. \nn \\
&& B_{2,0i} \to - B_{2,0i},\;\;\;\; B_{2,ij} \to + B_{2,ij}.
\eea
Thus the path integral transforms under time reversal, schematically, becomes
$\bZ^{\text{$4$d}}_{{\text{YM}}}[\cT B]$. 
By $\cT B$, we also mean $\cT B \cT^{-1}$
in the quantum operator form of $B$ (if we canonically quantize the theory).
More precisely,
\begin{widetext}
\begin{multline}
  \hspace*{-.8cm} 
\bZ^{\text{$4$d}}_{{\text{YM}}}[B]=\int [{\cal D} {a}]
  \exp\big( 
   \ii 
  \int\limits_{M^4} (\frac{  \theta}{8 \pi^2} \text{Tr}\, (F_{a'} \wedge F_{a'})-  \frac{\theta \rN}{8 \pi^2} B_2 \wedge B_2 ) \big) 
  \overset{\cT}{\longrightarrow} \\
  \hspace*{-.8cm} 
\bZ^{\text{$4$d}}_{{\text{YM}}}[\cT B]=     \int [{\cal D} {a}]
  \exp\big( 
   \ii 
  \int\limits_{M^4} (\frac{  -\theta}{8 \pi^2} \text{Tr}\, (F_{a'} \wedge F_{a'})-  \frac{-\theta \rN}{8 \pi^2} B_2 \wedge B_2 ) \big) 
  =\bZ^{\text{$4$d}}_{{\text{YM}}}[B]  \cdot
       \int [{\cal D} {a}]
  \exp\big( 
   \ii 
  \int\limits_{M^4} (\frac{  -2\theta}{8 \pi^2} \text{Tr}\, (F_{a'} \wedge F_{a'})-  \frac{-2\theta \rN}{8 \pi^2} B_2 \wedge B_2 ) \big).
\end{multline}
$\bullet$ When $\theta=0$, this remains the same $\bZ^{\text{$4$d}}_{{\text{YM}}}[\cT B]=\bZ^{\text{$4$d}}_{{\text{YM}}}[B]$.\\
$\bullet$ When $\theta=\pi$, this term transforms to

\begin{multline} 
\bZ^{\text{$4$d}}_{{\text{YM}}}[B] \cdot
       \int [{\cal D} {a}]
  \exp\big( 
   \ii 
{(-2 \pi)}  \int\limits_{M^4} (\frac{ 1 }{8 \pi^2} \text{Tr}\, (F_{a'} \wedge F_{a'})-  \frac{ \rN}{8 \pi^2}  B_2 \wedge B_2 ) \big)\\
=\bZ^{\text{$4$d}}_{{\text{YM}}}[B] \cdot
  \exp\big( 
   \ii 
{(-2 \pi)}  (-c_2) 
+\frac{  {{(-2 \pi)} \ii}}{8 \pi^2}   \int\limits_{M^4} ( \text{Tr} F_{a'}  \wedge  \text{Tr} F_{a'} -  { \rN}  B_2 \wedge B_2)
\big)\vert_{ \Tr(F_{a'}) = \rN B_2 }\\
=\bZ^{\text{$4$d}}_{{\text{YM}}}[B]  \cdot
  \exp\big( 
   \ii 
{2 \pi}  c_2
+\frac{  {{(-2 \pi)} \ii}}{8 \pi^2}   \int\limits_{M^4} ( \rN(\rN-1)  B_2 \wedge B_2)
\big)
=\bZ^{\text{$4$d}}_{{\text{YM}}}[B] \cdot
  \exp\big( 
\frac{  { -\ii \rN(\rN-1)}}{4 \pi}   \int\limits_{M^4} (   B_2 \wedge B_2)
\big)
\end{multline}
  \end{widetext}
where we apply the 2nd Chern number $c_2$ identity:
\be
\frac{  1}{8 \pi^2}   \int_{M^4} (\text{Tr} F_{a'}  \wedge  \text{Tr} F_{a'}  - \text{Tr}\, (F_{a'} \wedge F_{a'}) ) =  c_2 \in \Z.
\ee
We can add a 4d SPT state (of a higher form symmetry) as a counter term.
Consider again a 1-form $\Z_\rN$-symmetry ($\Z_{\rN,[1]}^e$) protected higher-SPTs, classified by a cobordism group $\Omega^4_{\SO}(\B^2 \Z_\rN,\U(1))$,
\bea
\exp( 
  \ii 
\int_{M^4}  \frac{p \pi}{ \rN} \mathcal{P}_2(B_2) ) \sim 
\exp( 
  \ii  \frac{p \rN}{4 \pi} \int B_2 \wedge B_2),
\eea
here we again convert the 2-cochain field $B_2$ to 2-form field $B_2$ (to recall, see Sec.~\ref{sec:theta}).
For any 4-manifold, according to \cite{Gaiotto:2014kfa, Putrov2016qdo1612.09298}, 
\bea
\frac{\rN p}{2} &\in& \Z \text{ (For even $\rN$, $p \in \Z$. For odd $\rN$, $p \in 2 \Z$)}. \;\;\\
p &\simeq& p + 2 \rN. 
\eea
For even $\rN$, there are $2\rN$ classes of 4d higher SPTs for $p \in  \Z$.
For odd $\rN$, there are $\rN$ classes of 4d higher SPTs for  $p \in 2 \Z$.

For spin 4-manifolds (when $p$ and $\rN$ are odd):
\bea
p &\in& \Z.\\
p &\simeq& p + \rN. 
\eea
In this case, there are $\rN$ classes on the spin manifold.
This 4d higher SPTs (counter term) under TR sym changes to:
$
\frac{p \rN}{4 \pi} \int B_2 \wedge B_2 \to  -\frac{p \rN}{4 \pi} \int B_2 \wedge B_2
$, or more precisely,
\be \label{eq:SPT-shift}
\int  \frac{p \pi}{ \rN} \mathcal{P}_2(B_2) \to - \int  \frac{p \pi}{ \rN} \mathcal{P}_2(B_2).
\ee
\subsubsection{Mixed time-reversal and 1-form-symmetry anomaly}
\label{sec:mix-T-1}

Now we discuss the {mixed time-reversal $\cT$ and 1-form $\Z_\rN$ symmetry anomaly of \cite{Gaiotto2017yupZoharTTT}} in details.
We re-derive based on our language in \cite{1405.7689}.
 The charge conjugation $\mathcal{C}$, parity $\mathcal{P}$ and time-reversal $\cT$ form $\mathcal{C}\mathcal{P} \cT$.
 Since $\mathcal{C}\mathcal{P} \cT$ is a global symmetry for this YM theory, 
we can also interpret this anomaly as
a {mixed $\mathcal{C}\mathcal{P}$ and 1-form $\Z_\rN$ symmetry anomaly}.

So overall, $\bZ^{\text{$4$d}}_{{\text{YM}}}[B]$, say with a 4d higher-SPT $\frac{p \rN}{4 \pi} \int B_2 \wedge B_2$ labeled by $p$,  is sent to
  \begin{widetext} \label{eq:4d-YM-5d-SPTs}
\begin{multline}   
\bZ^{\text{$4$d}}_{{\text{YM}}}[B]  \cdot
  \exp\big( 
\ii (\frac{  {- \rN(\rN-1)}}{4 \pi} +\frac{-2p \rN}{4 \pi})  \int_{M^4} (   B_2 \wedge B_2) \big)
   =\bZ^{\text{$4$d}}_{{\text{YM}}}[B]  \cdot
  \exp\big( 
\frac{  { -\ii \rN(\rN-1 +2p)}}{4 \pi}   \int_{M^4} (   B_2 \wedge B_2) \big).
\end{multline}
  \end{widetext}
\begin{enumerate}

\item
For even $\rN$, and $\theta= \pi$, here the 4d higher SPTs (counter term) labeled by 
$p$ becomes labeled by  $-(\rN-1)-p$. To check whether there is a mixed anomaly or not, 
which asks for the identification of two 4d SPTs before and after time-reversal transformation.
Namely
$(\rN-1 +2p)=0$ (mod out the classification of 4d higher SPTs given below \eq{eq:4dSPT-2})
\emph{cannot} be satisfied for any $p \in \Z$ (actually $p \in \Z_{2^{k+1}}$ via the Pontryagin square,
which sends a $\Z_{2^k}$-valued 2-form in 2d to $\Z_{2^{k+1}}$-class of 4d higher SPTs. For N = 2, we have $p \in \Z_4$.) 

So this indicates that for any $p$ (with or without 4d higher SPTs/counter term) in the YM vacua,
we \emph{detect} the mixed time-reversal $\cT$ and 1-form $\Z_\rN$ symmetry anomaly, which requires
a 5d higher SPTs to cancel the anomaly. We will write down this
5d higher SPTs/counter term in \Sec{sec:cobo-top}.

\item
For even $\rN$, and $\theta= 0$, we have  $\bZ^{\text{$4$d}}_{{\text{YM}}}[\cT B]=\bZ^{\text{$4$d}}_{{\text{YM}}}[B]$ without 4d SPTs.
With 4d SPTs, the only shift is \eq{eq:SPT-shift}, so to check the anomaly-free condition, we need 
$p=-p$, or $2p=0$, mod out the classification of 4d higher SPTs given below \eq{eq:4dSPT-2}.
This anomaly-free condition can be satisfied for $p=0$. 
For N = 2, we can also have $ 2 p = 0$ mod 4, which is true for $p=0,2$, even with the $p=2$-class of 4d SPTs.
In that case, there is no mixed higher anomaly of $\cT$ and $\Z_{\rN,[1]}^e$ symmetry,

\item
For odd $\rN$, and $\theta= \pi$,
the $(N-1 +2p)=0$ (mod out the classification of 4d higher SPTs given below \eq{eq:4dSPT-2})
can be satisfied for some $p=\frac{1-N}{2} \in \Z$, but $p$ needs to be even $p \in 2 \Z$  on a non-spin manifold.
If $p=\frac{1-N}{2} \in 2 \Z$, the 4d SPTs can be defined on a non-spin manifold.
If $p=\frac{1-N}{2} \in \Z$, the 4d SPTs can only be defined on a spin manifold.
So, for an odd $\rN$, there can be \emph{no} mixed anomaly at $\theta= \pi$, a 4d higher SPTs/counter term of $p=\frac{1-N}{2}$ 
preserves the $\cT$-symmetry and 1-form $\Z_\rN$-symmetry (such that two symmetries can be regulated locally onsite \cite{Wen2013oza1303.1803, 1405.7689, Wang2017locWWW1705.06728}).

\end{enumerate}

\subsubsection{Charge conjugation $\cC$, parity $\cP$, reflection $\cR$,  $\cC\cT$, $\cC\cP$ transformations,
$\Z_2^{CT} \times (\Z_{\rN,[1]}^e \rtimes  \Z_2^C )$ and
$\Z_2^{CT} \times \Z_{2,[1]}^e$
-symmetry, and their higher mixed anomalies} \label{sec:C-P-R}


Follow \Sec{sec:TR} and the discrete $\cT$ transformation in \eq{eq:T}, 
and we denote $a_\mu = a_\mu^{\al} T^{\al}$, below
we list down additional
discrete transformations including 
charge conjugation $\cC$, parity $\cP$,  $\cC\cT$, $\cC\cP$: 
\bea 
\cC: && a_\mu^{\al}(t,x) \to a_\mu^{\al}(t,x), \;\;\;\; (t, x_i) \to (t, x_i). \label{eq:C}
\\
&& \:    T^{\al}  \to -T^{{\al}*}, \;\;\;\; a_\mu \to -a_\mu^*.\nn\\
\cP: &&\; a_0 \to a_0,\;\;\;\; a_i \to -a_i, \;\;\;\;  (t, x_i) \to (t, -x_i).\\
\cC\cT: && a_0 \to +a_0^* ,\;\;\;\; a_i \to -a_i^*, \;\;\;\;  (t, x_i) \to (-t, x_i).\;\;\; \label{eq:CT}\\
\cC\cP: &&\; a_0 \to - a_0^*,\;\;\;\; a_i \to +a_i^*, \;\;\;\;  (t, x_i) \to (t, -x_i).\; \;\;  \\
\cC\cP\cT: &&\; a_\mu \to + a_\mu^*,\;\;\;\;  (t, x_i) \to (-t, -x_i).
\eea
{The $*$ means the complex conjugation.} 
In Euclidean spacetime, we can regard the former $\cT$ in \eq{eq:T} (or $\cC\cT$ \eq{eq:CT}) as a 
reflection $\cR$ transformation \cite{Witten2016cio1605.02391}, which we choose to flip any of the Euclidean coordinate.
See further discussions of the crucial role of discrete symmetries in YM gauge theories
in \cite{2017arXiv171111587GPW}.

We can ask whether there is any higher mixed anomalies between the above discrete symmetries and the 1-form center symmetry. 
We can easily check that whether the  {$\frac{\theta}{8 \pi^2}\text{Tr}\, (F_{a} \wedge F_{a})$ term flips sign to $-\frac{\theta}{8 \pi^2}\text{Tr}\, (F_{a} \wedge F_{a})$} 
under any of the discrete symmetries.
Among the $\cC, \cP$, and $\cT$, only the $\cC$ \emph{does not} flip the ${\theta}$ term
and $\cC$ is a good global symmetry for all ${\theta}$ values.
So the answer is that each of the 
\bea \label{eq:T-CT-mix}
\text{$\cT$, $\cP$, $\cC \cT$ and $\cC \cP$}, 
\eea
have itself mixed anomalies with the 1-form center symmetry. 
Only 
\bea \label{eq:C-PT-mix}
\text{$\cC$, $\cP \cT$,  and $\cC \cP \cT$}, 
\eea
do not have mixed anomalies with the 1-form center symmetry. 

Now, we come back to explain the non-commutative nature (the semi-direct product ``$\rtimes$'') of \eq{eq:YM-Sym} between 0-form and 1-form symmetries
$$
\Z_{\rN,[1]}^e \rtimes (\Z_2^T \times \Z_2^C).
$$
Obviously $(\Z_2^T \times \Z_2^C)$ is due to that $\cC$ and $\cT$ commute, and the combined diagonal group
diag$(\Z_2^T \times \Z_2^C)=\Z_2^{CT}$ has the
group generator $\cC \cT$.

We note that to physically understand some of the following statements, it may be helpful to view the symmetry transformation in the Minkowski/Lorentz signature instead of
the Euclidean signature.\footnote{In the Minkowski case, 
we also need to regard the time-reversal symmetries ($\cT$ and $\cC \cT$)  as anti-unitary symmetry, instead of the unitary symmetry (as the Euclidean case).}

\noindent
{
$\bullet$  The  non-commutative nature $\Z_{\rN,[1]}^e \rtimes \Z_2^T$ is due to that the
$\Z_2^T$ keeps 1-Wilson loop
$W_e=\Tr_{\text{R}}( \text{P} \exp(\ii \oint a)) \to \Tr_{\text{R}}( \text{P} \exp((-\ii )(-\oint a))) = W_e$ invariant,
while $\Z_2^T$ flips the 2-surface $U_e \to U_e^\dagger= U_e^{-1}$ due to the orientation of $U_e$ and its boundary 't Hooft line is flipped.
Thus, the 1-form $\Z_{\rN,[1]}^e$-symmetry charge of $W_e$, measured by the topological number of linking between $W_e$ and $U_e$,
now flips from $ n \in \Z_\rN$ to $ -n = \rN-n  \in \Z_\rN$. 
Since the charge operator of $\Z_{\rN,[1]}^e$ symmetry, $U_e$, is flipped thus does not commute under the $\Z_2^T$ symmetry,
this effectively defines the semi-direct product in a dihedral group like structure of $\Z_{\rN,[1]}^e \rtimes \Z_2^T$.
}

\noindent
{$\bullet$ The commutative nature $\Z_{\rN,[1]}^e \times \Z_2^{CT}$ is due to that the
$\Z_2^{CT}$ flips 1-Wilson loop $W_e=\Tr_{\text{R}}( \text{P} \exp(\ii \oint a)) \to W_e^\dagger= W_e^{-1}$,
while $\Z_2^{CT}$ keeps the 2-surface $U_e \to U_e$ invariant.
We can see that the $\Z_2^{CT}$ and $\Z_2^{T}$ flips the 1-loop and 2-surface oppositely.
Thus, the 1-form $\Z_{\rN,[1]}^e$-symmetry charge of $W_e$, measured by the topological number of linking between $W_e$ and $U_e$,
again flips from $ n \in \Z_\rN$ to $ -n = \rN-n \in \Z_\rN$. 
But the charge operator of $\Z_{\rN,[1]}^e$ symmetry, $U_e$, is invariant thus does commute under the $\Z_2^{CT}$ symmetry,
this effectively defines the direct product in a group structure of $\Z_{\rN,[1]}^e \times \Z_2^{CT}$.
}

\noindent
{$\bullet$ The non-commutative nature $\Z_{\rN,[1]}^e \rtimes \Z_2^C$ is due to that the
$\Z_2^C$ in \eq{eq:C} flips $W_e=$
$\Tr_{\text{R}}( \text{P} \exp(\ii \oint a))$ $\to$ $\Tr_{\text{R}}( \text{P} \exp(\ii (-\oint a^*))$
$=\Tr_{\text{R}}( \text{P} \exp(\ii (\oint a))^*)$
$=\Tr_{\text{R}}( \text{P} \exp(\ii (\oint a))^\dagger)$
$= W_e^\dagger= W_e^{-1}$,
while $\Z_2^C$ also flips the 2-surface $U_e \to U_e^\dagger= U_e^{-1}$ for the same reason.
Thus, the 1-form $\Z_{\rN,[1]}^e$-symmetry charge of $W_e$, measured by the topological number of linking between $W_e$ and $U_e$,
is invariant under $\Z_2^C$. 
But the charge operator of $\Z_{\rN,[1]}^e$ symmetry, $U_e$, is flipped thus does not commute under the $\Z_2^C$ symmetry,
this effectively defines the semi-direct product in a dihedral group like structure of $\Z_{\rN,[1]}^e \rtimes \Z_2^C$.}
{Note that the potentially related dihedral group structure of Yang-Mills theory under a dimensional reduction to $\R^3 \times S^1$ is recently explored in
\cite{Gaiotto2017yupZoharTTT, Aitken2018kky1804.05845}.}

When N = 2, it is obvious that we simply have the direct product $\Z_{2,[1]}^e \times \Z_2^T$ as \eq{eq:YM-SU2-Sym-T}.

We can rewrite \eq{eq:YM-Sym}'s 0-form and 1-form symmetries
\bea \label{eq:YM-Sym-2}
\boxed{
\Z_2^{CT} \times (\Z_{\rN,[1]}^e \rtimes  \Z_2^C )
}.
\eea
\\
We can rewrite \eq{eq:YM-SU2-Sym-T} as
\bea \label{eq:YM-SU2-Sym}
\boxed{
\Z_2^{CT} \times \Z_{2,[1]}^e  }.
\eea
{
It is related to the fact that for SU(2) YM theory, the charge conjugation $\Z_2^C$ is inside the gauge group, 
because there is no outer automorphism of SU(2) but only an inner automorphism ($\Z_2$) of SU(2).
For N=2, the charge conjugation matrix 
$\cC_{\SU(2)}=e^{\im\frac{\pi }{2}\sigma^2}\in \SU(2)$ is a matrix that provides an isomorphism map between fundamental representations 
of $SU(2)$ and its complex conjugate. 
We have $\cC_{SU(2)}\sigma^{{\al}} \cC_{SU(2)}^{-1}=-{(\sigma^{{{\al}}})^*}$. 
Let $U_{\SU(2)}$ be the unitary ${\SU(2)}$ transformation on the ${\SU(2)}$-fundamentals, so 
$
\cC_{SU(2)} U_{\SU(2)} \cC_{SU(2)}^{-1}=\exp(-\imth \frac{\theta}{2} (\sigma^{{{\al}}})^*)=U_{\SU(2)}^*,
$
which is a $\Z_2$ inner automorphism of SU(2).
}

{We propose that {the structure of \eq{eq:YM-Sym}, \eq{eq:YM-SU2-Sym-T}, \eq{eq:YM-Sym-2} and \eq{eq:YM-SU2-Sym} can be regarded as an analogous 2-group.}
It can be helpful to further organize this 2-group like data into the context of \cite{Benini2018reh1803.09336}.}


\subsection{2d $\mathbb{CP}^{\rN-1}$-sigma model}

Here we consider 2d $\mathbb{CP}^{\rN-1}$-model \cite{Witten1978bc}, which is
a 2d sigma model with a target space $\mathbb{CP}^{\rN-1}$.
The $\mathbb{CP}^{\rN-1}$ model is a 2d toy model which mimics some similar behaviors of
4d YM theory: dynamically-generated energy gap and asymptotic freedom, 
etc.
We will focus on 2d $\mathbb{CP}^{\rN-1}$-model at $\theta=\pi$.

\subsubsection{Related Models}
The path integral 
of 2d $\mathbb{CP}^{\rN-1}$-model is
\begin{widetext}
\begin{multline}
\label{eq:CP-pi}
\bZ^{\text{$2$d}}_{{\CP^{\rN-1}}}
\equiv\int [{\cal D} {z}] [{\cal D} {\bar z}] [{\cal D} a'] \exp\big( - S_{{\CP^{\rN-1}}+\theta}[z, {\bar z}, a'] \big)
\equiv \int  [{\cal D} {z}_j][{\cal D} {\bar z}_j]   [{\cal D} {a'}] \exp\big( - S_{{\CP^{\rN-1}}}[z_j, {\bar z}_j, a'] \big)\exp\big( - S_{\theta}[a'] \big)
\\
\equiv
\int
 [{\cal D} {z}] [{\cal D} {\bar z}]
 [{\cal D} {a'}] 
\delta(|z|^2-r^2) \exp\big(   \Big(- \int_{M^2} d^2x (\frac{1}{g'^2} |D_{a',\mu} z |^2)
+  \int_{M^2} (\frac{  \ii \theta}{2 \pi}  F_{a'}) \Big)\big).
\end{multline}
\end{widetext}
The ${z}_j \in \C$ is a complex-valued  field variable, with an index $j=1,\dots, \rN$.
(In math, the ${z}_j \sim c {z}_j$, identified by any complex number $c \in \C^{\times}$ excluding the origin,
is known as the homogeneous coordinates of the target space ${\CP^{\rN-1}}$.)
The delta function imposes a constraint: 
$|z|^2 \equiv \sum_{j=1}^{\rN} |z_j|^2=r^2$, here $r \in \R$ specifies the size of $\CP^{\rN-1}$.
The delta function $\delta(|z|^2-r^2)$ may be also replaced by a potential term, such
as the $\frac{\lambda}{4}(|z|^2-r^2)^2$ potential, at large $\lambda$ coupling energetically constraining $|z|^2=r^2$.
Here $|D_{a',\mu} z |^2 \equiv (D_{a',\mu} z )^\dagger (D_{a'}^\mu  z )$.
Here $F_{a'}=\dd a'$ is the U(1) field strength of $a'$.

For 2d $\CP^{1}$ (2d $\CP^{\rN-1}$ at $\rN=2$),
we can rewrite the model as the O(3) nonlinear sigma model (NLSM).
The O(3) NLSM is parametrized by an O(3) $=\SO(3) \times \Z_2$ N\'eel vector
$\vec {n}=({n}_1,{n}_2, {n}_3)$,
which {obeys}
\be
\vec {n} = \frac{1}{r^2} z^\dagger_i \vec \sigma_{ij} z_j
\ee 
with
$|\vec {n}|^2=1$ and Pauli matrix 
$\vec \sigma=(\sigma^1,\sigma^2,\sigma^3 )$. 
It is called N\'eel vector because the 2d $\CP^{1}$ or O(3) NLSM
describes the Heisenberg anti-ferromagnet phase of quantum spin system \cite{Haldane1982rjPLA, Haldane1983ruPRL}.
To convert \eq{eq:CP-pi} to \eq{eq:O3-pi}, notice that we do not introduce the kinetic Maxwell term $|F_{a'}|^2$ for the U(1) photon {field} ${a'}$,
thus ${a'}$ is an auxiliary field, that can be integrated out and \eq{eq:CP-pi} is constrained by the EOM:
\be
a_\mu'= - \frac{\ii}{r^2} \sum_{j=1}^2 \bar z_j \partial_{\mu} z_j
= \frac{\ii}{2r^2} \sum_{j=1}^2 (  z_j \partial_{\mu} \bar z_j-\bar z_j \partial_{\mu} z_j )
,\ee 
and we can derive:
\bea
|D_{a',\mu} z |^2 &=&  \sum_{j=1}^2 |D_{a',\mu} z_j |^2 =((\frac{r}{2 })^2 \partial_{\mu}  \vec n \cdot \partial^{\mu}  \vec n), \\
\frac{  \ii \theta}{2 \pi}  \epsilon^{\mu \nu} \partial_\mu {a'}_\nu&=& (\frac{  \ii \theta}{8 \pi}  \epsilon^{\mu \nu} \vec{n} \cdot ( \partial_{\mu}   \vec{n} \times  \partial_{\nu}   \vec{n})).
\eea
Then we rewrite $\bZ^{\text{$2$d}}_{{\CP^{1}}}$ as $\bZ^{\text{$2$d}}_{{\tO(3)}}$ of the O(3) NLSM path integral: 
\begin{widetext}
\begin{multline}
\label{eq:O3-pi}
\bZ^{\text{$2$d}}_{{\tO(3)}}
\equiv\int [{\cal D} \vec {n}]  \delta(|\vec {n}|^2-1) \exp\big( - S_{{\tO(3)}+\theta}[\vec {n}] \big)
\\
\equiv
\int [{\cal D} \vec {n}]  \delta(|\vec {n}|^2-1)
 \exp\big(   \Big(- \int_{M^2} d^2x (\frac{1}{g''^2} \partial_{\mu}  \vec n \cdot \partial^{\mu}  \vec n)
+  \int_{M^2} (\frac{  \ii \theta}{8 \pi}  \epsilon^{\mu \nu} \vec{n} \cdot ( \partial_{\mu}   \vec{n} \times  \partial_{\nu}   \vec{n})) \Big)\big).
\end{multline}
\end{widetext}
Note that 
$(\frac{  \ii \theta}{8 \pi}  \epsilon^{\mu \nu} \vec{n} \cdot ( \partial_{\mu}   \vec{n} \times  \partial_{\nu}   \vec{n}))
=(\frac{  \ii \theta}{4 \pi}  \vec{n} \cdot ( \partial_{\tau}   \vec{n} \times  \partial_{x}   \vec{n}))$.
The O(3) NLSM coupling ${g''}$ in
$((\frac{1}{g''})^2 \partial_{\mu}  \vec n \cdot \partial^{\mu}  \vec n)=((\frac{r}{2 g'})^2 \partial_{\mu}  \vec n \cdot \partial^{\mu}  \vec n)$ 
is related to the $\CP^{1}$ model via ${g''}= ({2 g'}/{r})$, which is inverse proportional to the radius size of the 2-sphere $\CP^{1}=S^2$.

In fact, the UV high energy theory of $\bZ^{\text{$2$d}}_{{\CP^{1}}}=\bZ^{\text{$2$d}}_{{\tO(3)}}$
is known to be, in Renormalization Group (RG) flow, flowing to the same IR conformal field theory  CFT from another UV model via the
$\SU(2)_1$-WZW model (Wess-Zumino-Witten model \cite{Wess1971yuWZ, Witten1983twGlobalACA, Witten1983arNab}).
The $\SU(\rN)_k$-WZW model at the level $k$ is
\begin{widetext}
\bea
\label{eq:WZW-2}
\bZ^{\text{WZW}}_{{\SU(\rN)_{k}}}= \int [{\cal D} U]  [{\cal D} U^\dagger]  \exp(-\frac{k}{8 \pi} \int_{M^2} \dd^2x \Tr\big(\partial_\mu U^\dagger \partial^\mu U \big)
+\frac{\ii k}{12 \pi}\int_{M'^3}  \Tr\big( (U^\dagger \dd  U)^3 \big) ),
\eea
\end{widetext}
with ${M^2=\partial({M'^3})}$.
At $\rN =2$, the UV theory of $\bZ^{\text{$2$d}}_{{\CP^{1}}}=\bZ^{\text{$2$d}}_{{\tO(3)}}$ flows to
this 2d CFT called the $\SU(2)_1$-WZW CFT at IR. 
The global symmetry can be preserved at IR.

For the general 2d $\CP^{\rN-1}$-model, its global symmetry can also be embedded into another 
$\SU(\rN)_1$-WZW model at UV; although unlike $\rN =2$ case,
$\CP^{\rN-1}$-models for $\rN >2$ conventionally and generically do not flow to an IR CFT.
 For $\rN >2$, there exist UV-symmetry preserving relevant deformations
 driving the RG flow away from an IR CFT.
 The global symmetry may be spontaneously broken, and the vacua can be gapped and/or degenerated.
See for example \cite{Affleck1987chHaldane, Affleck1988wz} and references therein.
 
\subsubsection{Global symmetry:\\
 $\Z_2^{CT} \times \PSU(2) \times \Z_2^{C}$ and
 $\Z_2^{CT} \times (\PSU(\rN) \rtimes \Z_2^{C'})$}

\label{sec:CP-global}

Let us check the global symmetry of 2d $\CP^{\rN-1}$-model.\\

\noindent
\underline{Continuous global symmetry:}
In \eq{eq:CP-pi}, it is easy to see the continuous global SU(N) transformation rotating between the SU(N) fundamental complex scalar multiplet $z_j$ 
{via}
$
z \to V z=V_{ij} z_j= (e^{\ii \Theta^\alpha T^\alpha})_{ij} z_j
$
which has its $\Z_\rN$-center subgroup being gauged away by the U(1) gauge field $a'$. 
So we have the net continuous global symmetry 
\bea
\PSU(\rN)=\SU(\rN)/\Z_\rN=\U(\rN)/\U(1),
\eea 
which acts on gauge invariant object faithfully (e.g. the
 $\PSU(2)=\SO(3)$ symmetry can act on the gauge-invariant $\vec n$ vector in the 2d $\CP^{1}$-model or O(3) NLSM faithfully).\\

Now we explore 2d $\CP^{\rN-1}$-model's discrete global symmetry as finite groups.\\

\noindent
\underline{Discrete global symmetry for $\rN=2$:}
\\
\noindent 
{$\bullet$ $\Z_2^T$-time-reversal symmetry, there is a $\cT$-symmetry allowed for any $\theta$, acting on fields and coordinates of \eq{eq:CP-pi} and \eq{eq:O3-pi},
whose transformations become
\bea
\Z_2^T: \;\;&& z_i  \to \epsilon_{ij}  \bar z_j , 
\quad  \vec n \to -\vec n, \; \nn\\
&& (a'_t, a'_x) \to (a'_t, -a'_x) , \quad (t,x) \to (-t,x).
\eea
Here a Pauli matrix $\sigma^2_{ij}$ gives $\epsilon_{ij} =\ii \sigma^2_{ij}$. \\
 \noindent
$\bullet$  $\Z_2^x$-translation symmetry ($\equiv \Z_2^{C}$ as an effective charge conjugation symmetry $\Z_2^{C}$ ) allowed for $\theta= 0, \pi$, acts as}
\bea
\Z_2^x \; (\equiv \Z_2^{C}):  \;\;&& z_i  \to \epsilon_{ij}  \bar z_j  ,\quad  \vec n \to -\vec n, \; \nn\\
&& (a'_t, a'_x) \to -(a'_t, a'_x) , \quad (t,x) \to (t,x). \quad
\eea
It is easy to understand the role of $\Z_2^x$-translation on the 
UV-lattice model of Heisenberg anti-ferromagnet (AFM) phase of quantum spin system \cite{Haldane1982rjPLA, Haldane1983ruPRL}.
Its AFM Hamiltonian operator is
\bea \label{eq:H}
\hat H = \sum_{\langle i,j\rangle} |J|\hat{\vec S}_i \cdot \hat{\vec S}_j + \dots 
\eea where ${\langle i,j\rangle}$
is for the nearest-neighbor lattice sites $i$ and $j$'s AFM interaction between spin operators $\hat{\vec S}$,
and $|J|>0$ is the AFM coupling.
So $\Z_2^x$-translation flips the spin orientation, also flips the AFM's N\'eel vector $\vec n \to -\vec n$.\\
 %
 {
  \noindent
$\bullet$  $\Z_2^P$-parity symmetry allowed for $\theta= 0, \pi$ acts as
\bea \label{eq:Z2parity}
\Z_2^P:  \;\;&& z_i  \to  + z_j  ,\quad  \vec n \to + \vec n, \; \nn\\
&& (a'_t, a'_x) \to (a'_t, -a'_x) , \quad (t,x) \to (t,-x). \quad
\eea
\noindent
$\bullet$  $\cC\cP\cT$-symmetry (as $\Z_2^{CPT}=\diag(\Z_2^{C}, \Z_2^{P}, \Z_2^{T})$, the diagonal symmetry generator of $\cC$, $\cP$ and $\cT$) allowed for any $\theta$ acts as
\bea \label{eq:Z2parity}
\Z_2^{CPT}:  \;\;&& z_i  \to  + z_j  ,\quad  \vec n \to + \vec n, \; \nn\\
&& (a'_t, a'_x) \to -(a'_t, a'_x) , \quad (t,x) \to (-t,-x). \quad
\eea
}
 \noindent
$\bullet$ 
${\Z_2^{C'}}$-another charge conjugation 
symmetry of $\CP^{\rN-1}$-model allowed for $\theta= 0, \pi$ acts as
\bea \label{eq:Z2charge}
{\Z_2^{C'}}:  \;\;&& z_i  \to   \bar z_i ,\quad   (n_1,n_2,n_3) \to (n_1, -n_2,n_3), \; \nn\\
&& (a'_t, a'_x) \to -(a'_t, a'_x) , \quad (t,x) \to (t,x).
\eea
 \noindent
$\bullet$ 
${\Z_2^{C' T}}$-symmetry
allowed for $\theta= 0, \pi$ acts as
\bea \label{eq:Z2CT}
{\Z_2^{C' T}}:  \;\;&& z_i  \to \epsilon_{ij}   z_j,\quad   (n_1,n_2,n_3) \to (-n_1, n_2, -n_3), \; \nn\\
&& (a'_t, a'_x) \to (-a'_t, a'_x) , \quad (t,x) \to (-t,x).
\eea

\noindent
$\bullet$  $\Z_2^{xT}$-symmetry ($\equiv \Z_2^{CT}$) as another choice of time-reversal allowed for $\theta= 0, \pi$, acts as
\bea
\Z_2^{xT} \;(\equiv \Z_2^{CT}):  \;\;&& z_i  \to    z_i  ,\quad  \vec n \to \vec n,  \\
&& (a'_t, a'_x) \to (-a'_t, a'_x) , \quad (t,x) \to (-t,x). \quad \nn
\eea

Next we check the commutative relation between the above continuous PSU(N) and the discrete symmetries
 
 For N = 2, we see that $\Z_2^T$ commutes with $\PSU(2)$, because
 $\cT V z = \ii \sigma^2 (V z)^*=\ii \sigma^2 V^* \bar z
 =V \ii \sigma^2 \bar z = V \cT z$.
 Similarly,  $\Z_2^x$ commutes with $\PSU(2)$.
 So, $\Z_2^{xT}$ commutes with $\PSU(2)$.
 We see that $\Z_2^{C'}$ does not commute with $\PSU(2)$, because
 $\cC' V z =(V z)^*= V^* \bar z$ while $V \cC' z= V \bar z$.
 Therefore, 
 {$\Z_2^{C' T}=\text{diag}(\Z_2^{C'}, \Z_2^T)$} also does not commute with $\PSU(2)$.\\

{Similar to \eq{eq:T-CT-mix} for 4d YM theory, here for 2d $\CP^{N-1}$ model, 
we can check that whether the $\frac{\theta}{2 \pi}(F_{a'})$ term 
{flips sign to $-\frac{\theta}{2 \pi}(F_{a'})$} 
under any of the discrete symmetries.
Among the $\cC, \cP$, and $\cT$ for 2d $\CP^{N-1}$ model, 
only the $\cT$ \emph{does not} flip the ${\theta}$ term
and $\cT$ is a good global symmetry for all ${\theta}$ values.
So each of the 
\bea \label{eq:T-CT-mix-Sigma}
\text{$\cC$, $\cP$, $\cC \cT$ and $\cP \cT$}, 
\eea
plays the similar role for the 2d anomalies at $\theta=\pi$. 
Only 
\bea \label{eq:T-CP-mix-Sigma}
\text{$\cT$, $\cC \cP$, and $\cC \cP \cT$}, 
\eea
are good symmetries for all $\theta$.}\\

\noindent
\underline{Global symmetry for $\rN=2$:}\\
Overall, for 2d $\CP^1$ model at $\theta=0,\pi$, we can combine the above to get the full
0-form global symmetries 
\bea \label{eq:CP1-full-sym}
\boxed{
\Z_2^T \times \PSU(2) \times \Z_2^{x}=\Z_2^T \times \tO(3)
},
\eea
which is the same as $$\Z_2^T \times \PSU(2) \rtimes {\Z_2^{C'}}$$ 
with a semi-direct product
``$\rtimes$'' since $\PSU(2)$ and ${\Z_2^{C'}}$ 
do not commute.

It is very natural  to regard $\Z_2^{xT}$-symmetry as the \emph{new} $\Z_2^{CT}$-symmetry,
because it flips the time coordinates $t \to -t$, but it does not complex conjugate the $z$. 
So we may define\footnote{
\label{ft:CT}
Above we discuss
$ \Z_2^{CT} \equiv \Z_2^{xT}$ and $\Z_2^{T}$  both commute with the $\PSU(2)$ (also $\SU(2)$) for bosonic systems (bosonic QFTs).
Indeed, the 
$ \Z_2^{CT}$ and $\Z_2^{T}$ reminisce the discussion of  \cite{2017arXiv171111587GPW} (e.g. Sec.~II),
for the case including the fermions (with the fermion parity symmetry ${\Z_2^F}$ acted by $(-1)^F$), 
we have the natural $\Z_2^{CT}$-time-reversal symmetry, without taking complex conjugation on the matter fields, 
which gives rise to  the full symmetry
$\frac{\text{Pin}^+\times \SU(2)}{{\Z_2^F}}$;
while the other $\Z_2^{T}$-time-reversal symmetry, involving complex conjugation on the matter fields, gives rise to  
$\frac{\text{Pin}^-\times \SU(2)}{{\Z_2^F}}$.
}
\bea
\Z_2^{CT} \equiv \Z_2^{xT}. 
\eea
Similarly, we may regard the $\Z_2^{x}$-translation as a \emph{new} charge conjugation symmetry $\Z_2^{C} \equiv \Z_2^{x}$.
 
Therefore, 0-form global symmetries \eq{eq:CP1-full-sym} can  also be
\bea \label{eq:CP1-full-sym-CT}
&&\boxed{
\Z_2^{CT} \times \PSU(2) \times \Z_2^{x} \equiv \Z_2^{CT} \times \PSU(2) \times \Z_2^{C} 
} \nn \\
&&\equiv
\boxed{ \Z_2^{CT} \times \tO(3)
}.
\eea

\noindent
\underline{Global symmetry for $\rN>2$:}\\
For 2d $\CP^{\rN-1}$ model \eq{eq:CP-pi} at $\theta=0,\pi$, $\rN>2$, 
we follow the above discussion and the footnote \ref{ft:CT},
we again can define a natural definition of $\Z_2^{CT}$ (without involving the complex conjugation of $z$ fields).
Then we have instead the full
0-form global symmetries: 
\bea \label{eq:CPN-full-sym}
\boxed{
\Z_2^{CT} \times (\PSU(\rN) \rtimes \Z_2^{C'}) 
}
,
\eea
where again $\Z_2^{C}$ 
acts on $z_i  \to   \bar z_i$,  $a'_\mu \to - a'_\mu$ and $(t,x) \to (t,x)$ as \eq{eq:Z2charge}. 


We remark that the SU(2) (or N = 2 for $\CP^{1}$ model) is special because its
order-2 automorphism is an inner automorphism. 
The SU(2) fundamental representation is equivalent to its conjugate. 
This is related to the fact that
both $\Z_2^{CT}$ and $\Z_2^{T}$ can commute with the SU(2) or PSU(2),
also the remark we made in the footnote \ref{ft:CT}.

For SU(N) 
{with $\rN > 2$, 
we gain an}
order-2 automorphism as an outer automorphism, which is 
the $\Z_2$ symmetry of Dynkin diagram A${}_{N-1}$, swapping fundamental with anti-fundamental representations. 
Although we have $\Z_2^{CT} \times\PSU(\rN) $ in \eq{eq:CPN-full-sym}, we would have
$\Z_2^{T} \ltimes \PSU(\rN)$ for $\rN > 2$.
See related and other detailed discussions in \cite{2017arXiv171111587GPW}.

{The above we have considered the global symmetry (focusing on the internal symmetry, and the discrete sector of the spacetime symmetry) 
without precisely writing down their continuous spacetime symmetry group part.}
In \Sec{sec:cobo-top}, we like to write down the ``full'' global symmetry including the spacetime symmetry group.

\section{Cobordisms, Topological Terms, and Manifold Generators:
Classification of All Possible Higher 't Hooft Anomalies}
\label{sec:cobo-top}

\subsection{Mathematical preliminary and co/bordism groups}

\label{sec:A}

Since we have obtained the full global symmetry $G$ (including the 0-form and higher symmetries) of 4d YM and 2d $\CP^{\rN-1}$ model,
we can now use the knowledge that their 't Hooft anomalies are classified by
5d and 3d cobordism invariants of the same global symmetry \cite{Freed2016}.
Namely, we can classify the 't Hooft anomalies
by enlisting the complete set of all possible cobordism invariants
from their corresponding  
5d and 3d bordism groups,  whose 5d and 3d manifold generators endorsed with the $G$ structure.

To begin with, 
we should rewrite the global symmetries in previous sections (e.g.
(\eq{eq:YM-Sym}/\eq{eq:YM-Sym-2}), 
(\eq{eq:YM-SU2-Sym-T}/\eq{eq:YM-SU2-Sym}))
into the form of
\bea  \label{eq:Gall}
G\equiv ({\frac{{G_{\text{spacetime} }} \ltimes  {\mathbb{G}_{\text{internal}} }}{{N_{\text{shared}}}}}),
\eea
where the ${G_{\text{spacetime} }}$ is the spacetime symmetry,
the ${\mathbb{G}_{\text{internal}} }$ the internal symmetry,\footnote{ \label{footnote:wj}
Later we denote the probed background spacetime $M$ connection over the spacetime tangent bundle $TM$, e.g. as
$w_j(TM)$ where $w_j$ is $j$-th Stiefel-Whitney (SW) class  \cite{milnor1974characteristic}.
We may also denote the probed background internal-symmetry/gauge connection over the principal bundle $E$, e.g. as
$w_j(E)=w_j(V_{{\mathbb{G}_{\text{internal}} }})$ where $w_j$ is also $j$-th SW class. In some cases, we may 
alternatively denote the latter as $w_j'(E)=w_j'(V_{{\mathbb{G}_{\text{internal}} }})$.
} 
the $\ltimes$ is a semi-direct product specifying a certain ``twisted'' operation (e.g. due to the symmetry extension from ${\mathbb{G}_{\text{internal}} }$ by ${G_{\text{spacetime} }}$)
and the ${N_{\text{shared}}}$ is the shared common normal subgroup symmetry between the two numerator groups.

The theories and their 't Hooft anomalies that we concern are in $d$d QFTs (4d YM and 2d $\CP^{\rN-1}$-model), 
but the topological/cobordism invariants are defined in the $D$d = $(d+1)$d manifolds.
The manifold generators for the bordism groups are actually the closed $D$d = $(d+1)$d manifolds.
We should clarify that although there can be 't Hooft anomalies for $d$d QFTs so ${\mathbb{G}_{\text{internal}} }$ may not be gauge-able, 
the SPTs/topological invariants defined in the closed $D$d = $(d+1)$d manifolds 
actually have ${\mathbb{G}_{\text{internal}} }$ always gauge-able
in that $D$d = $(d+1)$d.\footnote{This idea has been pursued to study the vacua of YM theories, for example, in \cite{2017arXiv171111587GPW} and 
references therein. See more explanations in \Sec{sec:con}'s \eq{eq:QFT-gauge-SPT}
}
This is related to the fact that in condensed matter physics, we say that the \emph{bulk} $D$d = $(d+1)$d SPTs
has an onsite local internal ${\mathbb{G}_{\text{internal}} }$-symmetry, thus this 
${\mathbb{G}_{\text{internal}} }$ must be gauge-able.

The \emph{new} ingredient in our present work slightly going beyond the cobordism theory of \cite{Freed2016} 
is that the ${\mathbb{G}_{\text{internal}} }$-symmetry may not only be an ordinary 0-form global symmetry,
but also include higher global symmetries.
The details of our calculation for such ``\emph{higher-symmetry-group cobordism theory}'' are provided in \cite{W2}.

Based on a theorem of Freed-Hopkin \cite{Freed2016} and an extended generalization that we propose \cite{W2},
there exists a 1-to-1 correspondence between  ``the invertible topological quantum field theories (iTQFTs) with symmetry (including higher symmetries)'' 
and ``a cobordism group.''
In condensed matter physics, this means that there is  
a 1-to-1 correspondence between ``the symmetric invertible topological order with symmetry (including higher symmetries)' that can be regularized on a lattice in its own dimensions' 
and ``a cobordism group,'' 
at least at lower dimensions.\footnote{
{We have used a mathematical fact that all smooth and differentiable manifolds are triangulable manifolds, based on Morse theory. 
On the contrary, triangulable manifolds are smooth manifolds at least for dimensions up to $D=4$ (i.e. the ``if and only if'' statement is true below $D\leq 4$). 
The concept of piecewise linear (PL) and smooth structures are equivalent in dimensions $D\leq 6$.
Thus all symmetric iTQFT classified by the cobordant properties of smooth manifolds have a triangulation (thus a lattice regularization) on a simplicial complex (thus a UV competition on a lattice). 
This implies a correspondence between ``the symmetric iTQFTs (on smooth manifolds)'' and ``the symmetric invertible topological orders (on triangulable manifolds)''
for $D\leq 4$.
See a recent application of this mathematical fact on the lattice regularization 
of symmetric iTQFTs and symmetric  invertible topological orders in \cite{Wang2018caiSM1809.11171} for various Standard Models of particle physics.
}}
More precisely, it is a 1-to-1 correspondence (isomorphism ``$\cong$'') between the following two well-defined ``mathematical objects'' (these ``objects'' turn out to be abelian groups):
\bea \label{eq:thm}
&&\left\{\begin{array}{ccc}\text{Deformation classes of reflection positive }\\\text{invertible } D
\text{-dimensional extended}\\
\text{topological field theories (iTQFT) with} \\
\text{symmetry group }
{\frac{{G_{\text{spacetime} }} \ltimes  {\mathbb{G}_{\text{internal}} }}{{N_{\text{shared}}}}}
\end{array}\right\}\nn \\
&&\cong[MT({\frac{{G_{\text{spacetime} }} \ltimes  {\mathbb{G}_{\text{internal}} }}{{N_{\text{shared}}}}}),\Sigma^{D+1}I\Z]_{\text{tors}}.
\eea
Let us explain the notation above:
 $MTG$ is the Madsen-Tillmann spectrum \cite{MadsenTillmann4} of the group $G$,
$\Sigma$ is the suspension, $I\Z$ is the Anderson dual spectrum, and ${\text{tors}}$ means taking only the finite group sector (i.e. the torsion group).

Namely, we classify the deformation classes of symmetric iTQFTs and also symmetric invertible topological orders (iTOs), via
this particular cobordism group
\bea
\Omega^{D}_{G} &\equiv&
\Omega^{D}_{({\frac{{G_{\text{spacetime} }} \ltimes  {\mathbb{G}_{\text{internal}} }}{{N_{\text{shared}}}}})} \nn\\
&\equiv&
\TP_D(G)\equiv[MTG,\Sigma^{D+1}I\Z].
\eea
 by classifying the cobordant relations of smooth, differentiable and triangulable
manifolds with a stable $G$-structure, via associating them to the homotopy
groups of Thom-Madsen-Tillmann spectra \cite{thom1954quelques,MadsenTillmann4},
given by a theorem in \Ref{Freed2016}. Here TP means the abbreviation of ``Topological Phases''
classifying the above symmetric iTQFT,
where our notations follow \cite{Freed2016} and \cite{W2}.
(For an introduction of the mathematical background for physicists, the readers can consult the Appendix A of \cite{2017arXiv171111587GPW}.)

Moreover, there are only the discrete/finite $\Z_n$-classes of the non-perturbative global 
't Hooft anomalies for YM and {$\CP^{\rN-1}$} model (so-called the \emph{torsion} group for $\Z_n$-class); 
there is no $\Z$-class perturbative anomaly (so-called the free class) for our QFTs.
So, we concern only the torsion group part of data in \eqn{eq:thm},
this is equivalent for us to simply look at the bordism group:
\bea
\Omega_{D}^{G} \equiv \Omega_{D}^{({\frac{{G_{\text{spacetime} }} \ltimes  {\mathbb{G}_{\text{internal}} }}{{N_{\text{shared}}}}})},
\eea
in order to classify all the 't Hooft anomalies for YM and {$\CP^{\rN-1}$} model.
  
Therefore, below we focus on the unoriented bordism groups and also some oriented bordism groups, replacing the orthogonal O group to a special orthogonal SO group:

Let $X$ be a fixed topological space, define the unoriented cobordism group
\bea
\Omega_D^{\tO}(X)&=&\{(M,f)\mid M\text{ is a closed }D\text{-manifold, }\nn\\
&&f:M\to X\text{ is a map}\}/\sim
\eea
$(M,f)\sim (M',f')$ if there exists a compact $(D+1)$-manifold $N$ and a map $h:N\to X$ such that $\partial N=M\sqcup M'$, and $h\mid_{M}=f$, $h\mid_{M'}=f'$.

Let $X$ be a fixed topological space, define the oriented cobordism group
\bea
\Omega_D^{\SO}(X)&=&\{(M,f)\mid M\text{ is a closed oriented }D\text{-manifold, }\nn\\
&&f:M\to X\text{ is a map}\}/\sim
\eea
$(M,f)\sim (M',f')$ if there exists a compact oriented $(D+1)$-manifold $N$ and a map $h:N\to X$ such that $\partial N=M\sqcup M'$, the orientations of $M$ and $M'$ are induced from that of $N$, and $h\mid_{M}=f$, $h\mid_{M'}=f'$.

Here $\sqcup$ is the disjoint union.
%

In particular, when $X=\B^2\Z_n$, $f:M\to \B^2\Z_n$ is a cohomology class in $\H^2(M,\Z_n)$.
When $X=\B G$, with $G$ is a Lie group or a finite group (viewed as a Lie group with discrete topology), then $f:M\to \B G$ is a principal $G$-bundle over $M$.
{To explain our notation, here $\B G$ is a classifying space of $G$, and $\B^2\Z_n$ is a higher classifying space (Eilenberg-MacLane space $K(\Z_n,2)$) of $\Z_n$.}\\

We have the following well-known facts:

\begin{itemize}
\item
 Unoriented cobordism groups are always $\Z_2$-vector spaces. 
 
\item

 $\Omega_D^{\SO}(X)$ is a subgroup of $\Omega_D^{\tO}(X)$ for $D\not\equiv0\mod4$. 
\end{itemize}

Our conventions in the following subsections are:
\begin{itemize}
\item
A map {between topological spaces} is always assumed to be continuous.
\item
For a top degree cohomology class with coefficients $\Z_2$, we often suppress an explicit integration over the manifold (i.e. pairing with the fundamental class $[M]$ with coefficients $\Z_2$), for example: 
$w_2(TM)w_3(TM)\equiv\int_Mw_2(TM)w_3(TM)$ where $M$ is a 5-manifold.

\item

{The group operation in cobordism group is induced from the disjoint union of manifolds. 
}

\item
 If $\Omega_D^{\tO}(X)=\Z_2^r$, then the group homomorphisms $\phi_i:\Omega_D^{\tO}(X)\to\Z_2$ for $1\le i\le r$ form a complete set of cobordism invariants of $\Omega_D^{\tO}(X)$ if $\phi=(\phi_1,\dots,\phi_r):\Omega_D^{\tO}(X)\to\Z_2^r$ is a group isomorphism. 
 
 \item
  The elements of $\Omega_D^{\tO}(X)$ are manifold generators if their images in $\Z_2^r$ under $\phi$ generate $\Z_2^r$.

\end{itemize}


In the following subsections, we consider the potential cobordism invariants/topological terms 
(5d and 3d [higher] SPTs for 4d YM and 2d $\CP^{\rN-1}$ model), and their manifold generators for bordism groups, as
the complete classification of all of their possible candidate higher 't Hooft anomalies.

First, we can convert the time reversal $\Z_2^{T'}$ ($\equiv \Z_2^{T}$ or $\Z_2^{CT}$) to the orthogonal O($D$)-symmetry group for such an underlying UV-completion of bosonic system (all gauge-invariant operators are bosons),
where the O($D$) is an extended symmetry group from SO($D$) via a short extension:
\bea
1 \to \SO(D) \to \tO(D) \to   \Z_2^{T'}  \to  1.
\eea
The $\SO(D)$ is the spacetime Euclidean rotational symmetry group for $D$d bosonic systems.\footnote{For the case of time-reversal symmetry, 
where there must be an underlying UV-completion of fermionic system (some gauge-invariant operators are fermions),
the more subtle time-reversal extension scenario is discussed in \cite{Freed2016} and \cite{2017arXiv171111587GPW}.}

Then we can easily list their converted full symmetry group $G$ and their relevant bordism groups, for 
SU(2) YM (\eq{eq:YM-SU2-Sym-T}/\eq{eq:YM-SU2-Sym}),
 SU(\rN) YM (\eq{eq:YM-Sym}/\eq{eq:YM-Sym-2}), 
$\CP^{1}$ model (\eq{eq:CP1-full-sym}/\eq{eq:CP1-full-sym-CT}),
and
$\CP^{\rN-1}$ model (\eq{eq:CPN-full-sym}),  
into the \eq{eq:Gall}'s form:
\begin{enumerate}[label=(\roman*)] 
\item
{$\Omega _5^{{\rm O}}(\mathrm {B}^2\mathbb {Z}_2)  \equiv \Omega _5^{({\rm O}\times \mathrm{B} \mathbb {Z}_2)}$}: 
This is the bordism group 
for $\Z_2^{CT} \times \Z_{2,[1]}^e$ in \eq{eq:YM-SU2-Sym} without $\Z_2^C$, which we will study in \Sec{sec:B},
here \eq{eq:Gall}'s $G={{\rm O}(D)\times \mathrm{B} \mathbb {Z}_2}$ or
$G={{\rm O}(D)\times \Z_{2,[1]}^e}$.
\item
{$\Omega _3^{{\rm O}}(\mathrm {B}{\rm O}(3))$}:
This is the bordism group for
${\Z_2^{CT} \times \tO(3)}$ in \eq{eq:CP1-full-sym-CT}, which we will study in \Sec{sec:D},
here \eq{eq:Gall}'s $G={{\rm O}(D)  \times \tO(3)}$.
\item
{$\Omega _5^{{\rm O}}(\mathrm {B}^2\mathbb {Z}_4)$}:
This is the bordism group 
for $\Z_2^{CT} \times (\Z_{\rN,[1]}^e \rtimes  \Z_2^C )$ in \eq{eq:YM-Sym-2} at N = 4 without $\Z_2^C$, which we will study in \Sec{sec:E},
here \eq{eq:Gall}'s $G={{\rm O}(D)\times \mathrm{B} \mathbb {Z}_4}$ or
$G={{\rm O}(D)\times \Z_{4,[1]}^e}$.
\item
{$\Omega _5^{{\rm O}}(\mathrm {B}\mathbb {Z}_2\ltimes \mathrm {B}^2\mathbb {Z}_4) \equiv \Omega _5^{({\rm O}\times (\mathbb {Z}_2 \ltimes \mathrm {B}\mathbb {Z}_4)) }$ and $\Omega _5^{{\rm O}_{\ltimes }}(\mathrm {B}\mathbb {Z}_2 \ltimes \mathrm {B}^2\mathbb {Z}_4) \equiv \Omega _5^{({\rm O}\times \mathbb {Z}_2) \ltimes \mathrm {B}\mathbb {Z}_4 }$}:\\
The first is the bordism group with a $CT$-time reversal,
for $\Z_2^{CT} \times (\Z_{\rN,[1]}^e \rtimes  \Z_2^C )$
 in \eq{eq:YM-Sym-2} at N = 4, which we will study in \Sec{sec:F},
here \eq{eq:Gall}'s $G={{\rm O}(D)\times (\mathbb {Z}_2 \ltimes \mathrm {B}\mathbb {Z}_4)}$ or
$G={{\rm O}(D)\times  (\Z_2^C  \ltimes \Z_{4,[1]}^e)}$.\\
The second is actually the re-written bordism group with a $\cT$-time reversal,
for $\Z_{\rN,[1]}^e \rtimes (\Z_2^T \times \Z_2^C)$
at N = 4, here \eq{eq:Gall}'s $G={({\rm O}(D)\times \mathbb {Z}_2) \ltimes \mathrm {B}\mathbb {Z}_4 }$ or
$G={({\rm O}(D)\times \mathbb {Z}_2^C) \ltimes \Z_{4,[1]}^e }$.
But we will \emph{not} study this, since it is simply a more complicated re-writing of the same result of \Sec{sec:F}.
\item
{$\Omega _3^{{\rm O}}(\mathrm {B}(\mathbb {Z}_2\ltimes {\rm PSU}(4)))$}:
This is the bordism group for
$\Z_2^{CT} \times (\PSU(\rN) \rtimes \Z_2^{C'})$ in \eq{eq:CPN-full-sym} at N = 4,  which we will study in \Sec{sec:G},
here \eq{eq:Gall}'s $G={{\rm O}(D)  \times (\PSU(\rN) \rtimes \Z_2^{C'})}$.
\end{enumerate}

Based on the relation between bordism groups 
and their $D$d = $(d+1)$d cobordism invariants to the $d$d anomalies of QFTs,
below we may simply abbreviate
{``5d cobordism invariants for characterizing 4d YM theory's anomaly''}
as 
$$\text{``5d (Yang-Mills) terms.''}$$
We may simply abbreviate 
{``3d cobordism invariants for characterizing 2d $\CP^{\rN-1}$ model's anomaly''}
as 
$$\text{``3d ($\CP^{\rN-1}$) terms.''}$$

\subsection{$\Omega_5^{\tO}(\B^2\Z_2)$}
\label{sec:B}

Follow \Sec{sec:A}, now we enlist all possible 't Hooft anomalies of 4d pure SU(2) YM at $\theta=\pi$ (but when the $\Z_2^C$-background field is turned off) by 
obtaining the 5d cobordism invariants from bordism groups of (\eq{eq:YM-SU2-Sym-T}/\eq{eq:YM-SU2-Sym}).

We are given a 5-manifold $M^5$ and a map $f: M^5\to\B^2\Z_2$.
Here the map $f: M^5\to\B^2\Z_2$ is the 2-form $B=B_2$ gauge field in the YM gauge theory \eq{eq:SU2YM} (and \eqn{eq:SUNYMN} at N = 2). 
We like to obtain 
the bordism invariants of $\Omega_5^{\tO}(\B^2\Z_2)$. We find the bordism group  \cite{W2}:\footnote{
Interestingly, the oriented version of the bordism group $\Omega_5^{\SO}(\B^2\Z_2)$ has also been studied recently in a different context in  \cite{Kapustin2017jrc1701.08264}. 
Here we study instead the unoriented bordism group $\Omega_5^{\tO}(\B^2\Z_2)$ new to the literature  \cite{W2}.
}
\bea
&&{\Omega_5^{\tO}(\B^2\Z_2)=\Z_2^4}, \\
&&\text{whose cobordism invariants are generated by}\nn\\
&&\left\{\begin{array}{l} 
B_2\cup\Sq^1B_2,\\ 
\Sq^2\Sq^1 B_2,\\ 
w_1(TM^5)^2\Sq^1 B_2,\\ 
w_2(TM^5)w_3(TM^5).
\end{array}\right.
\label{eq:bordism5OB2Z2}
\eea
Here $TM^5$ means the spacetime tangent bundle over $M^5$, see footnote \ref{footnote:wj}.
See \Ref{W2}, note that we derive on a 5d closed manifold,  
\bea
\Sq^2\Sq^1B_2&=&(w_2(TM^5)+w_1(TM^5)^2)\Sq^1B_2\nn\\
&=&(w_3(TM^5)+w_1(TM^5)^3)B_2, \\
w_1(TM^5)^2\Sq^1 B_2 &=&w_1(TM^5)^3B_2.
\eea 


We have a group isomorphism 
\bea
&\Phi_1:&\Omega_5^{\tO}(\B^2\Z_2)\to\Z_2^4\notag\\
&&(M^5,B_2)\mapsto(B_2\cup \Sq^1B_2,\Sq^2 \Sq^1B_2,\notag\\
&&w_1(TM^5)^2\cup \Sq^1B_2,w_2(TM^5)w_3(TM^5)).
\eea
\begin{enumerate} [label=\textcolor{blue}{\arabic*}., ref={\arabic*}]
\item
Consider $(M^5,B)=(\RP^2\times\RP^3,\alpha\cup\beta)$.
Here $\alpha$ is the generator of $\H^1(\RP^2,\Z_2)$, and $\beta$ is the generator of $\H^1(\RP^3,\Z_2)$.

Since 
\bea
&&\left\{\begin{array}{l} 
\Sq^1(\alpha\cup\beta)=\alpha^2\cup\beta+\alpha\cup\beta^2,\\ 
w_1(T(\RP^2\times\RP^3))=\alpha,\\ 
w_2(T(\RP^2\times\RP^3))=\alpha^2,\\ 
\text{and $w_3(T(\RP^2\times\RP^3))=0$},
\end{array}\right.
\eea
thus the $\Phi_1$ maps $(\RP^2\times\RP^3,\alpha\cup\beta)$ to $(1,0,0,0)$.

\item
Consider $(M^5,B)=(S^1\times\RP^4,\gamma\cup\zeta)$.
Here $\gamma$ is the generator of $\H^1(S^1,\Z_2)$, and $\zeta$ is the generator of $\H^1(\RP^4,\Z_2)$.

Since 
\bea
&&\left\{\begin{array}{l} 
\Sq^1(\gamma\cup\zeta)=\gamma\cup\zeta^2,\\
w_1(T(S^1\times\RP^4))=\zeta,\\ 
\text{and $w_2(T(S^1\times\RP^4))=0$,}
\end{array}\right.
\eea
 thus the $\Phi_1$ maps $(S^1\times\RP^4,\gamma\cup\zeta)$ to $(0,1,1,0)$.

\item
Consider $(M^5,B)=(S^1\times\RP^2\times\RP^2,\gamma\alpha_1)$.
Here $\gamma$ is the generator of $\H^1(S^1,\Z_2)$, and $\alpha_i$ is the generator of $\H^1(\RP^2,\Z_2)$ of the $i$-th factor $\RP^2$.
Since 
\bea
&&\left\{\begin{array}{l} 
\Sq^1(\gamma\alpha_1)=\gamma\alpha_1^2,\\ 
w_1(T(S^1\times\RP^2\times\RP^2))=\alpha_1+\alpha_2,\\ 
\text{and $w_2(T(S^1\times\RP^2\times\RP^2))=\alpha_1^2+\alpha_2^2+\alpha_1\alpha_2$,}
\end{array}\right.
\eea
thus the
$\Phi_1$ maps $(S^1\times\RP^2\times\RP^2,\gamma\alpha_1)$ to $(0,0,1,0)$.

\item
Let $\W$ be the Wu manifold $\SU(3)/\SO(3)$, 

Since $\Sq^1w_2(T\W)$ $=w_3(T\W)$,
thus the $\Phi_1$ maps $(\W,$ $w_2(T\W))$ to $(1,1,0,1)$, 
and $\Phi_1$ maps $(\W,0)$ to $(0,0,0,1)$.
\end{enumerate}

So we conclude that a generating set of manifold generators for $\Omega_5^{\tO}(\B^2\Z_2)$ is
\bea
{\{(M^5, B )\}}&=&\{(\RP^2\times\RP^3,\alpha\cup\beta),(S^1\times\RP^4,\gamma\cup\zeta),\notag\\
&&(S^1\times\RP^2\times\RP^2,\gamma\alpha_1),(\W,0)\}.
\eea
Note that $(\W,0)$ can be replaced by $(\W,w_2(T\W))$.

{
 The 4d Yang-Mills theory at $\theta=\pi$ has no 4d 't Hooft anomaly once the $\cT$ symmetry is not preserved. 
 This means that all 5d higher SPTs/cobordism invariant for 4d YM theory must vanish 
at $\Omega_5^{\SO}(\B^2\Z_2)$ when  $\cT$ (or $\cC\cT$) is removed.} 
{Compare with the data of $\Omega_5^{\SO}(\B^2\Z_2)$ given in Ref.~\cite{W2}, 
thus we find that the 5d terms (5d higher SPTs) for this 4d SU(2) YM are chosen among:
\bea \label{eq:select-N=2YM}
\boxed{
\left\{\begin{array}{l} 
B_2\cup\Sq^1B_2+\Sq^2\Sq^1 B_2,\\ 
w_1(TM^5)^2\Sq^1 B_2.
\end{array}\right.
}
\eea
}

This information will be used later to match the SU(2) YM anomalies at $\theta=\pi$.

\subsection{$\Omega_3^{\tO}(\B\tO(3))$}

\label{sec:D}

Follow \Sec{sec:A}, we enlist all possible 't Hooft anomalies of 2d $\CP^1$ model, or equivalently O(3) NLSM,  at $\theta=\pi$, by 
obtaining the 3d cobordism invariants from bordism groups of (\eq{eq:CP1-full-sym}/\eq{eq:CP1-full-sym-CT}).
From physics side, we will interpret the unoriented O($D$) spacetime symmetry with the time reversal from $\cC\cT$ instead of $\cT$.

We are given a 3-manifold $M^3$ and a map $f:M^3\to \B\tO(3)$.
Here the map $f:M^3\to \B\tO(3)$ is a principal $\tO(3)$ bundle
whose associated vector bundle is a rank 3 real vector bundle $E$ over $M^3$.

We like to obtain the bordism invariants of $\Omega_3^{\tO}(\B\tO(3))$. We compute the bordism group  \cite{W2}:
\bea
&&\Omega_3^{\tO}(\B\tO(3))=\Z_2^4, \\
&&\text{whose cobordism invariants are generated by}\nn\\
&&\left\{\begin{array}{l} 
w_1(E)^3,\\ 
w_1(E)w_2(E),\\ 
w_3(E),\\ 
w_1(E)w_1(TM^3)^2.
\end{array}\right.
\label{eq:bordism3OBO}
\eea

We have a group isomorphism 
\bea \label{eq:bordism3OBO-basis}
&\Phi_2:&\Omega_3^{\tO}(\B\tO(3))\to\Z_2^4\notag\\
&&(M^3,E)\mapsto(w_1(E)^3,w_1(E)w_2(E),\notag\\
&&w_3(E), w_1(E)w_1(TM^3)^2).  
\eea

Let $l_{\RP^n}$ denote the tautological line bundle over $\RP^n$ ($\RP^1=S^1$). If $x_n\in\H^1(\RP^n,\Z_2)$ denotes the generator, then $w(l_{\RP^n})=1+x_n$, $w(T\RP^n)=(1+x_n)^{n+1}$.

Let $\underline{n}$ denote the trivial real vector bundle of rank $n$, and let $+$ denote the direct sum.

By the Whitney sum formula, $w(E\oplus F)=w(E)w(F)$. Here $w(E)=1+w_1(E)+w_2(E)+\cdots$ is the total Stiefel-Whitney class of $E$.
Then we find:
\begin{enumerate}
\item
Since $w(3l_{\RP^3})=(1+x_3)^3=1+x_3+x_3^2+x_3^3$, and $w_1(T\RP^3)=0$,
thus the $\Phi_2$ maps $(\RP^3,3l_{\RP^3})$ to $(1,1,1,0)$.
\item
Since $w(l_{\RP^3}+\underline{2})=1+x_3$, thus the $\Phi_2$ maps 
$(\RP^3,l_{\RP^3}+\underline{2})$ to $(1,0,0,0)$.
\item
Since $w(l_{S^1}+\underline{2})=1+x_1$, and $w_1(T(S^1\times\RP^2))=x_2$, thus the $\Phi_2$ maps
$(S^1\times\RP^2,l_{S^1}+\underline{2})$ to $(0,0,0,1)$.

\item
Since $w(l_{S^1}+l_{\RP^2}+\underline{1})=(1+x_1)(1+x_2)=1+x_1+x_2+x_1x_2$, thus the  $\Phi_2$ maps
$(S^1\times\RP^2,l_{S^1}+l_{\RP^2}+\underline{1})$ to $(1,1,0,1)$.

\end{enumerate}

So a generating set of manifold generators for $\Omega_3^{\tO}(\B\tO(3))$ is 
\bea
{\{(M^3, E )\}}&=&\{(\RP^3,3l_{\RP^3}),(\RP^3,l_{\RP^3}+2),\notag\\
&&(S^1\times\RP^2,l_{S^1}+2),\notag\\
&&(S^1\times\RP^2,l_{S^1}+l_{\RP^2}+1)\}.
\eea

{Note that $(S^1\times\RP^2,l_{S^1}+2l_{\RP^2})$ is also a manifold generator.
Note $w(l_{S^1}+2l_{\RP^2})=(1+x_1)(1+x_2)^2=1+x_1+x_2^2+x_1x_2^2$, 
 therefore $\Phi_2$ maps $(S^1\times\RP^2,l_{S^1}+2l_{\RP^2})$ to $(0,1,1,1)$.}

\subsection{$\Omega_5^{\tO}(\B^2\Z_4)$}

\label{sec:E}

Follow \Sec{sec:A}, now we enlist all possible 't Hooft anomalies of 4d pure SU(4) YM at $\theta=\pi$ (but when the $\Z_2^C$-background field is turned off) by 
obtaining the 5d cobordism invariants from bordism groups of (\eq{eq:YM-Sym}/\eq{eq:YM-Sym-2}).

We are given a 5-manifold $M^5$ and a map $f: M^5\to\B^2\Z_4$.
Here the map $f: M^5\to\B^2\Z_4$ is the 2-form $B=B_2$ gauge field in the YM gauge theory \eq{eq:SU2YM} (and \eqn{eq:SUNYMN} at N = 4).

We compute 
the bordism invariants of $\Omega_5^{\tO}(\B^2\Z_4)$, we find the bordism group  \cite{W2}:
\bea
&&{\Omega_5^{\tO}(\B^2\Z_4)=\Z_2^4}, \\
&&\text{whose cobordism invariants are generated by}\nn\\
&&\left\{\begin{array}{l} 
B_2\cup \beta_{(2,4)} B_2,\\ 
\Sq^2 \beta_{(2,4)}  B_2,\\ 
w_1(TM^5)^2\beta_{(2,4)}B_2,\\ 
w_2(TM^5)w_3(TM^5).
\end{array}\right.
\eea
where $\beta_{(2,4)}:\H^*(M^5,\Z_4)\to\H^{*+1}(M^5,\Z_2)$ is the Bockstein homomorphism associated with the extension $\Z_2\to\Z_8\to\Z_4$ (see Appendix \ref{Bockstein}).

We have a group isomorphism 
\bea
&\Phi_3:&\Omega_5^{\tO}(\B^2\Z_4)\to\Z_2^4\notag\\
&&(M^5,B_2)\mapsto(B_2\cup\beta_{(2,4)}B_2,\Sq^2\beta_{(2,4)}B_2,\notag\\
&&w_1^2(TM^5)\beta_{(2,4)}B_2,w_2(TM^5)w_3(TM^5)).
\eea

Let $K$ be the Klein bottle. 
\begin{enumerate}

\item

Let $\alpha'$ be the generator of $\H^1(S^1,\Z_4)$, $\beta'$ be the generator of the $\Z_4$ factor of $\H^1(K,\Z_4)=\Z_4\times\Z_2$ (see Appendix \ref{Klein}),
and $\gamma'$ be the generator of $\H^2(S^2,\Z_4)$. {Note that} $\beta_{(2,4)}\beta'=\sigma$ where $\sigma$ is the generator of $\H^2(K,\Z_2)=\Z_2$ (see Appendix \ref{Klein}).

Since $\beta_{(2,4)}(\alpha'\cup\beta'+\gamma')=\alpha'\cup\sigma$ and $w_2(T(S^1\times K\times S^2))=w_1(T(S^1\times K\times S^2))^2=0$,
we find that $\Phi_3$ maps
$(S^1\times K\times S^2,\alpha'\cup\beta'+\gamma')$ to $(1,0,0,0)$.

\item
Following the notation of \cite{barden1965simply}, $X_2$ is a simply-connected 5-manifold which is orientable but non-spin.
Let $\theta'$ and $\eta'$ be two generators of $\H^2(X_2,\Z_4)=\Z_4^2$, {then} $\beta_{(2,4)}\theta'$ is one of the two generators of $\H^3(X_2,\Z_2)=\Z_2^2$.
Since $w_2(TX_2)=(\theta'+\eta')\mod2$, $w_1(TX_2)=0$ and $w_3(TX_2)=0$,
we find that $\Phi_3$ maps $(X_2,\theta')$ to $(1,1,0,0)$.

\item
Since $w_1(T(S^1\times K\times\RP^2))^2=w_2(T(S^1\times K\times\RP^2))=\alpha^2$ where $\alpha$ is the generator of $\H^1(\RP^2,\Z_2)$, 
we find that $\Phi_3$ maps
$(S^1\times K\times\RP^2,\alpha'\cup\beta')$ to $(0,0,1,0)$.

\item
$\W$ is the Wu manifold, while $\Phi_3$ maps $(\W,0)$ to $(0,0,0,1)$.

\end{enumerate}

So a generating set of manifold generators for $\Omega_5^{\tO}(\B^2\Z_4)$ is
\bea
{\{(M^5, B )\}}&=&\{(S^1\times K\times S^2,\alpha'\cup\beta'+\gamma'),
(X_2,\theta'),
\notag\\
&&(S^1\times K\times\RP^2,\alpha'\cup\beta'),(\W,0)\}.
\eea

Note that
\begin{enumerate}
\item
$(S^1\times K\times T^2,\alpha'\cup\beta'+\zeta')$ is also a generator where $\zeta'$ is the generator of $\H^2(T^2,\Z_4)$.
Since $\beta_{(2,4)}(\alpha'\cup\beta'+\zeta')=\alpha'\cup\sigma$ and $w_2(T(S^1\times K\times T^2))=w_1(T(S^1\times K\times T^2))^2=0$,
we find $\Phi_3$ maps
$(S^1\times K\times T^2,\alpha'\cup\beta'+\zeta')$ to $(1,0,0,0)$.

\item
$(K\times S^3/\Z_4,\beta'\cup\epsilon'+\phi')$ is also a generator where $S^3/\Z_4$ is the Lens space $L(4,1)$, $\epsilon'$ is the generator of $\H^1(S^3/\Z_4,\Z_4)$, $\phi'$ is the generator of $\H^2(S^3/\Z_4,\Z_4)$.
Since $\beta_{(2,4)}(\beta'\cup\epsilon'+\phi')=\sigma\cup\epsilon'+\beta'\cup\phi$ where $\phi$ is the generator of $\H^2(S^3/\Z_4,\Z_2)$, 
and $w_2(T(K\times S^3/\Z_4))=w_1(T(K\times S^3/\Z_4))^2=0$, 
we find that
$\Phi_3$ maps
$(K\times S^3/\Z_4,\beta'\cup\epsilon'+\phi')$ to $(1,0,0,0)$.
\end{enumerate}

\subsection{{$\Omega_5^{\tO}(\B\Z_2\ltimes\B^2\Z_4) \equiv \Omega_5^{(\tO \times (\Z_2 \ltimes \B\Z_4)) }$} 
}

\label{sec:F}

Follow \Sec{sec:A}, now we enlist all possible 't Hooft anomalies of 4d pure SU(4) YM at $\theta=\pi$ (when the $\Z_2^C$-background field can be turned on) by 
obtaining the 5d cobordism invariants from bordism groups of (\eq{eq:YM-Sym}/\eq{eq:YM-Sym-2}).

Note that again from physics side, we will interpret the unoriented O($D$) spacetime symmetry with the time reversal from $\cC\cT$ instead of $\cT$.
So we choose the former $\Omega_5^{\tO}(\B\Z_2\ltimes\B^2\Z_4) \equiv \Omega_5^{(\tO \times (\Z_2 \ltimes \B\Z_4)) }$ for $\cC\cT$, 
rather than the more complicated latter $\Omega_5^{\tO_{\ltimes}}(\B \Z_2 \ltimes \B^2\Z_4) \equiv \Omega_5^{(\tO \times \Z_2) \ltimes \B\Z_4 }$ for  $\cT$. 

Before we dive into $\Omega_5^{\tO}(\B\Z_2\ltimes\B^2\Z_4) \equiv \Omega_5^{(\tO \times (\Z_2 \ltimes \B\Z_4)) }$, we first study the simplified
``untwisted'' bordism group  $\Omega_5^{\tO}(\B\Z_2\times\B^2\Z_4)$.

We are given a 5-manifold $M^5$ and a 1-form field $A:M\to \B\Z_2$ and a 2-form $B=B_2:M^5\to\B^2\Z_4$ gauge field in the YM gauge theory \eq{eq:SU2YM} (and \eqn{eq:SUNYMN} at N = 4). 
We compute 
the bordism invariants of $\Omega_5^{\tO}(\B\Z_2\times\B^2\Z_4)$, we find the bordism group \cite{W2}
\bea
&&{\Omega_5^{\tO}(\B\Z_2\times\B^2\Z_4)=\Z_2^{12}}, \\
&&\text{whose cobordism invariants are generated by}\nn\\
&&\left\{\begin{array}{l} 
B_2\cup\beta_{(2,4)}B_2,\\ 
\Sq^2\beta_{(2,4)} B_2,\\ 
w_1(TM^5)^2\beta_{(2,4)} B_2,\\ 
w_2(TM^5)w_3(TM^5),\\
A^5,\;A^2\beta_{(2,4)}B_2,\\
A^3B_2,\;A^3w_1(TM^5)^2,\\
AB_2^2,\;Aw_1(TM^5)^4,\\
AB_2w_1(TM^5)^2,\;Aw_2(TM^5)^2.
\end{array}\right.
\eea

We also compute 
the bordism invariants of $\Omega_5^{\SO}(\B\Z_2\times\B^2\Z_4)$, we find  \cite{W2}
\bea
&&{\Omega_5^{\SO}(\B\Z_2\times\B^2\Z_4)=\Z_2^6}, \\
&&\text{whose cobordism invariants are generated by}\nn\\
&&\left\{\begin{array}{l} 
\Sq^2\beta_{(2,4)} B_2,\\ 
w_2(TM^5)w_3(TM^5),\\
A^5,\;A^3B_2,\\
AB_2^2,\;Aw_2(TM^5)^2.
\end{array}\right.
\eea

{
 The 4d Yang-Mills theory at $\theta=\pi$ has no 4d 't Hooft anomaly once the $\cC\cT$ (or $\cT$) symmetry is not preserved 
(as we discussed before that $\cC$-symmetry is a good symmetry for
any $\theta$ which has no anomaly directly from mixing with $\cC$ by its own). This means that
all 5d higher SPTs/cobordism invariant for 4d YM theory must vanish 
at $\Omega_5^{\SO}(\B\Z_2\times\B^2\Z_4)$ when $\cC\cT$ (or $\cT$) is removed. 
So the 5d SPTs for this 4d YM are chosen among:
\bea \label{eq:select-N=4YM}
\boxed{\left\{\begin{array}{l} 
B_2\cup\beta_{(2,4)}B_2,\\ 
w_1(TM^5)^2\beta_{(2,4)} B_2,\\ 
A^2\beta_{(2,4)}B_2,\;A^3w_1(TM^5)^2,\\
Aw_1(TM^5)^4,\;AB_2w_1(TM^5)^2.
\end{array}\right.}
\eea
}

Let $\alpha'$ be the generator of $\H^1(S^1,\Z_4)$, $\beta'$ be the generator of the $\Z_4$ factor of $\H^1(K,\Z_4)=\Z_4\times\Z_2$ (see Appendix \ref{Klein}),
$\gamma'$ be the generator of $\H^2(S^2,\Z_4)$. Note
$\beta_{(2,4)}\beta'=\sigma$ where $\sigma$ is the generator of $\H^2(K,\Z_2)=\Z_2$ (see Appendix \ref{Klein}).
Let $\alpha$ be the generator of $\H^1(\RP^2,\Z_2)$, $\beta$ be the generator of $\H^1(\RP^3,\Z_2)$, $\gamma$ be the generator of $\H^1(S^1,\Z_2)$.

Then a generating set of manifold generators for the Yang-Mills terms is
\bea
{\{(M^5, A, B )\}}&=&\{(S^1\times K\times S^2,A=0,B=\alpha'\cup\beta'+\gamma'),\notag\\
&&(S^1\times K\times\RP^2,A=0,B=\alpha'\cup\beta'),\notag\\
&&(S^1\times K\times\RP^2,A=\alpha,B=\alpha'\cup\beta'),\notag\\
&&(\RP^2\times\RP^3,A=\beta,B=0),\notag\\
&& (S^1\times\RP^4,A=\gamma,B=0),\notag\\
&&(S^1\times S^2\times\RP^2,A=\gamma,B=\gamma')\}.
\eea

{Now we discuss this group,}
{$\Omega_5^{\tO}(\B\Z_2\ltimes\B^2\Z_4) \equiv \Omega_5^{(\tO \times (\Z_2 \ltimes \B\Z_4)) }$,
here $\Z_2$ acts nontrivially on $\Z_4$.
We have a fibration
\bea
\xymatrix{
\B^2\Z_4\ar[r]&\B(\Z_2 \ltimes \B\Z_4)\ar[d]\\
&\B\Z_2}
\eea
which has the nontrivial Postnikov class in $\H^3(\B\Z_2,\Z_4)=\Z_2$.
\\
$B\in\H^2(M^5,\Z_{4,A})$ which is the twisted cohomology where $A\in\H^1(M^5,\Z_2)$ can be viewed as a group homomorphism $\pi_1(M^5)\to\text{Aut}(\Z_4)=\Z_2$. 
\\
We claim that among the candidates of the 5d higher SPTs/cobordism invariants 
for 4d SU(4) Yang-Mills theory at $\theta=\pi$, no one can vanish in $\Omega_5^{\tO}(\B\Z_2\ltimes\B^2\Z_4)$ (see Appendix \ref{BZ2B2Z4} and  \cite{W2}).
{Namely, we obtain that $\Omega_5^{\tO}(\B\Z_2\ltimes\B^2\Z_4)=\Z_2^{11}$, where only the $A^3B_2$ term is dropped, 
compared with $\Omega_5^{\tO}(\B\Z_2\times\B^2\Z_4)$.}
}

{
We have a group isomorphism 
\bea
&\Phi_3':&\Omega_5^{\tO}(\B\Z_2\ltimes\B^2\Z_4)\to\Z_2^{11}\notag\\
&&(M^5,A,B_2)\mapsto(B_2\cup\beta_{(2,4)}B_2, AB_2w_1(TM^5)^2,\nn\\
&&w_1(TM^5)^2\beta_{(2,4)} B_2, A^2\beta_{(2,4)}B_2, \nn\\
&&A^3w_1(TM^5)^2, Aw_1(TM^5)^4,\nn\\
&&
\Sq^2\beta_{(2,4)} B_2, w_2(TM^5)w_3(TM^5),\nn\\
&&
A^5,
AB_2^2,Aw_2(TM^5)^2).
\eea
}
This group isomorphism $\Phi_3'$ will be used in
\Sec{sec:5d3d}.

\subsection{$\Omega_3^{\tO}(\B(\Z_2\ltimes\PSU(4)))$}

\label{sec:G}

Follow \Sec{sec:A}, we enlist all possible 't Hooft anomalies of 2d $\CP^{\rN-1}$ model at N = 4,  at $\theta=\pi$, by 
obtaining the 3d cobordism invariants from bordism groups of (\eq{eq:CPN-full-sym}).
From physics side, we will interpret the unoriented O($D$) spacetime symmetry with the time reversal from $\cC\cT$ instead of $\cT$.

We are given a 3-manifold $M^3$ and a map $f:M^3\to \B(\Z_2\ltimes\PSU(4))$ which corresponds to a principal $\Z_2\ltimes\PSU(4)$ bundle $E$ over $M^3$.

We compute the bordism invariants of $\Omega_3^{\tO}(\B\tO(3))$, we find the bordism group  \cite{W2}:
\bea
&&\Omega_3^{\tO}(\B(\Z_2\ltimes\PSU(4)))=\Z_2^4, \\
&&\text{whose cobordism invariants are generated by}\nn\\
&&\left\{\begin{array}{l} 
w_1(E)^3,\\ 
w_1(E)w_1(TM^3)^2,\\ 
\beta_{(2,4)}w_2(E),\\ 
w_1(E)(w_2(E)\mod2).
\end{array}\right.
\eea
where $E$ is a principal $\Z_2\ltimes\PSU(4)$ bundle over $M^3$ which is a pair
$(w_1(E),w_2(E))\in \H^1(M^3,\Z_2)\times \H^2(M^3,\Z_{4,w_1(E)})$ where $\H^2(M^3,\Z_{4,w_1(E)})$ is the twisted cohomology, $w_1(E)$ can be viewed as a group homomorphism $\pi_1(M^3)\to\text{Aut}(\Z_4)=\Z_2$. 

{In the following discussion, we use the ordinary cohomology instead of the twisted cohomology. 
If $w_1(E)=0$ or $w_2(E)=0$, then this simplification has no effect, while for the term $w_1(E)(w_2(E)\mod2)$, though both $w_1(E)$ and $w_2(E)$ may be nonzero, but we will see that for certain 3-manifold $M^3$, for example, $M^3=S^1\times T^2$, if the action of $\pi_1(M^3)=\Z^3$ on $\Z_4$ is  nontrivial on only one factor $\Z$, namely, $w_1(E)=\gamma$, then $w_2(E)$ may not be twisted by $w_1(E)$, for example, $w_2(E)=\zeta'$, then the ordinary cohomology is sufficient for our discussion.
}

We have a group isomorphism 
\bea
&\Phi_4:&\Omega_3^{\tO}(\B(\Z_2\ltimes\PSU(4)))\to\Z_2^4\notag\\
&&(M^3,w_1(E),w_2(E))\mapsto(w_1(E)^3,w_1(E)(w_2(E)\mod2)   ,\notag\\
&&\beta_{(2,4)}w_2(E),w_1(E)w_1(TM^3)^2). 
\eea

\begin{enumerate}
\item
Recall that $\beta$ is the generator of $\H^1(\RP^3,\Z_2)$.
Since $w_1(T\RP^3)=0$,
$\Phi_4$ maps 
$(\RP^3,\beta,0)$ to $(1,0,0,0)$.

\item
Recall that $\gamma$ is the generator of $\H^1(S^1,\Z_2)$, $\gamma'$ is the generator of $\H^2(S^2,\Z_4)$.
$\Phi_4$ maps
$(S^1\times S^2,\gamma,\gamma')$ to $(0,1,0,0)$.

\item
$K$ is the Klein bottle.
Recall that $\alpha'$ is the generator of $\H^1(S^1,\Z_4)$, $\beta'$ is the generator of the $\Z_4$ factor of $\H^1(K,\Z_4)=\Z_4\times\Z_2$ (see Appendix \ref{Klein}).
Since $\beta_{(2,4)}(\alpha'\cup\beta')=\alpha'\cup\sigma$ where $\sigma$ is the generator of $\H^2(K,\Z_2)$, $\Phi_4$ maps 
$(S^1\times K,0,\alpha'\cup\beta')$ to $(0,0,1,0)$. 

\item 
Recall that $\gamma$ is the generator of $\H^1(S^1,\Z_2)$.
Since $w_1(T(S^1\times\RP^2))=\alpha$ where $\alpha$ is the generator of $\H^1(\RP^2,\Z_2)$,
$\Phi_4$ maps 
$(S^1\times\RP^2,\gamma,0)$ to $(0,0,0,1)$.

\end{enumerate}

So a generating set of manifold generators for $\Omega_3^{\tO}(\B(\Z_2\ltimes\PSU(4)))$ is 
\bea
{\{(M^3, w_1(E),  w_2(E) )\}}&=&\{(\RP^3,\beta,0),(S^1\times S^2,\gamma,\gamma') ,\notag\\
&&(S^1\times K,0,\alpha'\cup\beta'),\notag\\
&&  (S^1\times\RP^2,\gamma,0)\}.
\eea

Note that 
$(S^1\times T^2,\gamma,\zeta')$ is also a manifold generator, where $\gamma$ is the generator of $\H^1(S^1,\Z_2)$, $\zeta'$ is the generator of $\H^2(T^2,\Z_4)$.
$\Phi_4$ maps $(S^1\times T^2,\gamma,\zeta')$ to $(0,1,0,0)$.

\section{Review and Summary of Known Anomalies via Cobordism Invariants} 
\label{sec:rev}

Follow \Sec{sec:cobo-top}, we have obtained the co/bordism groups relevant from the given full $G$-symmetry of 4d YM and 2d $\CP^{\rN-1}$ models.
Therefore, 
based on the correspondence between $d$d 't Hooft anomalies and $D$d=$(d+1)$d topological terms/cobordism/SPTs invariants,
we have obtained the classification of all possible higher 't Hooft anomalies for these 4d YM and 2d $\CP^{\rN-1}$ models.

Below we first match our result to the known anomalies found in the literature, and we shall put these 
known anomalies into a more mathematical precise thus a more general framework, under the cobordism theory.
We will write down the precise $d$d 't Hooft anomalies and $D$d=$(d+1)$d cobordism/SPTs invariants for them.
We will also clarify the physical interpretations (e.g. from condensed matter inputs) of anomalies.

\subsection{Mixed higher-anomaly of time-reversal $\Z_2^{CT}$ and 1-form center $\Z_\rN$-symmetry of SU($\rN$)-YM theory} 

First recall in 
\Sec{sec:mix-T-1}, we re-derives the mixed higher-anomaly 
of time-reversal $\Z_2^T$ and 1-form center $\Z_\rN$-symmetry of 4d SU($\rN$)-YM, at even $\rN$, 
discovered in [\onlinecite{Gaiotto2017yupZoharTTT}]. 
By turning on 2-form $\Z_{\rN}$-background field $B=B_2$ coupling to YM theory,
the $\Z_2^T$-symmetry shifts the 4d YM with an additional 5d higher SPTs term \eq{eq:4d-YM-5d-SPTs}.
We also learned that the same mixed higher-anomaly occur by replacing $\Z_2^T$ to \eq{eq:T-CT-mix},
$$
\text{$\Z_2^{CT}$, $\Z_2^{P}$, and $\Z_2^{CP}$}, 
$$
For our preference, we focus on ${\cC\cT}$ instead of ${\cT}$.
This type of anomaly has the linear dependence on ${\cC\cT}$ (thus linear also $\cT$) and quadratic dependence on $B_2$.
Compare with our \eq{eq:bordism5OB2Z2}, 
we find that the precise form for 5d cobordism invariant/ 4d higher 't Hooft anomaly is:
\bea \label{eq:BSq1B-1}
\boxed{B_2\Sq^1B_2}.
\eea
{More precisely, we need to consider instead \eq{eq:w1PB},
$B_2\Sq^1B_2+\Sq^2\Sq^1B_2 {= \frac{1}{2} \tilde w_1(TM) \cP_2(B_2)}$, see
\Sec{sec:newSU2YM} for details and derivations.}

\begin{widetext}

\begin{figure}[!h]
\centering
\includegraphics[scale=1.1]{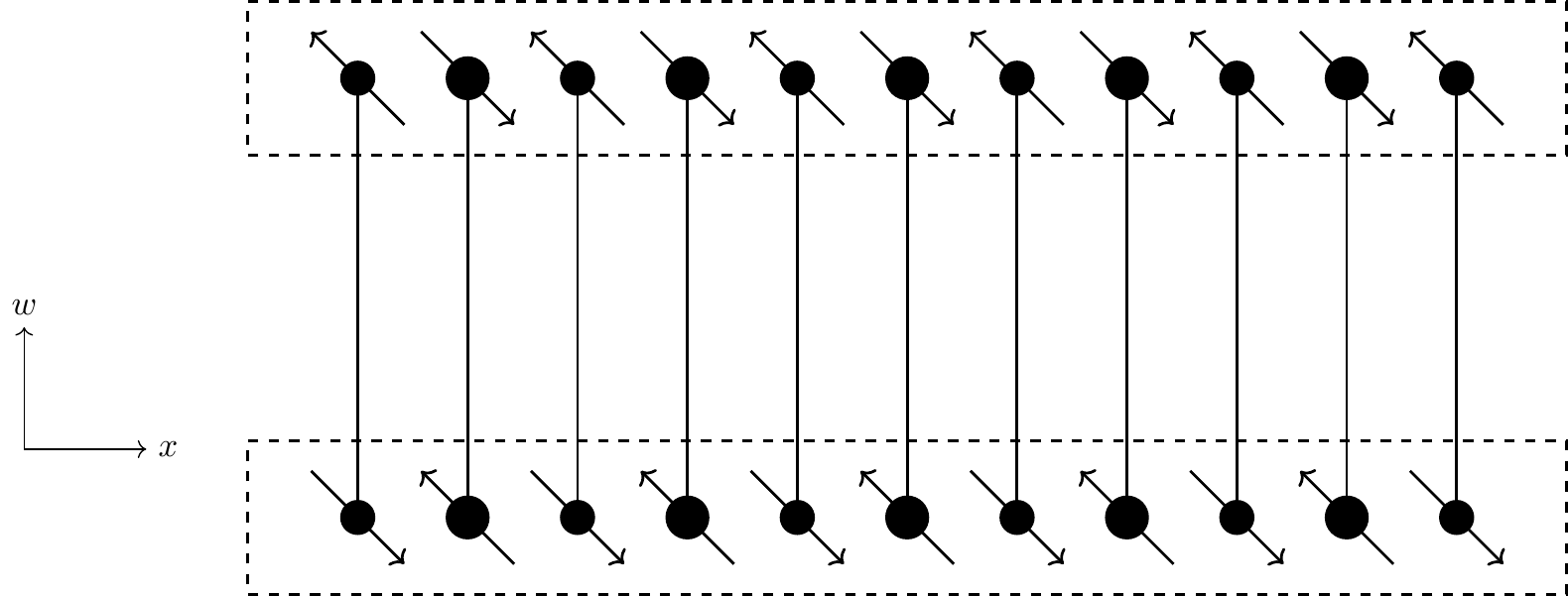}
\caption{An interpretation of 't Hooft anomaly ${w_1(E)({w_2(V_{\mathrm{SO}(3)})}+{w_1(TM)^2})}$  (\eq{eq:anom-stacking-Haldane})
for 2d  $\mathbb{CP}^1$ or O(3) NLSM model, is obtained from the 2d dangling spin-1/2 gapless modes living on the 2d boundary 
a 3d layer-stacking system of the 2d spin-1 Haldane chain.
Each vertical solid line represents a  2d spin-1 Haldane chain \eq{eq:Haldane}.
The 2d boundary combined from the dangling spin-1/2 gapless modes, encircled by the dashed-line rectangle, on the top and on the bottom,
are effectively the 2d $\mathbb{CP}^1$ model.
The ${w_1(E){w_2(V_{\mathrm{SO}(3)})}}$ has been identified by \cite{Metlitski2017fmd1707.07686}.
The same trick of \cite{Metlitski2017fmd1707.07686} applies to \eq{eq:Haldane}
 teaches us the more complete anomaly \eq{eq:anom-stacking-Haldane}.
The $\mathbb{Z}_2^x \equiv \mathbb{Z}_2^C$-translational nature of Heisenberg anti-ferromagnet (AFM) is purposefully emphasized by the
spin-1/2 orientation ($\up$ or $\down$) and the two different sizes of the black dots ($\bullet$).
Notice that understand the 3d bulk nature obtained from this stacking, inform us that
 the 3d bulk SPTs also includes a secretly hidden $A_x \cup A_x \cup A_x \equiv A_x^3$ or the
 ${w_1(E)^3 \equiv w_1(\mathbb{Z}_2^x)^3}$-anomaly \cite{WangSantosWen1403.5256, Metlitski2017fmd1707.07686}.} 
\label{fig:stack}
\end{figure}
\end{widetext}

\subsection{Mixed anomaly of $\Z_2^C=\Z_2^x$- and time-reversal $\Z_2^{CT}$ or SO(3)-symmetry of $\mathbb{CP}^{1}$-model}
\label{sec:CTSO3}

Now we move on to 
2d $\CP^1$ or O(3) NLSM model at $\theta=\pi$, we get the full
0-form global symmetries 
\eq{eq:CP1-full-sym-CT},
$\Z_2^{CT} \times \PSU(2) \times \Z_2^{x}$
$\equiv$ $\Z_2^{CT} \times \PSU(2) \times \Z_2^{C}$
$=$ $\Z_2^{CT} \times \tO(3)$.

It has been known that there is a non-perturbative global discrete anomaly from the $\Z_2^C$ (a discrete translational $\Z_2^x$ symmetry)
since the work of Gepner-Witten \cite{GepnerWitten1986wi}.
More recently, this non-perturbative global discrete anomaly
has been revisited by \cite{FuruyaOshikawa2015coaO, Yao2018kelHsiehOshikawa1805.06885} 
to understand the nature of symmetry-protected gapless critical phases.

We can compare this anomaly (associated with $\Z_2^x$ symmetry and to the $\PSU(2)$-symmetry)
to the 3d cobordism invariant/ 2d 't Hooft anomaly we derive in \eq{eq:bordism3OBO}.
We find that $w_1(E)w_2(V_{\SO(3)})$,
where $w_1(E)=w_1(V_{\tO(3)})=w_1(\Z_2^x)$,
is the natural choice to describe the anomaly.

Ref.~\cite{SulejmanpasicTanizaki1802.02153}, 
detects a so-called mixed $\cC$$\cP$$\cT$-type anomaly.
We can interpret their anomaly as the mixed
$\cC$ ($\Z_2^C=\Z_2^x$) with the $\cC$$\cT$ ($\Z_2^{CT}$)
type anomaly.
We compare it
to the 3d cobordism invariant/ 2d 't Hooft anomaly that we derived in \eq{eq:bordism3OBO},
and find
{$w_1(E)w_1(TM)^2=w_1(\Z_2^x)w_1(TM)^2$} is the natural choice to describe this anomaly. 

So overall, compare with \eq{eq:bordism3OBO}, we can interpret the above 2d anomalies are captured by
a 3d cobordism invariant for N = 2 case:
\bea \label{eq:anom-stacking-Haldane}
\boxed{{w_1(E)w_2(V_{\SO(3)}) }+{w_1(E)w_1(TM)^2}}.
\eea 

A very natural physics derivation to understand \eq{eq:anom-stacking-Haldane} is by the stacking 2d 
Haldane spin-1 chain picture \cite{Metlitski2017fmd1707.07686}, see \Fig{fig:stack}.
The Haldane spin-1 chain is a 2d SPTs protected by  
spin-1 rotation SO(3) symmetries and time-reversal (here $\Z_2^{CT}$); its 2d SPTs/topological term is well-known as:
\bea  \label{eq:Haldane}
\int_{\text{2d spin-1 chain}} w_2(V_{\SO(3)}) + w_1(TM)^2, 
\eea
obtained from group cohomology data $\H^2(\B \SO(3), \U(1))$ $=\Z_2$ and $\H^2(\B \Z_2^T, \U(1))$ $=\Z_2$ \cite{Chen2011pg1106.4772}. 
If the time-reversal or SO(3) symmetry is preserved, the boundary has 2-fold degenerate spin-1/2 modes on each 1d edge.
The layer stacking of such spin-1/2 modes to a 2d boundary (encircled by the dashed-line rectangle in \Fig{fig:stack}) 
can actually give rise to gapless 2d $\CP^1$ / O(3) NLSM / $\SU(2)_1$-WZW model.
Part of its anomaly is captured by the $\Z_2^x$-translation ($w_1(E)=w_1(\Z_2^x)$) times the
\eq{eq:Haldane}, which renders and thus we derive \eq{eq:anom-stacking-Haldane}.

\subsection{$w_3(E)$ anomaly of $\mathbb{CP}^{1}$-model}

Ref.~\cite{Komargodski2017dmc1705.04786} 
studies the anomaly of the same system,
and detects the anomaly $w_3(E)=w_3(V_{\tO(3)})$, 
we can convert it to
\bea \label{eq:w3E}
&&{\boxed{w_3(E)}=w_3(V_{\tO(3)})}\\
&&=w_1(V_{\tO(3)})^3+w_1(V_{\tO(3)})w_2(V_{\SO(3)})+w_3(V_{\SO(3)}) \nn\\
&&{=w_1(\Z_2^x)^3+w_1(\Z_2^x)w_2(V_{\SO(3)})+w_3(V_{\SO(3)}) } \nn\\
&&{=w_1(E)^3+w_1(E)w_2(V_{\SO(3)})+w_3(V_{\SO(3)}) }\nn\\
&&{=w_1(E)w_2(E) +w_1(TM)w_2(E) }\nn\\
&&\boxed{=w_1(E)^3+w_1(E)w_2(V_{\SO(3)})+w_1(TM)w_2(V_{\SO(3)}) }.\nn
\eea
We also note that 
\bea
&&w_1(E)w_2(E)=w_1(V_{\tO(3)})w_2(V_{\tO(3)}) \nn\\
&&=w_1(V_{\tO(3)})^3+w_1(V_{\tO(3)})w_2(V_{\SO(3)}) \nn\\
&&=w_1(\Z_2^x)^3+w_1(\Z_2^x)w_2(V_{\SO(3)}) \nn\\
&&=w_1(E)^3+w_1(E)w_2(V_{\SO(3)}). 
\eea
{Similar equality and anomaly are discussed in \cite{Cordova2017vab1711.10008} in a different topic on Chern-Simons matter theories.}

To summarize, we note that:\\
The $w_1(E)^3$ is  $(1,0,0,0)$ in the basis of \eq{eq:bordism3OBO-basis}.\\
The $w_1(E)w_2(E)$ is  $(0,1,0,0)$ in the basis of \eq{eq:bordism3OBO-basis}.\\
The $w_1(E)w_2(V_{\SO(3)})$ is  $(1,1,0,0)$ in the basis of \eq{eq:bordism3OBO-basis}.\\
The $w_3(V_{\SO(3)})$ $=$ $w_3(E)+w_1(E)w_2(E)$ $=$ $w_1(TM)w_2(E)$ is  $(0,1,1,0)$ in the  basis of \eq{eq:bordism3OBO-basis}.\\
The $w_3(E)=w_3(V_{\tO(3)})$ is  $(0,0,1,0)$ in the basis of our \eq{eq:bordism3OBO-basis}.\\

Therefore, Ref.~\cite{Komargodski2017dmc1705.04786}'s anomaly \eq{eq:w3E} given by
$w_3(E)$ $=w_1(E)^3$ $+w_1(E)w_2(V_{\SO(3)})$ $+w_1(TM)w_2(E)$
coincides with one of the cobordism invariant as $(0,0,1,0)$ in the basis of our \eq{eq:bordism3OBO-basis}.
We had explained the physical meaning of $w_1(E)w_2(V_{\SO(3)})$ term in \eq{eq:anom-stacking-Haldane}.
We will explain the meaning of $w_1(E)^3$ in \Sec{sec:cubic}
and the meaning of $w_1(TM)w_2(E)$ in \Sec{sec:w1(TM)w2(E)}

\subsection{Cubic anomaly of $\Z_2^C$ of $\mathbb{CP}^{1}$-model} 
\label{sec:cubic}

Now we like to capture the physical meaning of a cubic anomaly of $\Z_2^C=\Z_2^x$-symmetry in \eq{eq:w3E}: 
\bea \label{eq:CP1-anom-A3}
\boxed{w_1(E)^3 \equiv w_1(\Z_2^x)^3  \equiv (A_x)^3},
\eea
which is a sensible cobordism invariant as the $(1,0,0,0)$ in the basis of \eq{eq:bordism3OBO-basis}.
Ref.~[\onlinecite{Metlitski2017fmd1707.07686}] 
also points out this $w_1(E)^3$ or the $A_x^3$-anomaly, where $A_x$ is regarded as the $\Z_2^x$-translational background gauge field.
We know that the 2d boundary physics we look at in \Fig{fig:stack} (encircled by the dashed-line rectangle) describes 
the gapless CFT theory of SU(2)$_1$ WZW model at the level $k=1$.
The SU(2)$_1$ WZW model at $k=1$ is equivalent to a $c=1$ compact non-chiral boson theory 
(the left and right chiral central charge $c_L=c_R=1$, but the chiral central charge $c_-=c_L- c_R=0$) 
at the self-dual radius \cite{DiFrancesco1997nkBook}.
Although properly we could use non-{Abelian bosonization} method \cite{Witten1983arNab}, here focusing on the
abelian $\Z_2^x$-symmetry and its anomaly, we can simply use the {Abelian bosonization}.

Since the SU(2)$_1$ WZW model at $k=1$ is equivalent to a $c=1$ compact non-chiral boson theory at the self-dual radius,
 we consider an action
\bea
&& {{S}}_{2\text{d}}= 
\frac{1}{2\pi \alpha'} \int dz  d\bar z \; (\partial_z \Phi) (\partial_{\bar z}  \Phi) + \dots, \\
&&{{S}}_{2\text{d}}= 
\frac{1}{4\pi} \int dt\; dx \; \big( K_{IJ} \partial_t \phi_{I} \partial_x \phi_{J} -V_{IJ}\partial_x \phi_{I}   \partial_x \phi_{J} \big) + \dots.\nn
\eea
requiring a rank-2 symmetric bilinear form $K$-matrix, 
\be \label{eq:K-mat}
K_{IJ} =\bigl( {\begin{smallmatrix} 
0 &1 \\
1 & 0
\end{smallmatrix}}  \bigl) \oplus 
\dots; \quad
V_{IJ} =\bigl( {\begin{smallmatrix} 
v &0 \\
0 & v
\end{smallmatrix}}  \bigl) \oplus 
\dots.
\ee

The first form of the action is familiar in string theory and a $c=1$ compact non-chiral boson theory
at the self-dual radius.
(In string theory, we are looking at
$R = \sqrt \alpha' = \sqrt 2$.)

The second form of the action is the familiar 2d boundary of 3d bosonic SPTs. 
This second description is also known as Tomonaga-Luttinger liquid theory \cite{Tomonaga1950zz, Luttinger1963zz, Haldane1981zzaLuttingerJPhC} in condensed matter physics. 
It is a $K$-matrix multiplet generalization of the usual chiral boson theory of Floreanini and Jackiw \cite{Floreanini1987asJackiw}.
The reason we write $\dots$ in \eq{eq:K-mat} is that there could be additional 3d SPTs sectors for 2d $\CP^1$-model (e.g. \eq{eq:2d-anomaly-CP1}), 
more than what we focus on in this subsection.
Here we trade the boson scalar $\Phi$ to $\phi_{1}$, while  $\phi_{2}$ is the dual boson field.
We can determine the bosonic anomalies \cite{WangSantosWen1403.5256} by looking at the anomalous symmetry transformation on the 2d theory, living on the boundary of which 3d SPTs.
We use the mode expansion for a multiplet scalar boson field theory \cite{WangSantosWen1403.5256},
with zero modes ${\phi_{0}}_{I}$ and winding modes $P_{\phi_J}$:
$$
\phi_I(x) ={\phi_{0}}_{I}+K^{-1}_{IJ} P_{\phi_J} \frac{2\pi}{L}x+\ii \sum_{{n\neq 0}}^{n \in \mathbb{Z}} \frac{1}{n} \alpha_{I,n} e^{- \ii n x \frac{2\pi}{L}},
$$
which satisfy the commutator 
$[{\phi_{0}}_{I},  P_{\phi_J}]=\ii\delta_{IJ}$. The Fourier modes satisfy a generalized Kac-Moody algebra: 
$[\alpha_{I,n} , \alpha_{J,m} ]= n K^{-1}_{IJ}\delta_{n,-m}$.
{For a modern but self-contained pedagogical treatment on a canonical quantization of $K$-matrix multiplet (non-)chiral boson theory,  
the readers can consult Appendix B of \cite{LianWang2018xep1801.10149}.}

Follow \cite{Metlitski2017fmd1707.07686}, based on the identification of spin observables of Hamiltonian model \eq{eq:H}
and the abelian bosonized theory,
 we can map the symmetry transformation to the continuum description
on the boson multiplet $\phi_I(x)=(\phi_1 (x),$ $\phi_2  (x)$).
The commutation relation is $ [\phi_I(x_1), K_{I'J} \partial_x \phi_{J}(x_2)]= {2\pi} \ii  \delta_{I I'} \delta(x_1-x_2)$.
The continuum limit of 2d anomalous symmetry transformation is  \cite{Lu:2012dt1205.3156} \cite{WangSantosWen1403.5256}: 
\bea \label{eq:globalS2d}
&&\text{S}^{(p)}_{\rN}
=
\exp[
\frac{\ii}{\rN}\,
(
\int^{L}_{0}\,dx\,\partial_{x}\phi_{2}
+
p\,\int^{L}_{0}\,dx\,\partial_{x}\phi_{1}  
)
], \;\;\;\;\;\\
&&\text{S}^{(p)}_{\rN}
 {\begin{pmatrix} 
\phi_1 (x)   \\
\phi_2(x) 
\end{pmatrix}} (\text{S}^{(p)}_{\rN})^{-1}
=
 {\begin{pmatrix} 
\phi_1 (x)   \\
\phi_2(x)  
\end{pmatrix}} 
+\frac{2\pi}{\rN}{\begin{pmatrix} 
1   \\
p  \end{pmatrix}}.  
 \nn
\eea
Here $L$ is the compact spatial $S^1$ circle size of the 2d theory.
For 2d $\mathbb{CP}^{1}$-model, we have $\rN=2$ and $p=1$,
this is indeed known as the Type I \emph{bosonic anomaly} in \cite{WangSantosWen1403.5256},
which also recovers one anomaly found in \cite{Metlitski2017fmd1707.07686} and in \cite{Komargodski2017dmc1705.04786}'s  \eq{eq:w3E}.

\subsection{Mixed anomaly of time-reversal $\Z_2^T$ and 0-form flavor $\Z_\rN$-center symmetry of $\mathbb{CP}^{1}$-model}

\label{sec:w1(TM)w2(E)}

Ref.~\cite{Yamazaki:2017ulc,Tanizaki:2017qhf1710.08923}  point out another anomaly of $\mathbb{CP}^{1}$-model, which mixes between
time-reversal (which we have chosen to be $\cC\cT$) and the PSU(2) symmetry (which is viewed as the twisted flavor symmetry in \cite{Yamazaki:2017ulc,Tanizaki:2017qhf1710.08923}).
Compare with \eq{eq:bordism3OBO}, we can interpret the above 2d anomalies are captured by
a 3d cobordism invariant for N = 2 case:
\be \label{eq:w1TMw2SO3}
\boxed{w_1(TM) w_2(V_{\SO(3)})=w_1(TM)w_2(E) = { w_3(V_{\SO(3)})}}.
\ee  
This also coincides with the last anomaly term in \eq{eq:w3E}'s $w_3(E)$.
We derive the above first equality in \eq{eq:w1TMw2SO3} based on 
$\Sq^1(w_1(E)^2)$ $=2 w_1(E)  \Sq^1 w_1(E)$ $= 0$ and
 combine Wu formula,
$ \Sq^1(w_1(E)^2)$ $=w_1(TM)(w_1(E)^2)$ $= 0$.
Thus,
\bea
&&w_1(TM) w_2(V_{\SO(3)}) =w_1(TM)(w_2(E)+w_1(E)^2)\nn \\
&& =w_1(TM)w_2(E) =w_1(TM) w_2(V_{\tO(3)}).
\eea
The last equality in \eq{eq:w1TMw2SO3} is due to
$w_1(TM) w_2(V_{\SO(3)})$ $= \Sq^1 w_2(V_{\SO(3)})$ $={ w_3(V_{\SO(3)})}.$

See \Ref{W2}, we can combine the Steenrod-Wu formula for $i<j$:
\bea
 \Sq^i(w_j) = w_iw_j+\sum_{k=1}^i \binom{j-i-1+k}{k}
w_{i-k} w_{j+k},\quad \;
\eea
and Wu formula 
\bea
\Sq^{d-j}(x_j)=u_{d-j} x_j,  \text{ for any } x_j \in \H^j(M^d;\Z_2)\quad
\eea
to obtain:
\bea
&&w_1(E)w_2(E)+w_3(E)=\Sq^1(w_2(E))=w_1(TM)w_2(E),\nn\\
&&\Rightarrow w_3(E) =(w_1(E)+w_1(TM))w_2(E)\nn\\
&&\Rightarrow w_1(TM)w_2(E) = w_3(E) + w_1(E) w_2(E),
\eea
so we derive $w_1(TM)w_2(E)$ is  $(0,1,1,0)$ in the basis of \eq{eq:bordism3OBO-basis}.
The physical meaning of the 2d anomaly \eq{eq:w1TMw2SO3} will be explored later in \Sec{sec:rule}, \Sec{sec:5d3d} and in Fig.~\ref{fig:YM-reduce},
which can be understood as the dimensional reduction of 4d anomaly of YM theory compactified on a 2-torus with twisted boundary conditions 
\cite{Yamazaki:2017ulc} \cite{Yamazaki:2017dra}.

In \Sec{sec:cubic}, We had checked some of the 2d bosonic anomaly by dimensional reducing from 4d to 2d, can be captured by
abelian bosonization method as Type I bosonic anomaly in \cite{WangSantosWen1403.5256}.
Some of the anomalies in the above may be also related to other (Type II or Type III) bosonic discrete anomalies, when
we break down the global symmetry to certain subgroups.

%
%
%
%
 %
%
%

\section{Rules of The Game  for Anomaly Matching Constraints} 
\label{sec:rule}

With all the QFT and global symmetries information given in \Sec{sec:remark},
and all the possible anomalies enumerated by the cobordism theory computed in \Sec{sec:cobo-top},
and all the known anomalies in the literature derived and re-written in terms of cobordism invariants organized in \Sec{sec:rev},
now we are ready to set up the rules of the game to determine the full anomaly constraints for these QFTs (4d SU(N) YM theory and 2d $\CP^{\rN-1}$ model at $\theta=\pi$).\\

{Below we simply abbreviate the 
``5d invariant'' as the 5d cobordism/(higher) SPTs invariants which captures the anomaly of 4d SU(N) YM at $\theta=\pi$ at even N,
and
``3d invariant'' as the 3d cobordism/SPTs invariants which captures the anomaly of 2d $\CP^{\rN-1}$ at $\theta=\pi$ at even N.
Our convention chooses the natural time-reversal symmetry transformation as $\cC\cT$.
}\\

\noindent
Rules:
\begin{enumerate}[label=\textcolor{blue}{Rule \arabic*.}, ref={Rule \arabic*}]

\item \label{Rule 1}
For 5d invariant,
for 4d SU(N) YM  at $\theta=\pi$ of an even integer N must have analogous anomaly captured by 5d cobordism term of 
$\sim w_1(TM)(B_2)^2$ (up to some properly defined normalization and quantization). 
(It will become transparent later in \eq{eq:w1PB} that the precise term needs to be $\frac{1}{2} \tilde w_1(TM) \cP_2(B_2))= B_2\Sq^1B_2+\Sq^2\Sq^1B_2$.)

\item \label{Rule 2}
The chosen 5d invariants may be non-vanished in O-bordism group, but they are vanished in SO-bordism group.

\item \label{Rule 3}
The 3d invariant for 2d $\CP^1$ model must include the 3d cobordism invariants discussed in \Sec{sec:rev}, in particular, \eq{eq:2d-anomaly-CP1}.

\item \label{Rule 4}
The 3d invariant for other 2d $\CP^{\rN-1}$ for even N (e.g. 2d $\CP^3$) model must include some of familiar terms generalizing that of 2d $\CP^1$ model.

\item \label{Rule 5} 
Due to the physical meanings of $\cC \cT$ and $\cT$ (and other orientation-reversal symmetries), 
we must impose a swapping symmetry for 5d invariants.

\item \label{Rule 6} 
Relating the 5d and 3d invariants: There is a dimensional reductional constraint and physical meanings 
between the 5d and 3d invariants, for example, by the twist-compactification on 2-torus $T^2$.

\item \label{Rule 7} 
The 5d invariants for a 4d pure YM theory must involve the nontrivial 2-form $B_2$ field. 
The 5d terms that involve no $B_2$ dependence should be discarded. 

\end{enumerate}

%
\onecolumngrid
\begin{widetext}
\begin{figure}[!h]
\centering
\includegraphics[scale=2.]{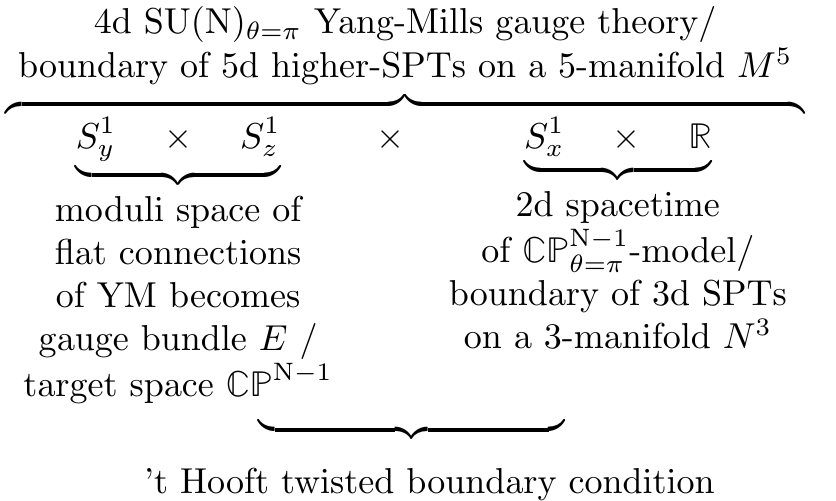}
  \caption{Follow the setup of the twisted boundary condition induced 't Hooft {boundary condition} 
  \cite{tHooft1979rtgEM} along the 2-torus $T^2_{zx} \equiv S^1_z \times S^1_x$,
and the twisted compactification \cite{Yamazaki:2017ulc} \cite{Yamazaki:2017dra}, we examine that the higher anomaly of 4d SU(N) YM theory at $\theta=\pi$ induces
the anomaly of 2d $\mathbb{CP}^{\mathrm{N}-1}$ model at $\theta=\pi$. The 4d YM on ${S^1_x \times S^1_y \times S^1_z \times \mathbb{R}}$
is compactified along the small size of $T^2_{yz} \equiv S^1_y \times S^1_z$, whose moduli space of flat connections becomes the target space $\mathbb{CP}^{\mathrm{N}-1}$ \cite{Looijenga1976, FriedmanWitten1997yq9701162},
while the remained $ S^1_x \times \mathbb{R}$ becomes the 2d spacetime of 2d $\mathbb{CP}^{\mathrm{N}-1}$ model.
Our goal, in Sec.~\ref{sec:rule}, \ref{sec:newSUn} and \ref{sec:newCPn} is to identify the underlying 't Hooft anomalies 
of 4d SU(N) YM and 2d $\mathbb{CP}^{\mathrm{N}-1}$ model, namely identifying the theories living on the boundary ($\equiv$ bdry) of 5d and 3d (higher) SPTs
when all the (higher) global symmetries needed to be regularized strictly \emph{onsite} and \emph{local} 
(e.g. [\onlinecite{Wen2013oza1303.1803, 1405.7689, Wang2017locWWW1705.06728}]). 
(Higher global symmetries can be regularized, instead on the 0-simplex, on the higher-simplices: 1-simplex, 2-simplex, etc.)
The twisted boundary condition of 4d YM for 1-form $\mathbb{Z}_\mathrm{N}$-center symmetry (as a higher symmetry twist of \cite{1405.7689})
can be dimensionally reduced to the 0-form $\mathbb{Z}_\mathrm{N}$-flavor symmetry twisted \cite{Dunne2012aeUnsal1210.2423} in the 2d $\mathbb{CP}^{\mathrm{N}-1}$ model. 
In fact, the twisted boundary conditions can be designed numerically, and the twisted boundary conditions can result in a fractional instanton number in a gauge theory, 
see, for example, recent numerical attempts and Reference therein \cite{Itou2018wkm1811.05708}.
}
\label{fig:YM-reduce}
\end{figure}

\pagebreak
\end{widetext}
\twocolumngrid

Here are the explanations for our rules.\\

\ref{Rule 1} is  based on \Sec{sec:SUN-YM-mix-higher},  
for 4d SU(N) YM  at $\theta=\pi$ of an even integer N must have analogous anomaly captured by 5d cobordism term of 
$\sim w_1(TM)B_2^2$ (up to some properly defined normalization and quantization), where we choose the linear time-reversal symmetry transformation from $\cC\cT$
and a quadratic term of 2-form fields $B_2$ coupling to 1-form center symmetry.\\

\ref{Rule 2}'s physical reasoning is that the time-reversal symmetry transformation from $\cC\cT$ plays an important role for the anomaly.
We can see from \Sec{sec:C-P-R} that only when 
time-reversal or orientation reversal is involved ({$\cT$, $\cP$, $\cC \cT$ and $\cC \cP$}), we have the mixed higher anomalies for YM theory;
while for the others ($\cC$,  $\cP \cT$ and  $\cC \cP \cT$), we do not gain mixed anomalies (e.g. with the 1-form center symmetry). \\

\ref{Rule 3} is dictated by the known physics derivations in \Sec{sec:rev} and in the literature.\\

\ref{Rule 4} will become clear in \Sec{sec:newCPn}.\\

\ref{Rule 5}, the swapping symmetry for 5d invariants between $\cC \cT$ and $\cT$ (and other orientation-reversal symmetries), 
we will interpret the unoriented O($D$) spacetime symmetry with the time reversal from $\cC\cT$ or from $\cT$
can be swapped. 
This means that we can choose 
the 5d topological invariant from the former $\Omega_5^{\tO}(\B\Z_2\ltimes\B^2\Z_4) \equiv \Omega_5^{(\tO \times (\Z_2 \ltimes \B\Z_4)) }$ for $\cC\cT$, 
rather than the more complicated latter $\Omega_5^{\tO_{\ltimes}}(\B \Z_2 \ltimes \B^2\Z_4) \equiv \Omega_5^{(\tO \times \Z_2) \ltimes \B\Z_4 }$ for  $\cT$. 
We focus on the 5d terms involving $\cC \cT$-symmetry.
\\

\ref{Rule 6} about the dimensional reduction from 5d to 3d (or 4d to 2d) is explained in \Fig{fig:YM-reduce} and the main text, such as in \Sec{sec:5d3d}.
We should also find the mathematical meanings behind this constraint in \Sec{sec:5d3d}.\\

\ref{Rule 7} is based on the physical input that there should be \emph{no} obstruction to regularize a pure YM theory by imposing only ordinary 0-form symmetry alone \emph{onsite}.
The obstruction only comes from 
regularizing a pure YM theory with the involvement of restricting
both the higher 1-form center symmetry to be \emph{on-link and local}, and the ordinary 0-form symmetry to be \emph{on-site and local}.
Here \emph{on-link} means that the symmetry acts locally on the 1-simplex,
and
\emph{on-site} means that the symmetry acts locally on the 0-simplex or a point.

Thus, it is necessary to turn on the 2-form background field $B_2$ in order to detect the 't Hooft anomaly of YM theory.
Namely, the 5d cobordism invariants of the form $w_1(TM)^{\rm t} \cup A^{\rm a}$ with $\rm t + \rm a =5$ should be discarded
out of the candidate list of 5d term for 4d YM anomalies.\\

Therefore, we can refine the set \eqn{eq:select-N=4YM} to a smaller subset
satisfying  \ref{Rule 7}:
\bea  \label{eq:select-N=4YM-2}
\boxed{\left\{\begin{array}{l} 
B_2\cup\beta_{(2,4)}B_2,\\ 
w_1(TM^5)^2\beta_{(2,4)} B_2,\\ 
A^2\beta_{(2,4)}B_2,\;AB_2w_1(TM^5)^2.
\end{array}\right.}
\eea


\section{Rules of Dimensional Reduction: 5d to 3d}

\label{sec:5d3d}

Now we aim to utilize the Rule 6 in \Sec{sec:rule} and the new anomaly of 2d $\mathbb{CP}^{\rN-1}$-model  found in \Sec{sec:newCPn},
to deduce the new higher anomaly of 4d YM theory --- which later will be organized in \Sec{sec:newSUn}.

From the physics side, follow \cite{Yamazaki:2017ulc}, see Fig.\ref{fig:YM-reduce},
 we choose the 4d YM living on ${S^1_x \times S^1_y \times S^1_z \times \R}$,
such that the size $L_y$, $L_z$ of $S^1_y \times S^1_z$
is taken to be much smaller than the size $L_x$ of $S^1_x$,
namely  $L_y, L_z \ll L_x$. Then, below the energy gap scale
$$
\Delta_E \ll {L_y}^{-1} \text{ and } {L_z}^{-1}, 
$$
the resulting 2d theory on ${S^1_x \times \R}$
is given by a sigma model with a target space of $\mathbb{CP}^{\rN-1}$.
There are several indications that the low energy theory is a 2d
$\mathbb{CP}^{\rN-1}$-model:
\begin{itemize}
\item The 4d and 2d instanton matchings in \cite{atiyah1984, Donaldson1984tm} and other mathematical works.
The $\theta=\pi$-topological term of SU(N) YM 
is mapped to the $\theta=\pi$-topological term of 2d  $\mathbb{CP}^{\rN-1}$-model.

\item The moduli space of flat connections on the 2-torus $T^2=S^1_y \times S^1_z$ of 4d YM theory 
is the projective space $\mathbb{CP}^{\rN-1}$ \cite{Looijenga1976, FriedmanWitten1997yq9701162} (up to the geometry details of
no canonical Fubini-Study metric and singularities mentioned in \cite{Yamazaki:2017ulc} and footnote \ref{footnote:FS}). See Fig.\ref{fig:YM-reduce}.

\item The 1-form $\Z_\rN$-center symmetry of 4d YM is dimensionally reduced, in addition to 1-form symmetry itself, also to a 0-form $\Z_\rN$-flavor of 2d $\CP^{\rN-1}$ model.
The twisted boundary condition of 4d YM for 1-form $\Z_\rN$-center symmetry (e.g., \cite{Tanizaki:2017qhf1710.08923} as a higher symmetry twist of \cite{1405.7689})
can be dimensionally reduced to the 0-form $\Z_\rN$-flavor symmetry twisted \cite{Dunne2012aeUnsal1210.2423} in the 2d $\CP^{\rN-1}$ model. 

\item Ref.~\cite{Yamazaki:2017dra} derives that the physical meaning of the 2d anomaly \eq{eq:w1TMw2SO3}
is directly descended from the 4d anomaly \eq{eq:BSq1B-1} of YM theory by twisted $T^2$ compactification.
 
\end{itemize}

Encouraged by the above physical and mathematical evidences, 
in this section, we formalize the 4d and 2d anomaly matching under the twisted $T^2$ compactification,
into a mathematical precise problem of the
5d and 3d cobordism invariants (SPTs/topological terms) matching, under a
2-torus $T^2$ dimensional reduction.

Below we follow our notations of the bordism groups in \Sec{sec:cobo-top},
and their $D$d = $(d+1)$d cobordism invariants to the $d$d  anomalies of QFTs.
We may simply abbreviate, 
$$\text{``5d cobordism invariants for 4d YM theory's anomaly''}$$
$$\equiv  \text{``5d (Yang-Mills) terms.''}$$
We may simply abbreviate, 
$$\text{``3d cobordism invariants for 2d $\CP^{\rN-1}$ model's anomaly''}$$
$$\equiv  \text{``3d ($\CP^{\rN-1}$) terms.''}$$

\subsection{Mathematical Set-Up}\label{sec:set-up}

We aim to find the corresponding 3d $\CP^{\rN-1}$ terms which are obtained from certain 5d Yang-Mills terms under a ``$T^2$ reduction''.

Recall that we have the following group isomorphisms in \Sec{sec:B}, \Sec{sec:D}, \Sec{sec:F}, and \Sec{sec:G}:
\bea \label{eq:Phi1}
&\Phi_1:&\Omega_5^{\tO}(\B^2\Z_2)\to\Z_2^4,\notag\\
&&(M^5,B_2)\mapsto(B_2\cup \Sq^1B_2,\Sq^2 \Sq^1B_2,\notag\\
&&w_1(TM^5)^2\cup \Sq^1B_2,w_2(TM^5)w_3(TM^5)).
\eea 
\bea \label{eq:Phi2}
&\Phi_2:&\Omega_3^{\tO}(\B\tO(3))\to\Z_2^4,\notag\\
&&(M^3,E)\mapsto(w_1(E)^3,w_1(E)w_2(E),\notag\\
&&w_3(E), w_1(E)w_1(TM^3)^2).
\eea
\bea  \label{eq:Phi3'}
&\Phi_3':&\Omega_5^{\tO}(\B\Z_2\ltimes\B^2\Z_4)\to\Z_2^{11}\notag\\
&&(M^5,A,B_2)\mapsto(B_2\cup\beta_{(2,4)}B_2, AB_2w_1(TM^5)^2,\nn\\
&&w_1(TM^5)^2\beta_{(2,4)} B_2, A^2\beta_{(2,4)}B_2, \nn\\
&&A^3w_1(TM^5)^2, Aw_1(TM^5)^4,\nn\\
&&
\Sq^2\beta_{(2,4)} B_2, w_2(TM^5)w_3(TM^5),\nn\\
&&
A^5,
AB_2^2,Aw_2(TM^5)^2).
\eea
and
\bea
&\Phi_4:&\Omega_3^{\tO}(\B(\Z_2\ltimes\PSU(4)))\to\Z_2^4,\notag\\
&&(M^3,w_1(E),w_2(E))\mapsto(w_1(E)^3,w_1(E)(w_2(E)\mod2)   ,\notag\\
&&\beta_{(2,4)}w_2(E),w_1(E)w_1(TM^3)^2). 
\eea

We have the group isomorphisms
\bea\label{eq:Phi1'}
&\Phi_1':&\Omega_5^{\tO}(\B^2\Z_2)/\Omega_5^{\SO}(\B^2\Z_2)\to\Z_2^2,\notag\\
&&[(M^5,B_2)]\mapsto(B_2\cup \Sq^1B_2+\Sq^2 \Sq^1B_2,\notag\\
&&w_1(TM^5)^2\cup \Sq^1B_2).
\eea

and

\bea  \label{eq:Phi3''}
&\Phi_3'':&\Omega_5^{\tO}(\B\Z_2\ltimes\B^2\Z_4)/\Omega_5^{\SO}(\B\Z_2\ltimes\B^2\Z_4)\to\Z_2^{6}\notag\\
&&[(M^5,A,B_2)]\mapsto(B_2\cup\beta_{(2,4)}B_2, AB_2w_1(TM^5)^2,\nn\\
&&w_1(TM^5)^2\beta_{(2,4)} B_2, A^2\beta_{(2,4)}B_2, \nn\\
&&A^3w_1(TM^5)^2, Aw_1(TM^5)^4.
\eea

\subsubsection{General Set-Up}\label{sec:set-up-general}

In general, if $V$ and $W$ are two vector spaces, $\rho:V\to W$ is a linear map, and we choose the  bases of $V$ and $W$ to be $(e_1,\dots,e_n)$ and $(f_1,\dots, f_m)$ respectively. Suppose the matrix of $\rho$ with respect to the bases chosen above is an $m \times n$ matrix $\tA$, namely, we have
$$\rho(e_1,\dots,e_n)=(f_1,\dots, f_m)\tA.$$

Let $\rho^*:W^*\to V^*$ be the dual linear map of $\rho$, and $(e_1^*,\dots,e_n^*)$ and $(f_1^*,\dots, f_m^*)$ be the dual bases of $V^*$ and $W^*$ respectively. Suppose $\tA=(a_{ij})$, then we have
$$\rho(e_i)=\sum_{j=1}^m  a_{ji} f_j,$$
and 
$$\rho^*(f_j^*)=\sum_{i=1}^na_{ji}e_i^*.$$
Namely, the matrix of $\rho^*$ with respect to the dual bases chosen above is $\tA^T$, the transpose of $\tA$.

\subsubsection{Mathematical Set-Up for $\rN=2$}\label{sec:set-up-N=2}

Below we elaborate on the $\rN=2$ case first.

Let $W$ in this subsection be the quotient group $\Omega_5^{\tO}(\B^2\Z_2)/\Omega_5^{\SO}(\B^2\Z_2)$. 

We choose a basis of $W$ to be $(f_1,f_2)$ where 
\bea
\left\{
\begin{array}{l}
f_1=\Phi_1'^{-1}(1,0),\\
f_2=\Phi_1'^{-1}(0,1).
\end{array}
\right.
\eea
The dual basis of $W^*$ is $(f_1^*,f_2^*)$
where 
\bea
\left\{
\begin{array}{l}
f_1^*=B\Sq^1B+\Sq^2\Sq^1B,\\
f_2^*=w_1(TM^5)^2\Sq^1B.
\end{array}
\right.
\eea

Let $V=\Omega_3^{\tO}(\B\tO(3))$.

We choose a basis of $V$ to be $(e_1,e_2,e_3,e_4)$ where
\bea
\left\{
\begin{array}{l}
e_1=\Phi_2^{-1}(1,0,0,0),\\
e_2=\Phi_2^{-1}(0,1,0,0),\\
e_3=\Phi_2^{-1}(0,0,1,0),\\
e_4=\Phi_2^{-1}(0,0,0,1).
\end{array}
\right.
\eea

The dual basis of $V^*$ is $(e_1^*,e_2^*,e_3^*,e_4^*)$
where 
\bea
\left\{
\begin{array}{l}
e_1^*=w_1(E)^3,\\
e_2^*=w_1(E)w_2(E),\\
e_3^*=w_3(E),\\
e_4^*=w_1(E)w_1(TN^3)^2.
\end{array}
\right.
\eea

We aim to construct a linear map $F:V\to W$, then we can find the image of $f_i^*$ under $F^*$, which is the desired 3d term reduced from the 5d term $f_i^*$.

\subsubsection{Mathematical Set-Up for $\rN=4$}\label{sec:set-up-N=4}
\label{sec:Set-UpforN=4}

Below we elaborate on the $\rN=4$ case.

Let $W$ be the subgroup of the quotient group $\Omega_5^{\tO}(\B\Z_2\ltimes\B^2\Z_4)/\Omega_5^{\SO}(\B\Z_2\ltimes\B^2\Z_4)$ defined by
\bea
W&=&\{[(M^5,A,B)]\in\Omega_5^{\tO}(\B\Z_2\ltimes\B^2\Z_4)/\Omega_5^{\SO}(\B\Z_2\ltimes\B^2\Z_4)\mid\nn\\
&&\text{the 3-rd to 6-th } \text{components of }\Phi_3''[(M^5,A,B)]\text{ are 0}\}.\nn\\
\eea

{First, {the 5-th to 6-th components of} $\Phi_3''[(M^5,A,B)]$ {are 0} because
we derive that from \eqn{eq:select-N=4YM-2}, 
for \eqn{eq:Phi3''}, only $B_2 \beta_{(2,4)}B_2$, 
$w_1(TM^5)^2\beta_{(2,4)} B_2$, 
$A^2\beta_{(2,4)}B_2$, 
and
 $AB_2w_1(TM^5)^2$ are possible 5d terms in \eqn{eq:Phi3''} for YM theory, while
other terms  
$A^3w_1(TM^5)^2$ and $Aw_1(TM^5)^4$ must \emph{not} appear in the 4d anomaly of SU(N) YM at $\rN=4$.
Moreover, we will discuss the possibility of $w_1(TM^5)^2\beta_{(2,4)}B_2$ and $A^2\beta_{(2,4)}B_2$ separately on an upcoming work.
Here we assume that $w_1(TM^5)^2\beta_{(2,4)}B_2$ and $A^2\beta_{(2,4)}B_2$ are unlikely for a 4d SU(4) YM theory, thus 
we will set {the 3-rd to 6-th components of} $\Phi_3''[(M^5,A,B)]$ {are 0}.
}

We choose a basis of $W$ to be $(f_1,f_2)$ where 
\bea
\left\{
\begin{array}{l}
f_1=\Phi_3''^{-1}(1,0,0,\dots,0),\\
f_2=\Phi_3''^{-1}(0,1,0,\dots,0).
\end{array}
\right.
\eea
The dual basis of $W^*$ is $(f_1^*,f_2^*)$
where 
\bea
\left\{
\begin{array}{l}
f_1^*=B\beta_{(2,4)}B,\\
f_2^*=Aw_1(TM^5)^2B.
\end{array}
\right.
\eea

Let $V=\Omega_3^{\tO}(\B(\Z_2\ltimes\PSU(4)))$.

We choose a basis of $V$ to be $(e_1,e_2,e_3,e_4)$ where
\bea
\left\{
\begin{array}{l}
e_1=\Phi_4^{-1}(1,0,0,0),\\
e_2=\Phi_4^{-1}(0,1,0,0),\\
e_3=\Phi_4^{-1}(0,0,1,0),\\
e_4=\Phi_4^{-1}(0,0,0,1).
\end{array}
\right.
\eea

The dual basis of $V^*$ is $(e_1^*,e_2^*,e_3^*,e_4^*)$
where 
\bea
\left\{
\begin{array}{l}
e_1^*=w_1(E)^3,\\
e_2^*=w_1(E)(w_2(E)\mod2),\\
e_3^*=\beta_{(2,4)}w_2(E),\\
e_4^*=w_1(E)w_1(TN^3)^2.
\end{array}
\right.
\eea

We aim to construct a linear map $G:V\to W$, then we can find the image of $f_i^*$ under $G^*$, which is the desired 3d term reduced from the 5d term $f_i^*$.

\subsubsection{Construction of maps:\\ 
$F$ for $\rN=2$ and $G$ for $\rN=4$}
\label{sec:map-F-G}

{To construct the map $F$ and $G$, we first discuss the background fields for the global symmetries in both 4d/5d and 2d/3d. }
\begin{enumerate}
	\item {For $\rN=2$,  because our bordism calculations are only for the cases where time reversal commutes with other symmetries, the time reversal symmetry in 5d is identified as $\cC\cT$, and the background field is $w_1(TM^5)$. However, for $\rN=2$, the charge conjugation symmetry is trivial. In 3d, because both $\cT, \cC$ and $\cC\cT$  commute with all other symmetries, so we identify time reversal symmetry as $\cT$ whose background is $w_1(TN^3)$, and the charge conjugation symmetry in 3d is identified as $\cC$ whose background is $w_1(E)$. }
	\item {For $\rN=4$, because our bordism calculations are only for the cases where time reversal commutes with other symmetries, in 5d, we identify $\cC\cT$ as time reversal symmetry whose background field is $w_1(TM^5)$, and the charge conjugation is denoted as $\cC$ whose background field is $A$. In 3d, for the same reason as above, we identify $\cC\cT$ as time reversal symmetry whose background field is $w_1(TN^3)$, and  the charge conjugation is denoted as $\cC$ whose background field is $w_1(E)$. }

\end{enumerate}

We proceed to discuss the reduction rules for the symmetry background fields. We first focus on $\rN=2$. 
\begin{enumerate}
	\item $\cC\cT$ in 5d reduces to $\cC\cT$ in 3d. Correspondingly the background field for time reversal in 5d $w_1(TM^5)$ reduces to $w_1(TN^3)+ w_1(E)$. 
	This is because restricting $w_1(TM^5)$ to a submanifold should be of the form  $w_1(TN^3)+(...)$. To determine (...), we notice that 
	the 2d $\CP^1$ model with a theta term at $\theta =\pi \mod 2 \pi$ (but not other $\theta$ except $\theta =0 \mod 2 \pi$), respects the charge conjugation symmetry. On the other hand, the 4d YM with a theta term at $\theta =\pi \mod 2 \pi$ (but not other $\theta$ except $\theta =0 \mod 2 \pi$), respects the $\cC\cT$ symmetry. 
	We thus demand $w_1(TM^5)$ reduces to $w_1(TN^3)+w_1(E)$. 
	Formally, we need to find an embedding $\iota:N^3\hookrightarrow M^5$ such that 
	\bea\label{embeddingN2}
	w_1(TM^5)\mid_{N^3}=w_1(TN^3)+w_1(E). 
	\eea
	In summary, we find the symmetry reduction
	\begin{eqnarray}
	\text{5d}: \cC\cT \to \text{3d}: \cC\cT.
	\end{eqnarray}

	\item Following \Ref{Yamazaki:2017ulc}, the twisted boundary condition by the center symmetry (which is $B$) for the Yang-Mills is reduced to the twisted boundary condition by the $\Z_{\rN}$ global symmetry (which is $w_2(E)+\bar{K}_1w_1(E)^2$ for $\rN=2$) of $\CP^{\rN-1}$. Here we find two possibilities of the reduction, labeled by $\bar{K}_1\in \Z_2$. Formally, we need to find $B\in \H^2(M^5,\Z_2)$ such that  
	\be
	\begin{split}
		B\mid_{N^3}&=w_2(E)+\bar{K}_1 w_1(E)^2\\&=w_2(V_{\SO(3)}) +(\bar{K}_1+1) w_1(E)^2.
	\end{split}
	\ee
	We will determine  $\bar{K}_1$ at the end of \Sec{N=2reduction}. 
\end{enumerate}

After choosing an embedding $\iota:N^3\hookrightarrow M^5$, we have the Poincar\'e dual $\text{PD}:\H^2(M^5,\Z_2)\xrightarrow{\sim}\H_3(M^5,\Z_2)$, we denote $B'=\text{PD}^{-1}(\iota_*([N^3]))$. We require that $B'$ is the cup product of two different degree-1 cohomology classes. 
{We impose this condition because, as discussed at the beginning of \Sec{sec:5d3d},  the 3d submanifold $N^3$ is reduced from the 5d manifold $M^5$ by a 2-torus.}
{Suppose the normal bundle of the embedding $N^3\hookrightarrow M^5$ is $\nu$, then $\nu$ is the direct sum of two different line bundles, and the condition $B'\mid_{N^3}=w_2(\nu)$ is satisfied.}
{We also require that $\Sq^1B=\Sq^1B'$ to ensure that the map $F$ in Lemma \ref{lemmaF} is well-defined.}

Now we construct the map $F$, which we organize as the following lemma. 
\begin{lemma}\label{lemmaF}
	The map  $F$ is defined as
	\begin{eqnarray}
	F:\Omega_3^{\tO}(\B\tO(3))\to \Omega_5^{\tO}(\B^2\Z_2)/\Omega_5^{\SO}(\B^2\Z_2)
	\end{eqnarray}
	sending $(N^3,E)\in \Omega_3^{\tO}(\B\tO(3))$ to $[(M^5,B)]\in \Omega_5^{\tO}(\B^2\Z_2)/\Omega_5^{\SO}(\B^2\Z_2)$. $F$ is well defined. 
\end{lemma}

\begin{proof}

For the map $F$ to be well-defined, the trivial element in $\Omega_3^{\tO}(\B\tO(3))$ must be mapped to the trivial element in $\Omega_5^{\tO}(\B^2\Z_2)/\Omega_5^{\SO}(\B^2\Z_2)$. Hence we need to prove that
\begin{eqnarray}
\left\{\begin{array}{llll}
\int_{N^3}w_1(E)^3=0,\\
\int_{N^3}w_1(E)w_2(E)=0,\\
\int_{N^3}w_3(E)=0,\\
\int_{N^3}w_1(E)w_1(TN^3)^2=0
\end{array}
\right.
\end{eqnarray}
implies 
\begin{eqnarray}
\left\{\begin{array}{ll}
\int_{M^5}B\Sq^1B+\Sq^2\Sq^1B=0,\\
\int_{M^5}w_1(TM^5)^2\Sq^1B=0
\end{array}
\right.
\end{eqnarray}

Since $TM^5\mid_{N^3}=TN^3\oplus \nu$ where $\nu$ is the normal bundle, 
by the Whitney sum formula for the total Stiefel-Whitney class, we have $w(TM^5\mid_{N^3})=w(TN^3)w(\nu)$.
So $w_1(TM^5)\mid_{N^3}=w_1(TN^3)+w_1(E)$ implies $w_1(E)=w_1(\nu)$, and $w_3(TM^5)\mid_{N^3}=w_3(TN^3)+w_2(TN^3)w_1(\nu)+w_1(TN^3)w_2(\nu)$, if $w_1(E)w_1(TN^3)^2=0$, then $w_2(TN^3)w_1(\nu)=0$ since $w_2(TN^3)=w_1(TN^3)^2$. Also since $w_3(TN^3)=0$, so $w_3(TM^5)\mid_{N^3}= w_1(TN^3)w_2(\nu)$.

{Consider $E\oplus \nu$, by splitting principle, the total Stiefel-Whitney class $w(E\oplus \nu)$ is the product of linear factors (of the form $1+x$ where $x$ is a degree-1 cohomology class). Since $w_1(E\oplus \nu)=w_1(E)+w_1(\nu)=0$, so $w_2(E\oplus \nu)$ is a 
sum of squares.\footnote{Suppose the total Stiefel-Whitney class $w(E\oplus \nu)=(1+x_1)^{n_1}(1+x_2)^{n_2}(1+x_3)^{n_3}$, then $w_1(E\oplus \nu)=n_1x_1+n_2x_2+n_3x_3=0$ implies $n_1=n_2=n_3=0\mod2$, so $n_i=2k_i$, and $w_2(E\oplus \nu)=k_1x_1^2+k_2x_2^2+k_3x_3^2$.} Since $w_2(E\oplus \nu)=w_2(E)+w_2(\nu)+w_1(E)w_1(\nu)=w_2(E)+w_2(\nu)+w_1(E)^2$, we have $\Sq^1w_2(E)=\Sq^1w_2(\nu)$. By Wu formula, we have $w_1(TN^3)w_2(\nu)=w_1(TN^3)w_2(E)$.}

By Wu formula, we have $\Sq^2\Sq^1B=(w_2(TM^5)+w_1(TM^5)^2)\Sq^1	B=(w_3(TM^5)+w_1(TM^5)^3)B$ and $w_1(TM^5)^2\Sq^1B=w_1(TM^5)^3B$.

We have
$\int_{M^5}w_1(TM^5)^2\Sq^1B'=\int_{M^5}w_1(TM^5)^3B'=\int_{N^3}(w_1(E)^3+w_1(TN^3)^3+w_1(E)w_1(TN^3)^2+w_1(TN^3)w_1(E)^2)=0$.

We also have 
$\int_{M^5}w_2(TM^5)\Sq^1B'=\int_{M^5}w_3(TM^5)B'=\int_{N^3}w_1(TN^3)w_2(\nu)=\int_{N^3}w_1(TN^3)w_2(E)=\int_{N^3}(w_1(E)w_2(E)+w_3(E))=0$ and
$\int_{M^5}B'\Sq^1B=\int_{N^3}\Sq^1(w_2(E)+\bar{K}_1w_1(E)^2)=\int_{N^3}(w_1(E)w_2(E)+w_3(E))=0$.

We also have
$\int_{M^5}w_1(TM^5)BB'=\int_{N^3}(w_1(TN^3)+w_1(E))(w_2(E)+\bar{K}_1w_1(E)^2)=0$.

By Wu formula, $\Sq^1(BB')=w_1(TM^5)BB'$, so $\int_{M^5}\Sq^1(BB')=\int_{M^5}(\Sq^1B)B'+B\Sq^1B'=0$, so $\int_{M^5}B\Sq^1B'=0$.

Since we impose the condition $\Sq^1B=\Sq^1B'$, we have proved the statement.

\end{proof}

We further discuss the reduction rules for the symmetry background fields in the case $\rN=4$. 
\begin{enumerate}
	\item $\cC\cT$ in 5d reduces to $\cC\cT$ in 3d. Correspondingly the background field for time reversal in 5d $w_1(TM^5)$ reduces to $w_1(TN^3)$. Moreover, the charge conjugation $\cC$ in 4d YM is a symmetry for any $\theta$, while $\cT$ in 2d $\CP^1$ model is also a symmetry for any $\theta$. Thus we demand that $A$ in 5d reduces to $w_1(TN^3)+ w_1(E)$ (which is the background for $\cT$) in 3d. Formally, we need to find an embedding $\iota:N^3\hookrightarrow M^5$ and $A\in \H^1(M^5,\Z_2)$  such that 
	\bea\label{embeddingN4}
	w_1(TM^5)\mid_{N^3}=w_1(TN^3),
	\eea
	and 
	\bea
	A\mid_{N^3}=w_1(TN^3)+w_1(E).
	\eea
	In summary, we find the symmetry reduction
	\begin{eqnarray}
	\text{5d}: (\cC\cT, \cC) \to \text{3d}: (\cC\cT, \cT). 
	\end{eqnarray}
	
	\item Following \Ref{Yamazaki:2017ulc}, the twisted boundary condition by the center symmetry (which is $B$) for the Yang-Mills is reduced to the twisted boundary condition by the $\Z_{\rN}$ global symmetry (which is $w_2(E)=w_2(V_{\PSU(4)})$ for $\rN=4$) of $\CP^{\rN-1}$.  Formally, we need to find $B\in \H^2(M^5,\Z_4)$ such that  
	\be
	\begin{split}
		B\mid_{N^3}=w_2(E).
	\end{split}
	\ee
\end{enumerate}

After choosing an embedding $\iota:N^3\hookrightarrow M^5$, we have the Poincar\'e dual $\text{PD}:\H^2(M^5,\Z_2)\xrightarrow{\sim}\H_3(M^5,\Z_2)$, we denote $B'=\text{PD}^{-1}(\iota_*([N^3]))$. We require that $B'$ is the cup product of two different degree-1 cohomology classes. {We impose this condition because, as discussed at the beginning of \Sec{sec:5d3d},  the 3d submanifold $N^3$ is reduced from the 5d manifold $M^5$ by a 2-torus.}
{We also require that $(\tilde B+B')\beta_{(2,4)}B=0$ and $Aw_1(TM^5)^2(\tilde B+B')=0$ to ensure that the map $G$ in Lemma \ref{lemmaG} is well-defined.} {Actually since
$TM^5\mid_{N^3}=TN^3\oplus\nu$ where $\nu$ is the normal bundle of rank 2, $w_1(TM^5)\mid_{N^3}=w_1(TN^3)$ implies $w_1(\nu)=0$. We claim that $\nu$ is a trivial bundle\footnote{Every orientable real line bundle is a trivial bundle. Since $B'$ is the cup product of two different degree-1 cohomology classes, $\nu$ is the direct sum of two different line bundles, each of them is orientable, thus a trivial bundle.}, thus
 $M^5=N^3\times T^2$,  the 3rd to 6th components of $\Phi_3''[(M^5,A,B)]$ are zero and $[(M^5,A,B)]\in W$.}

Now we construct the map $G$, which we organize as the following lemma. 
\begin{lemma}\label{lemmaG}
	The map $G$ is defined as
	$G:\Omega_3^{\tO}(\B(\Z_2\ltimes\PSU(4)))\to W\subset \Omega_5^{\tO}(\B\Z_2\ltimes\B^2\Z_4)/\Omega_5^{\SO}(\B\Z_2\ltimes\B^2\Z_4)$ sending $(N^3,w_1(E),w_2(E))\in \Omega_3^{\tO}(\B(\Z_2\ltimes\PSU(4)))$ to $[(M^5,A,B)]\in W$. $G$ is well defined. 
\end{lemma}

\begin{proof}
	
	For the map $G$ to be well-defined, the trivial element in $\Omega_3^{\tO}(\B(\Z_2\ltimes\PSU(4)))$ must be mapped to the trivial element in $W$. Hence we need to prove that
	\begin{eqnarray}
	\left\{\begin{array}{llll}
	\int_{N^3}w_1(E)^3=0,\\
	\int_{N^3}w_1(E)w_2(E)=0,\\
	\int_{N^3}\beta_{(2,4)}w_2(E)=0,\\
	\int_{N^3}w_1(E)w_1(TN^3)^2=0
	\end{array}
	\right.
	\end{eqnarray}
	implies
	\begin{eqnarray}
	\left\{\begin{array}{ll}
	\int_{M^5}B\beta_{(2,4)}B=0,\\
	\int_{M^5}Aw_1(TM^5)^2B=0
	\end{array}
	\right.
	\end{eqnarray}

	We have
	$\int_{M^5}B'\beta_{(2,4)}B=\int_{N^3}\beta_{(2,4)}w_2(E)=0$
	and $\int_{M^5}Aw_1(TM^5)^2B'=\int_{N^3}(w_1(TN^3)+w_1(E))w_1(TN^3)^2=0$.
	
	Since we impose the conditions $(\tilde B+B')\beta_{(2,4)}B=0$ and $Aw_1(TM^5)^2(\tilde B+B')=0$, we have proved the statement.

\end{proof}

\subsection{From $\Omega_3^{\tO}(\B\tO(3))$ to $\Omega_5^{\tO}(\B^2\Z_2)/\Omega_5^{\SO}(\B^2\Z_2)$}
\label{N=2reduction}

In this subsection, we use the linear dual of the map $F$ defined in Lemma \ref{lemmaF} to reduce the bordism  invariants of $\Omega_5^{\tO}(\B^2\Z_2)/\Omega_5^{\SO}(\B^2\Z_2)$ to the bordism invariants of $\Omega_3^{\tO}(\B\tO(3))$. 

We first consider {the 5d cobordism invariants that characterize the 4d SU(2) YM theory's anomaly}.
We may also name these 5d invariants as ``\emph{5d Yang-Mills terms},'' ``\emph{5d terms},'' ``\emph{Yang-Mills terms},'' ``\emph{5d anomaly polynomial of Yang-Mills}" 
or  ``\emph{5d iTQFTs whose boundary can live 4d Yang-Mills}." 

{
Based on the discussions around \eq{eq:select-N=2YM} in \Sec{sec:B} and the \ref{Rule 1} in \Sec{sec:rule}, we can safely propose that the 5d Yang-Mills term for $\rN=2$  is at most } 
\be \label{eq:4dYM-N=2-anom}
B\Sq^1B+\Sq^2\Sq^1B+K_1w_1(TM^5)^2\Sq^1B, 
\ee
where $K_1 \in \Z_2 = \{0,1\}$.

Amusingly, \Ref{Wan2019oyr1904.00994} actually derive \eqn{eq:4dYM-N=2-anom} based on putting 4d YM on unorientable manifolds, and then turning on background $B$ fields.
\Ref{Wan2019oyr1904.00994} also gives mathematical and physical interpretations of the $K_1$ term,
based on the gauge bundle constraint,
\bea \label{eq:gauge-bundle-N=2}
\boxed{w_2(V_{\SO(3)})=B+K_1w_1(TM^5)^2+K_2w_2(TM^5)}.\;\; \quad
\eea

{
We should emphasize that our approach in this work to derive this possible term \eqn{eq:4dYM-N=2-anom} 
is sharply distinct from \Ref{Wan2019oyr1904.00994}, although we obtain the same result!
Although the starting definitions of the $K_1$ in \eqn{eq:4dYM-N=2-anom} 
and the $K_1$ in \eqn{eq:gauge-bundle-N=2}  (and in \Ref{Wan2019oyr1904.00994})
are distinct,
below we should also derive that the two $K_1$ are actually equivalent.
Thus, we use the same label for both $K_1$.
}

We show that, using the linear dual of the map $F$,  the \emph{5d Yang-Mills term} in \eqn{eq:4dYM-N=2-anom} reduces to the anomaly polynomial in 3d in theorem \ref{5d3dN2}. 
\begin{widetext}
\begin{theorem}\label{5d3dN2}
	The 5d anomaly polynomial for the SU(2) YM theory 
	\begin{eqnarray}\label{N2result5d}
	{B\Sq^1B+\Sq^2\Sq^1B+K_1w_1(TM^5)^2\Sq^1B}
	\end{eqnarray}
	reduces to the anomaly polynomial of 2d $\CP^1$ theory
		\bea \label{N2result}
		&&(\bar{K}_1+1)w_1(E)^3+w_3(E)+K_1(w_1(E)^3+w_1(E)w_1(TN^3)^2)\nn\\
		&=&{\bar{K}_1w_1(E)^3+w_1(TN^3)w_2(V_{\SO(3)})+w_1(E)w_2(V_{\SO(3)})+K_1(w_1(E)^3+w_1(E)w_1(TN^3)^2)}
		\eea
\end{theorem}
\end{widetext}

\begin{proof}

We first define the following notations to simplify the proof. 
\begin{enumerate}
	\item  $\alpha$ is the generator of the cohomology $\H^1(\RP^2,\Z_2)$.
	\item  $\beta$ is the generator of $\H^1(\RP^3,\Z_2)$.
	\item $\gamma$ is the generator of $\H^1(S^1,\Z_2)$.
	\item $\zeta$ is the generator of $\H^1(\RP^4,\Z_2)$.
\end{enumerate}
Recall that the manifold generators of $\Omega_3^{\tO}(\B\tO(3))$ are 
\bea
(N^3, E )=\left\{\begin{array}{llll}(\RP^3,l_{\RP^3}+2)\quad\textcircled{1},\\
(\RP^3,3l_{\RP^3})\quad\textcircled{2},\\
(S^1\times\RP^2,l_{S^1}+l_{\RP^2}+1)\quad\textcircled{3},\\
(S^1\times\RP^2,l_{S^1}+2)\quad\textcircled{4}.
\end{array}
\right.
\eea
Using the definitions in \Sec{sec:set-up}, we find $\Phi_2(\textcircled{1})=(1,0,0,0)$, $\Phi_2(\textcircled{2})=(1,1,1,0)$, $\Phi_2(\textcircled{3})=(1,1,0,1)$, $\Phi_2(\textcircled{4})=(0,0,0,1)$.
Thus $e_1=\textcircled{1}$, $e_2=\textcircled{1}+\textcircled{3}+\textcircled{4}$, $e_3=\textcircled{2}+\textcircled{3}+\textcircled{4}$, $e_4=\textcircled{4}$.

{
Following the construction of the map $F$ in Lemma \ref{lemmaF}, we have:\footnote{Note that $\textcircled{5}$ $(N^3,E)=(S^1\times\RP^2,l_{S^1}+2l_{\RP^2})$ is also a manifold generator of $\Omega_3^{\tO}(\B\tO(3))$, actually $\textcircled{5}=e_2+e_3+e_4$.
	Since $w_1(TN^3)=\alpha$, $w_1(E)=\gamma$, $w_2(E)=\alpha^2$, $w_2(E)+\bar{K}_1w_1(E)^2=\alpha^2$, 
	$w_1(TN^3)+w_1(E)=\gamma+\alpha$,
	we have
	$M^5=S^1\times\RP^2\times\RP^2$, $B=\gamma\alpha_1+\alpha_2^2$ where $B'=\gamma\alpha_1$.
	So $F(\textcircled{5})=f_1+f_2=F(e_2+e_3+e_4)$. 
	So $F$ is indeed well-defined.}
}

$\textcircled{1}$ $(N^3,E)=(\RP^3,l_{\RP^3}+2)$:

Since $w_1(TN^3)=0$, $w_1(E)=\beta$, $w_2(E)=0$, $w_2(E)+\bar{K}_1w_1(E)^2=\bar{K}_1\beta^2$, $w_1(TN^3)+w_1(E)=\beta$,
we have $M^5=S^1\times\RP^4$, $w_1(TM^5)=\zeta$, $B=\gamma\zeta+\bar{K}_1\zeta^2$ where $B'=\gamma\zeta$.

$\textcircled{2}$ $(N^3,E)=(\RP^3,3l_{\RP^3})$:

Since $w_1(TN^3)=0$, $w_1(E)=\beta$, $w_2(E)=\beta^2$, $w_2(E)+\bar{K}_1w_1(E)^2=(\bar{K}_1+1)\beta^2$, $w_1(TN^3)+w_1(E)=\beta$,
we have $M^5=S^1\times\RP^4$, $w_1(TM^5)=\zeta$, $B=\gamma\zeta+(\bar{K}_1+1)\zeta^2$ where $B'=\gamma\zeta$.

$\textcircled{3}$ $(N^3,E)=(S^1\times\RP^2,l_{S^1}+l_{\RP^2}+1)$:

Since $w_1(TN^3)=\alpha$, $w_1(E)=\gamma+\alpha$, $w_2(E)=\gamma\alpha$, $w_2(E)+\bar{K}_1w_1(E)^2=\gamma\alpha+\bar{K}_1\alpha^2$, $w_1(TN^3)+w_1(E)=\gamma$,
we have $M^5=\RP^2\times\RP^3$, $w_1(TM^5)=\alpha$, $B=\alpha\beta+\bar{K}_1\beta^2$ where $B'=\alpha\beta$.

$\textcircled{4}$ $(N^3,E)=(S^1\times\RP^2,l_{S^1}+2)$:

Since $w_1(TN^3)=\alpha$, $w_1(E)=\gamma$, $w_2(E)=0$, $w_2(E)+\bar{K}_1w_1(E)^2=0$, $w_1(TN^3)+w_1(E)=\gamma+\alpha$,
we have $M^5=S^1\times\RP^2\times\RP^2$, $w_1(TM^5)=\alpha_1+\alpha_2$, $B=\gamma\alpha_1$ where $B'=\gamma\alpha_1$.

Recall that in \Sec{sec:set-up}
$f_1=\Phi_1'^{-1}(1,0)$, $f_2=\Phi_1'^{-1}(0,1)$, we have
$F(\textcircled{1})=(\bar{K}_1+1)f_1+f_2$, $F(\textcircled{2})=\bar{K}_1f_1+f_2$, $F(\textcircled{3})=(\bar{K}_1+1)f_1$, $F(\textcircled{4})=f_2$.

So
$F(e_1)=(\bar{K}_1+1)f_1+f_2$, $F(e_2)=0$, $F(e_3)=f_1$, $F(e_4)=f_2$, and 
$F^*(f_1^*)=(\bar{K}_1+1)e_1^*+e_3^*$, $F^*(f_2^*)=e_1^*+e_4^*$.

\end{proof}

{Compared with the 5d anomaly polynomial \eqn{N2result5d}, the 3d anomaly polynomial \eqn{N2result} contains an extra parameter $\bar{K}_1$. In the following, we argue that $\bar{K}_1=1$ by comparing with the results in the literature. 
\begin{enumerate}
	\item In \Ref{Metlitski2017fmd1707.07686}, the authors computed the anomaly for the $\mathbb{CP}^1$ model with the $\Z_2^x(\equiv \Z_2^C)$-translation symmetry and the $\SO(3)$ symmetry, on an oriented manifold. Denote $T_x$ as the generator of $\Z_2^x$. Because in \Ref{Metlitski2017fmd1707.07686} $T_x^2$ acts trivially on all physical observables , in our notation they only considered the case $K_1=0$.    The $\Z_2^x$ symmetry background field is $w_1(E)$, and the $\SO(3)$ background field is $w_2(V_{\SO(3)})$.  \Ref{Metlitski2017fmd1707.07686} found the following anomaly polynomial
	\begin{equation}\label{1707}
	w_1(E)^3+ w_1(E) w_2(V_{\SO(3)})
	\end{equation} 
	By setting $w_1(TN^3)=0$ (because \Ref{Metlitski2017fmd1707.07686} only discussed oriented manifold) and $K_1=0$ in \eqn{N2result} and comparing with \eqn{1707}, we find $\bar{K}_1=1$. 
	\item In \Ref{Komargodski2017dmc1705.04786}, the authors computed the anomaly for the $\mathbb{CP}^1$ model with the $\Z_2^x(\equiv \Z_2^C)$-translation symmetry, the $\SO(3)$ symmetry, and time reversal symmetry on an unorientable manifold. Using the same notation as above and furthermore we denote the time reversal background as $w_1(TN^3)$, \Ref{Komargodski2017dmc1705.04786} found the anomaly polynomial
	\begin{equation}
	w_1(E)^3+ w_1(E) w_2(V_{\SO(3)}) + w_1(TN^3) w_2(V_{\SO(3)}) 
	\end{equation}
	By setting $K_1=0$ in \eqn{N2result}, this again requires $\bar{K}_1=1$. 
\end{enumerate}
}

\subsection{From $\Omega_3^{\tO}(\B(\Z_2\ltimes\PSU(4)))$ to $W\subset \Omega_5^{\tO}(\B\Z_2\ltimes\B^2\Z_4)/\Omega_5^{\SO}(\B\Z_2\ltimes\B^2\Z_4)$}\label{N=4reduction}

In this subsection, we use the linear dual of the map $G$ defined in Lemma \ref{lemmaG} to reduce the bordism  invariants of $W$ to the bordism invariants of $\Omega_3^{\tO}(\B(\Z_2\ltimes\PSU(4)))$. 

We first consider {the 5d cobordism invariants that characterize the 4d SU(4) YM theory's anomaly}
(abbreviate them as ``\emph{5d Yang-Mills terms}'').
{
Based on the discussions around \eq{eq:select-N=4YM} in \Sec{sec:F} and the \ref{Rule 1} and \ref{Rule 7} in \Sec{sec:rule}, we can safely propose that the 5d Yang-Mills term for $\rN=4$ is at most
\bea
&&B\beta_{(2,4)}B+K_1' w_1(TM^5)^2\beta_{(2,4)}B+K_C' A^2\beta_{(2,4)}B\nn\\
&&+K_{C,1}' Aw_1(TM^5)^2B, 
\eea
{where $K_1',K_C',K_{C,1}' \in \Z_2 = \{0,1\}$.
$K_1',K_C',K_{C,1}'$ are distinct couplings different from the gauge bundle constraint couplings $K_1,K_C,K_{C,1}$ later in \Eq{eq:gauge-bundle-constraint-N=4}.}

In \Sec{sec:Set-UpforN=4}, we also mentioned the appearances of $w_1(TM^5)^2\beta_{(2,4)}B$ and $A^2\beta_{(2,4)}B$ for the anomaly of 4d SU(4) YM are unlikely and the full discussion is left for the future work \cite{WWZ2019-2}.}
Following the derivation in \Ref{Wan2019oyr1904.00994},
we find that the 5d Yang-Mills term for $\rN=4$ is 
\be\label{N4YM}
B\beta_{(2,4)}B+K_{C,1}Aw_1(TM^5)^2B
\ee
where $K_{C,1}$ is from the gauge bundle constraint 
\bea \label{eq:gauge-bundle-constraint-N=4}
&&w_2(V_{\PSU(4)})=B+2(K_1w_1(TM^5)^2+K_2w_2(TM^5)\nn\\
&&+K_CA^2+K_{C,1}Aw_1(TM^5))\mod4.
\eea
The detailed derivation will be left for the future work \cite{WWZ2019-2}. 

In the following, we show that, using the linear dual of the map $G$,  the \emph{5d Yang-Mills term} in \eqn{N4YM} reduces to the anomaly polynomial in 3d in theorem \ref{5d3dN4}. 
\begin{widetext}
	\begin{theorem}\label{5d3dN4}
		The 5d anomaly polynomial for the SU(4) YM theory 
		\begin{eqnarray}\label{N4result5d}
		B\beta_{(2,4)}B+K_{C,1}Aw_1(TM^5)^2B
		\end{eqnarray}
		reduces to the anomaly polynomial of 2d $\CP^3$ theory
		\bea \label{N4result}
		\beta_{(2,4)}w_2(E)+ K_{C,1}w_1(E)w_1(TN^3)^2
		={\beta_{(2,4)}w_2(V_{\PSU(4)})+K_{C,1}w_1(E)w_1(TN^3)^2.}
		\eea
	\end{theorem}
\end{widetext}
\begin{proof}

For simplicity, we define the following notations:
\begin{enumerate}

\item $K$ is the Klein bottle.
\item $\alpha'$ is the generator of $\H^1(S^1,\Z_4)=\Z_4$.
\item $\beta'$ is the generator of the $\Z_4$ factor of $\H^1(K,\Z_4)=\Z_4\times\Z_2$ (see Appendix \ref{Klein}).  Note that $(\beta'\mod2)^2=2\beta_{(2,4)}\beta'=0$.
\item $\alpha$ is the generator of $\H^1(\RP^2,\Z_2)=\Z_2$.
\item $\gamma$ is the generator of $\H^1(S^1,\Z_2)=\Z_2$.

\end{enumerate}

Using the definitions in \Sec{sec:set-up}, recall that the manifold generators of $\Omega_3^{\tO}(\B(\Z_2\ltimes\PSU(4)))$ are
\bea
(N^3, w_1(E), w_2(E))=\left\{\begin{array}{llll}(\RP^3,\beta,0)=e_1,\\
(T^3,\gamma_1,\alpha_2'\alpha_3')=e_2,\\
(S^1\times K, 0, \alpha'\beta')=e_3,\\
(S^1\times\RP^2,\gamma,0)=e_4.
\end{array}
\right.
\eea

Following the construction of the map $G$ in Lemma \ref{lemmaG}, we have:

\begin{enumerate}
\item
$(N^3, w_1(E), w_2(E))=(\RP^3,\beta,0)$, $w_1(TN^3)=0$,
we have $M^5=T^2\times\RP^3$, since both $B\beta_{(2,4)}B$ and $Aw_1(TM^5)^2B$ must vanish on $M^5$, so $G(e_1)=0$.

\item
$(N^3, w_1(E), w_2(E))=(T^3,\gamma_1,\alpha_2'\alpha_3')$, $w_1(TN^3)=0$,
we have $M^5=T^5$, since both $B\beta_{(2,4)}B$ and $Aw_1(TM^5)^2B$ must vanish on $M^5$, so $G(e_2)=0$.

\item
$(N^3, w_1(E), w_2(E))=(S^1\times K, 0, \alpha'\beta')$, $w_1(TN^3)=\beta'\mod2$,
 we have $M^5=K\times T^3$, $w_1(TM^5)=\beta'\mod2$, $A=\beta'\mod2$, $B=\alpha_1'\beta'+\alpha_2'\alpha_3'$, where $B'=\alpha_2'\alpha_3'\mod2$.
So $G(e_3)=f_1$.

\item
$(N^3, w_1(E), w_2(E))=(S^1\times\RP^2,\gamma,0)$, $w_1(TN^3)=\alpha$,
we have $M^5=T^3\times \RP^2$, $w_1(TM^5)=\alpha$, $A=\gamma_1+\alpha$, $B=\alpha_2'\alpha_3'$, where $B'=\alpha_2'\alpha_3'$.
So $G(e_4)=f_2$.

\end{enumerate}
So $G^*(f_1^*)=e_3^*$, $G^*(f_2^*)=e_4^*$.

\end{proof}

Next we can elaborate the new higher anomaly of 4d YM theory in \Sec{sec:newSUn}.

\section{New Higher Anomalies of 4d SU($\rN$)-YM Theory}

\label{sec:newSUn}

We provide more details on the anomaly of 4d YM theory. 
We 
deduce the new higher anomaly of 4d YM theory written in terms of invariants given in \Sec{sec:cobo-top}, 
and satisfying Rules in \Sec{sec:rule} and following the physical/mathematical 5d to 3d reduction scheme in \Sec{sec:5d3d}.

\subsection{SU($\rN$)-YM at $\rN=2$} 

\label{sec:newSU2YM}

Let us formulate the potentially complete 't Hooft anomaly for 
4d SU($\rN$)-YM for $\rN=2$ at ${\theta=\pi}$, written in terms of a 5d cobordism invariant in \Sec{sec:cobo-top}.

Based on \ref{Rule 3} and \ref{Rule 6} in \Sec{sec:rule}, we deduce that
4d anomaly must match {2d $\mathbb{CP}^{1}$-model anomaly}'s \eq{eq:2d-anomaly-CP1}
via the sum of following two terms (5d SPTs). The first term is:
\bea\label{eq:w1PB}
&&B_2\Sq^1B_2+\Sq^2\Sq^1B_2 \\
&&
{= \frac{1}{2} \tilde w_1(TM) \cP_2(B_2).}\nn
\eea
which is dictated by \ref{Rule 1}  in \Sec{sec:rule}.
(Note that {$\Sq^2\Sq^1B_2$}{$=(B_2 \cup_1 B_2) \cup_1 (B_2 \cup_1 B_2)$}.)
Here $\tilde w_1(TM)\in\H^1(M,\Z_{4,w_1})$ is the mod 4 reduction of the twisted first Stiefel-Whitney class of the tangent bundle $TM$ of a 5-manifold $M$ which is the pullback of $\tilde w_1$ under the classifying map $M\to\B\tO(5)$.
Here $\Z_{w_1}$ denotes the orientation local system, the twisted first Stiefel-Whitney class $\tilde w_1\in \H^1(\B\tO(n),\Z_{w_1})$ is the pullback of the nonzero element of $\H^1(\B\tO(1),\Z_{w_1})=\Z_2$ under the determinant map $\B\det:\B\tO(n)\to\B\tO(1)$. 
{Since $2\tilde w_1(TM)=0$ mod 4, }
$\tilde w_1(TM)\cP_2(B_2)$ is even, so it makes sense to divide it by 2.
If $w_1(TM)=0$, then $\Z_{w_1}=\Z$ and $\H^1(\B\tO(1),\Z_{w_1})=\H^1(\B\tO(1),\Z)=0$, so $\tilde w_1=0$. Namely, $\frac{1}{2}\tilde w_1(TM)\cP_2(B_2)$ vanishes when $w_1(TM)=0$.

We can derive the last equality of \eq{eq:w1PB} by proving that both LHS and RHS are bordism invariants of $\Omega_5^{\tO}(\B^2\Z_2)$ and they coincide on manifold generators of $\Omega_5^{\tO}(\B^2\Z_2)$.

We can also prove that 
{
\bea
\beta_{(2,4)}\cP_2(B_2)&=&\frac{1}{4}\delta\cP_2(B_2)\mod2\nn\\
&=&\frac{1}{4}\delta(B_2\cup B_2+B_2\hcup{1}\delta B_2)\nn\\
&=&\frac{1}{4}(\delta B_2\cup B_2+B_2\cup\delta B_2+\delta(B_2\hcup{1}\delta B_2))\nn\\
&=&\frac{1}{4}(2B_2\cup\delta B_2+\delta B_2\hcup{1}\delta B_2)\nn\\
&=&B_2\cup(\frac{1}{2}\delta B_2)+(\frac{1}{2}\delta B_2)\hcup{1}(\frac{1}{2}\delta B_2)\nn\\
&=&B_2\Sq^1B_2+\Sq^1B_2\hcup{1}\Sq^1B_2\nn\\
&=&B_2\Sq^1B_2+\Sq^2\Sq^1B_2.
\eea
Here we have used $\beta_{(2,4)}=\frac{1}{4}\delta\mod2$,
\bea
\delta(u\hcup{1}v)=u\cup v-v\cup u+\delta u\hcup{1}v+u\hcup{1}\delta v
\eea
for 2-cochain $u$ and 3-cochain $v$ \cite{Steenrod1947}, 
$\Sq^1=\beta_{(2,2)}=\frac{1}{2}\delta\mod2$,
and $\Sq^kz_n=z_n\hcup{n-k}z_n$.
}
The first term contains two appear together in order to satisfy \ref{Rule 2}.

The other term is: 
\bea
w_1(TM)^2\Sq^1B_2.
\eea
We also check that the sum of two terms satisfy the \ref{Rule 5} in \Sec{sec:rule}.
Besides, \ref{Rule 7} restricts us to focus on the bordism group $\Omega_5^{\tO}(\B^2\Z_2)$ and discards other terms involving $\Omega_5^{\tO}(\B \Z_2 \times \B^2\Z_2)$.
Our final answer of 4d anomaly and 5d cobordism/SPTs invariant is  combined and given in
\eq{eq:SU2YM-5dSPT}:
\bea \label{eq:SU2YM-5dSPT-0}
B_2\Sq^1B_2+\Sq^2\Sq^1B_2+K_1w_1(TM)^2\Sq^1B_2. 
\eea
To our understanding, the whole expression indicates a new higher anomaly for this YM theory, 
which turns out to be new to the literature. 

\subsection{SU($\rN$)-YM at $\rN=4$} 

{Let us propose some 't Hooft anomaly for 
4d SU($\rN$)-YM at $\rN=4$ at ${\theta=\pi}$, written in terms of a 5d cobordism invariant in \Sec{sec:cobo-top}.
Here we \emph{do not} claim to have a complete set of 't Hooft anomaly.
As at $\rN=4$ we need to specify:
\begin{enumerate} 
\item the gauge bundle constraint \Eq{eq:gauge-bundle-constraint-N=4}.
\item the appropriate \emph{charge conjugation symmetry background field} coupling to YM.
\end{enumerate} 
However, we only have a potentially complete gauge bundle constraint \Eq{eq:gauge-bundle-constraint-N=4},
but we do not yet know whether we have captured all possible \emph{charge conjugation symmetry background field} coupling to YM.
The second issue will be left in the future work. 
}

Based on \ref{Rule 4} in \Sec{sec:rule}, we deduce the  {2d $\mathbb{CP}^{3}$-model anomaly}'s \eq{eq:2d-anomaly-CP3} generalizing the \eq{eq:2d-anomaly-CP1}.
Based on \ref{Rule 3} and \ref{Rule 6}, we deduce that
4d anomaly must match {2d $\mathbb{CP}^{3}$-model anomaly}'s \eq{eq:2d-anomaly-CP3}
via the sum of following two terms (5d SPTs). The first term is:
\bea
B_2\beta_{(2,4)}B_2
&=&\frac{1}{4}\tilde w_1(TM)\cP_2(B_2),
\eea
which is dictated by \ref{Rule 1}  in \Sec{sec:rule}.
Here $\tilde w_1(TM)\in\H^1(M,\Z_{8,w_1})$ is the mod 8 reduction of the twisted first Stiefel-Whitney class of the tangent bundle $TM$ of a 5-manifold $M$ which is the pullback of $\tilde w_1$ under the classifying map $M\to\B\tO(5)$.
Here $\Z_{w_1}$ denotes the orientation local system, the twisted first Stiefel-Whitney class $\tilde w_1\in \H^1(\B\tO(n),\Z_{w_1})$ is the pullback of the nonzero element of $\H^1(\B\tO(1),\Z_{w_1})=\Z_2$ under the determinant map $\B\det:\B\tO(n)\to\B\tO(1)$. Since $2\tilde w_1=0$, $\tilde w_1(TM)\cP_2(B_2)$ is divided by 4, so it makes sense to divide it by 4.
If $w_1(TM)=0$, then $\Z_{w_1}=\Z$ and $\H^1(\B\tO(1),\Z_{w_1})=\H^1(\B\tO(1),\Z)=0$, so $\tilde w_1=0$. Namely, $\frac{1}{4}\tilde w_1(TM)\cP_2(B_2)$ vanishes when $w_1(TM)=0$.

We can derive the last equality by proving that both LHS and RHS are bordism invariants of $\Omega_5^{\tO}(\B^2\Z_4)$ and they coincide on manifold generators of $\Omega_5^{\tO}(\B^2\Z_4)$.

We can also prove that
{
\bea
\beta_{(2,8)}\cP_2(B_2)&=&\frac{1}{8}\delta\cP_2(B_2)\mod2\nn\\
&=&\frac{1}{8}\delta(B_2\cup B_2+B_2\hcup{1}\delta B_2)\nn\\
&=&\frac{1}{8}(\delta B_2\cup B_2+B_2\cup\delta B_2+\delta(B_2\hcup{1}\delta B_2))\nn\\
&=&\frac{1}{8}(2B_2\cup\delta B_2+\delta B_2\hcup{1}\delta B_2)\nn\\
&=&B_2\cup(\frac{1}{4}\delta B_2)+2(\frac{1}{4}\delta B_2)\hcup{1}(\frac{1}{4}\delta B_2)\nn\\
&=&B_2\beta_{(2,4)}B_2+2\beta_{(2,4)}B_2\hcup{1}\beta_{(2,4)}B_2\nn\\
&=&B_2\beta_{(2,4)}B_2+2\Sq^2\beta_{(2,4)}B_2\nn\\
&=&B_2\beta_{(2,4)}B_2
\eea
which is dictated by \ref{Rule 1}  in \Sec{sec:rule}.
(Note that {$\tilde B_2=B_2\mod2$}.)
Here we have used $\beta_{(2,8)}=\frac{1}{8}\delta\mod2$,
\bea
\delta(u\hcup{1}v)=u\cup v-v\cup u+\delta u\hcup{1}v+u\hcup{1}\delta v
\eea
for 2-cochain $u$ and 3-cochain $v$ \cite{Steenrod1947},
$\beta_{(2,4)}=\frac{1}{4}\delta\mod2$,
and $\Sq^kz_n=z_n\hcup{n-k}z_n$.
}

The other term is: 
\bea
Aw_1(TM)^2B_2.
\eea
We also check that the sum of two terms satisfy the \ref{Rule 2} and \ref{Rule 5} in \Sec{sec:rule}.
By imposing  \ref{Rule 7}, we can rule out thus discard many other 5d terms in the bordism group $\Omega_5^{\tO}(\B\Z_2\ltimes\B^2\Z_4)$.
In summary, our final answer of 4d anomaly and 5d cobordism/SPTs invariant is combined and given in
\eq{eq:SU4YM-5dSPT}:
\be \label{eq:SU4YM-5dSPT-0}
B_2 \beta_{(2,4)}B_2+K_{C,1}Aw_1(TM)^2B_2.
\ee
To our understanding, the whole expression indicates a new higher anomaly for this YM theory, 
new to the literature.


\begin{widetext}

\section{New Anomalies of 2d $\mathbb{CP}^{\rN-1}$-model}
\label{sec:newCPn}

In this section, we provide more details and summarize the anomaly for the $\mathbb{CP}^{\rN-1}$-model. 

For 2d  $\mathbb{CP}^{1}$-model at $\theta=\pi$, in theorem \ref{5d3dN2}, we find that the 't Hooft anomaly is the combination of the cobordism invariants
\eq{eq:anom-stacking-Haldane},  \eq{eq:w3E}, \eq{eq:CP1-anom-A3} and \eq{eq:w1TMw2SO3}, which we repeat for readers' convenience:
\bea \label{eq:2d-anomaly-CP1}
\text{2d $\mathbb{CP}^{1}$-model anomaly : } &&\boxed{
 w_1(E)^3+w_1(TN^3)w_2(V_{\SO(3)})+w_1(E) w_2(V_{\SO(3)})
+ K_1(w_1(E)^3+w_1(E) w_1(TN^3)^2) }
\nn 
\\  &&
=w_3(E)+ K_1( w_1(E)^3+w_1(E)w_1(TN^3)^2).
\eea
Recall that,
under the basis $(w_1(E)^3,w_1(E)w_2(E), w_3(E), w_1(E)w_1(TN^3)^2)$,
we can express the following cobordism invariants using \eq{eq:bordism3OBO-basis}
$$
\left\{\begin{array}{l}
\text{$w_1(E)^3$ is $(1,0,0,0)$, }\\
\text{$w_1(E) w_2(V_{\SO(3)})$ is $(1,1,0,0)$,}\\
\text{$w_1(TN^3) w_2(V_{\SO(3)})$ is $(0,1,1,0)$,}\\
\text{$w_1(E) w_1(TN^3)^2$ is $(0,0,0,1)$,} \\
\text{$w_3(E)$ is $(0,0,1,0)$.}
\end{array}
\right.
$$
To summarize, the overall anomaly of
2d $\mathbb{CP}^{1}$-model can be expressed as a 3d cobordism invariant/topological term 
\eq{eq:2d-anomaly-CP1}, 
which is $(K_1,0,1,K_1)$ under the basis $(w_1(E)^3,w_1(E)w_2(E), w_3(E), w_1(E)w_1(TN^3)^2)$ of \eq{eq:bordism3OBO-basis}.

{For 2d  $\mathbb{CP}^{\rN-1}$-model at $\theta=\pi$, at even N,
Ref.~\cite{Komargodski2017dmc1705.04786} proposes an important quantity (called $u_3$ in Ref.~\cite{Komargodski2017dmc1705.04786}),  
which is an element  $u_3 \in \H^3(\B ({\PSU(\rN) \rtimes \Z_2^C}),\Z^C) $ as an anomaly for that 2d theory.
First we notice that one needs to generalize the second SW class from $w_2 \in \H^2(\B {\PSU(2)},\Z_2)=\Z_2 $ to $\tilde w_2 \in \H^2(\B {\PSU(\rN) },\Z_\rN)=\Z_\rN$. 
Moreover, there is an additional $\Z_2^C$ twist modifying the ${\PSU(2)}$-bundle to ${\PSU(\rN) \rtimes \Z_2^C}$-bundle.
In the definition of $u_3\in \H^3(\B(\PSU(\rN)\rtimes \Z_2^C),\Z^C) = \Z_\rN$, 
$C$ specifies the symmetry as a charge conjugation $\Z_2^C$.
This means that $\dd u_3 \neq 0$, but $\dd_A u_3 = 0$, where $\dd_A$ is a twisted differential.
The construction of these classes is a Bockstein operator for the extension 
applied to $u_2 \in \H^2(\B(\PSU(\rN) \rtimes \Z_2^C), \Z_\rN^C)$.
Eventually, the 3d invariant for the 2d anomaly term of Ref.~\cite{Komargodski2017dmc1705.04786} 
is $u_3 \in \H^3(\B(\PSU(\rN) \rtimes \Z_2^C), \U(1)) = \Z_2$.
}
 
{
In our setup, we consider $\tilde w_3(E) \equiv  \tilde w_3(V_{\PSU(\rN) \rtimes \Z_2})  \in \H^3(\B(\PSU(\rN) \rtimes \Z_2^C), \Z_2) = \Z_2$ 
here $E$ is the background gauged bundle of ${\PSU(\rN) \rtimes \Z_2}$.\\
For N = 2, we derive that $\tilde w_3(E) = w_3(E)  {=w_1(E)w_2(E)+ w_1(TN)w_2(E) }$ in \eq{eq:w3E}. We emphasize that $w_1(TN)$ and $w_1(E)$ are the symmetry background fields for $\cT$ and $\cC$ respectively.

For N = 4, in theorem \ref{5d3dN4}, we find that the 3d anomaly polynomial is
\begin{eqnarray}
\beta_{(2,4)}w_2(E)+ K_{C,1}w_1(E)w_1(TN^3)^2
={\beta_{(2,4)}w_2(V_{\PSU(4)})+K_{C,1}w_1(E)w_1(TN^3)^2.}
\end{eqnarray}
Let us remind the notations explained in \Sec{sec:G}. When N = 4, we have $E$ is the principal ${\Z_2 \ltimes\PSU(4)}$ bundle, 
while
$w_2(E) \in \H^2(M,\Z_{4,w_1(E)})$ is  a $\Z_4$-valued second twisted cohomology class, 
$w_1(E) \in \H^1(M,\Z_{2})$ is a group homomorphism $\pi_1(M)\to\text{Aut}(\Z_4)=\Z_2$.

We emphasize that when $K_{C,1}=0$, our anomaly polynomial 
\begin{eqnarray}\label{N43d}
\beta_{(2,4)}w_2(E)
={\beta_{(2,4)}w_2(V_{\PSU(4)})}
\end{eqnarray}
is consistent with the result in \Ref{Komargodski2017dmc1705.04786}. They derived that the anomaly polynomial is 
\bea\label{1705}
\tilde w_3(E)=\frac{1}{2}w_1(E) w_2(E)+\beta_{(2,4)}w_2(E) 
=\frac{1}{2}w_1(E) w_2(E)+\frac{1}{2} \tilde w_1(TN^3)w_2(E).
\eea
Compared with \eqn{N43d} there is an additional $\frac{1}{2}w_1(E) w_2(E)$ in \eqn{1705}. This superficial mismatch is because $w_{1}(TN^3)$ is identified as the background field for $\cC\cT$ in our work, while $\cT$ in \Ref{Komargodski2017dmc1705.04786}. If we replace $w_{1}(TN^3)$ in \eqn{N43d} by $w_1(TN^3)+ w_1(E)$, we correctly obtain \eqn{1705}.

}

Based on \ref{Rule 4} in \Sec{sec:rule}, we propose that 3d invariant for the anomaly of 2d $\mathbb{CP}^{3}$-model is:
\bea \label{eq:2d-anomaly-CP3}
\begin{split}
\text{2d $\mathbb{CP}^{3}$-model anomaly : } 
	& \boxed{
 \beta_{(2,4)} w_2(E)+K_{C,1}w_1(E)w_1(TN^3)^2}\\
= &
 \frac{1}{2} \tilde w_1(TN^3) w_2(E)
+K_{C,1}w_1(E)w_1(TN^3)^2
  \\
=& 
 \boxed{
 \frac{1}{2} \tilde w_1(TN^3) w_2(V_{{\PSU(4)}})
+K_{C,1}w_1(E)w_1(TN^3)^2.}
\end{split}
\eea

We should mention our anomaly term contains the previous anomaly found in the literature for more generic even N \cite{Komargodski2017dmc1705.04786, Yao2018kelHsiehOshikawa1805.06885, {Ohmori2018qza1809.10604}}. 

\end{widetext}


%


\section{Symmetric TQFT, Symmetry-Extension and Higher-Symmetry Analog of Lieb-Schultz-Mattis theorem}
\label{sec:sTQFT}

Since we know the potentially complete 't Hooft anomalies of
the above 
4d SU($\rN$)-YM and 2d $\CP^{\rN-1}$-model at ${\theta=\pi}$,
we wish to constrain their low-energy dynamics further, based on the anomaly-matching.
This thinking can be regarded as 
a formulation of a higher-symmetry analog of ``Lieb-Schultz-Mattis theorem \cite{Lieb:1961fr} \cite{Hastings:2003zx}.''
For example, the consequences of low-energy dynamics, under the anomaly saturation can be:\\

\noindent
\begin{enumerate} 
\item Symmetry-breaking:\\ 
$\bullet$ (say $C T$- or $T$-symmetry or other discrete or continuous $G$-symmetry breaking).\\
\item Symmetry-preserving: \\
$\bullet$ Gapless, conformal field theory (CFT),\\
$\bullet$ Intrinsic topological orders.\\ 
\; \quad\quad (Symmetry-preserving TQFT)\\
$\bullet$ Degenerate ground states.\\
etc.\\
\item Symmetry-extension \cite{Wang2017locWWW1705.06728}:
Symmetry-extension is another exotic possibility, which does not occur naturally without fine-tuning or artificial designed, explained in \cite{Wang2017locWWW1705.06728}.
However, symmetry-extension is a useful intermediate step, 
to obtain another earlier scenario: \emph{symmetry-preserving TQFT}, via gauging the extended-symmetry.
\end{enumerate}

Recently Lieb-Schultz-Mattis theorem has been applied to
higher-form symmetries acting on extended objects, see \cite{Kobayashi2018yukRyu1805.05367} and references therein.

In this section, we like to ask, whether it is possible to have a fully symmetry-preserving TQFT
to saturate the higher anomaly we discussed earlier, for 4d SU($\rN$)-YM and 2d $\CP^{\rN-1}$-model?
We use the systematic approach of \emph{symmetry-extension} method developed 
in Ref.~\cite{Wang2017locWWW1705.06728}. We will consider its
generalization to \emph{higher-symmetry-extension} method, also developed 
in our parallel work Ref.~\cite{Wan2018djlW2.1812.11955}.\footnote{One can also formulate a 
lattice realization of version given in \cite{Prakash2018ugo1804.11236}.
Closely related work on this symmetry-extension method 
include \cite{Kapustin1404.3230, Tachikawa2017gyf, 2018PTEP1801.05416, Guo2018vij1812.11959} and references therein.}

We will trivialize the 4d and 2d 't Hooft anomaly of 4d YM and 2d-$\CP^{\rN-1}$ models
(again we abbreviate them as 5d Yang-Mills and 3d $\CP^{\rN-1}$ terms) 
by pullback the global symmetry to the extended symmetry.
If the pullback trivialization is possible, then it means that we
 can use the ``symmetry-extension'' method of \cite{Wang2017locWWW1705.06728}
 to construct a fully symmetry-preserving TQFT, 
 at least as an exact solvable model.\footnote{
 \label{footnote-sTQFT}
 A caveat: One needs to beware that the dimensionality affects the dynamics 
 and stability of long-range entanglement, the ``\emph{symmetry-preserving TQFT}''
 at 2d or below can be destroyed by local perturbations. 
In addition, the construction of ``\emph{symmetry-extended TQFT}'' after gauging the extended symmetry can be in fact
``\emph{spontaneously symmetry breaking}'' due to dynamics.
See detailed explorations in \cite{Wang2017locWWW1705.06728}.
 More recently, 
\Ref{Wan2019oyr1904.00994, CordovaCO2019, KO-Strings-2019-talk}
find that
the ``\emph{higher-form symmetry spontaneously breaking}'' occurs in attempts to construct a 4d higher-symmetry anomalous TQFT.
See also a systematic discussion of {higher-form symmetry spontaneously breaking} in \cite{Lake2018dqm1802.07747}.
} 
 
 In below, when we write an induced fiber sequence: 
 \bea \label{eq:KGG-extend}
 [\B K]  \to \B \mathbb{G}' \to B\mathbb{G},
 \eea 
 we mean that
  $[\B K]$ is the extension from a finite group $K$ with the classifying space $\B K$, while $B\mathbb{G}$ is the classifying space of the original full symmetry $\mathbb{G}$ (including the higher symmetry).
  Moreover, the bracket in   $[\B K]$ means that the full-anomaly-free $K$ can be dynamically gauged
  to obtain a dynamical $K$ gauge theory as a symmetry-$\mathbb{G}$ preserving TQFT, see \cite{Wang2017locWWW1705.06728}.
  
  {
  However, as noticed in \cite{Wang2017locWWW1705.06728, Wan2018djlW2.1812.11955, Wan2019oyr1904.00994},
  there are a few possibilities of dynamical fates for the attempt to construct a theory via the symmetry-extension \eqn{eq:KGG-extend}:
\begin{enumerate}[leftmargin=1.5mm, label=\textcolor{blue}{\Roman*.}, ref={\Roman*}]
\item \emph{No $\mathbb{G}'$-symmetry extended gapped phase}:\\  \label{phase-I}
$\mathbb{G}'$-symmetry extended gapped phase is impossible to construct via \eqn{eq:KGG-extend}. Namely,
a $\mathbb{G}$-anomaly cannot be trivialized by pulling back to be $\mathbb{G}'$-anomaly free. 
Although we cannot prove that the symmetry-preserving gapped phase is impossible in general (say, beyond the symmetry-extension method
of \cite{Wang2017locWWW1705.06728}), the recent works \cite{Wan2018djlW2.1812.11955, CordovaCO2019, KO-Strings-2019-talk}
suggest a strong correspondence between ``the impossibility of symmetric gapped phase'' and 
``the non-existence of such $\mathbb{G}'$.''
\item \emph{$\mathbb{G}'$-symmetry extended gapped phase}:\\ \label{phase-II}
$\mathbb{G}'$-symmetry extended gapped phase can be constructed via \eqn{eq:KGG-extend}.
Namely, there exists certain $\mathbb{G}'$, such that
a $\mathbb{G}$-anomaly can be trivialized by pulling back to be $\mathbb{G}'$-anomaly free. 
However, there are at least two possible fates after dynamically gauging $K$:
\begin{enumerate}[leftmargin=2.0mm, label=\textcolor{blue}{(\alph*)}, ref={(\alph*)}]
\item \emph{$\mathbb{G}$-spontaneously symmetry-breaking (SSB) phase}: \label{phase-i} \\
After dynamically gauging $K$, the $\mathbb{G}$-symmetry would be spontaneously broken. 
\item \emph{Anomalous $\mathbb{G}$-symmetry-preserving $K$-gauge phase}: \label{phase-ii}\\
After dynamically gauging $K$,
the $\mathbb{G}$-symmetry would not be broken,
thus we obtain a $\mathbb{G}$-symmetry-preserving and dynamical $K$-gauge TQFT.
\end{enumerate}
\end{enumerate}
In this work, we will mainly focus on determining whether the phases can be
$\mathbb{G}$-symmetry extended gapped phase (namely the phase \ref{phase-II}) or not (namely the phase  \ref{phase-I}).
If the {$\mathbb{G}'$-symmetry extended gapped phase} 
is possible, we will comment briefly about the dynamics
after gauging the extended $K$: Whether it will be spontaneously symmetry-breaking (namely the phase \ref{phase-i})
or symmetry-preserving  (namely the phase \ref{phase-ii}).
Further detail discussions about the fate of 4d SU(N)$_{\theta=\pi}$ YM dynamics of these phases are pursuit recently in 
\Ref{Wan2019oyr1904.00994} and \cite{CordovaCO2019, KO-Strings-2019-talk}. 
}

The new ingredient and generalization here we need to go beyond the symmetry-extension method of \cite{Wang2017locWWW1705.06728} are:\\
(1) \emph{Higher-symmetry extension}: We consider a higher group $\mathbb{G}$ or higher classifying space $\B \mathbb{G}$.\\
(2) {Co/Bordism group and group cohomology of higher group $\mathbb{G}$ or higher classifying space $\B \mathbb{G}$}.\\
Another companion work of ours \cite{Wan2018djlW2.1812.11955} also implements this method,
and explore the constraints on the low energy dynamics for adjoint quantum chromodynamics theory in 4d (adjoint QCD$_4$).

 We first summarize the mathematical checks, and then we will explain their physical implications in the end of this section and in \Sec{sec:con}.
  
 In the following subsections, we will not directly present quantum Hamiltonian models involoving these higher-group cohomology cocycles and Stiefel-Whitney classes.
Nonetheless, we believe that it is fairly \emph{straightforward} 
to generalize the quantum Hamiltonian models of \cite{Wan1211.3695, Wan2014woa1409.3216, Wang1404.7854} to obtain
lattice Hamiltonian models for our \Sec{sec:sym-ext-H-1}, \Sec{sec:sym-ext-H-2}, \Sec{sec:sym-ext-H-3} and \Sec{sec:sym-ext-H-4} below. 
A sketch of the design of the lattice Hamiltonian models can be found in \Ref{Wan2019oyr1904.00994}.

\subsection{$\Omega_5^{\tO}(\B^2\Z_2)$: $\Z_{4,[1]}$-symmetry-extended but $\Z_{2,[1]}$-spontaneously symmetry breaking}
\label{sec:sym-ext-H-1}

We consider $B_2\Sq^1 B_2 + \Sq^2 \Sq^1 B_2 +K_1w_1(TM)^2\Sq^1B_2$ of \eq{eq:SU2YM-5dSPT-0} and \eq{eq:SU2YM-5dSPT}
for 4d SU($\rN$)$_{\theta=\pi}$-YM's anomaly at $\rN=2$.

Since
$\Sq^2\Sq^1B_2=(w_2(TM)+w_1^2(TM))\Sq^1B_2$
and
$\Sq^1B_2$ can be trivialized by $\B^2\Z_4\to\B^2\Z_2$ since when $B_2=B_2'\mod2$, $B_2':M\to\B^2\Z_4$, and $\Sq^1B_2=2\beta_{(2,4)}B_2'=0$ (see Appendix \ref{Bockstein}).

So
$B_2\Sq^1 B_2 + \Sq^2 \Sq^1 B_2 +K_1w_1(TM)^2\Sq^1B_2$
can be trivialized via 
\bea
[ \B^2\Z_{2,[1]}] \to \B \tO(d) \times \B^2\Z_{4,[1]}^e \to \B\tO(d) \times\B^2\Z_{2,[1]}^e,\quad\quad
\eea
which we shorthand the above induced fibration as
as
\bea
[ \B^2\Z_{2,[1]} ] \to \B {\mathbb{G}'}  \to \B {\mathbb{G}}.
\eea

{Given 
\bea
\omega_5^{\mathbb{G}}=B_2\Sq^1 B_2 + \Sq^2 \Sq^1 B_2 +K_1w_1(TM)^2\Sq^1B_2 \nn\\
= (B_2 + (1+K_1)w_1(TM)^2+w_2(TM)) \Sq^1 B_2, \quad
\eea
we find that pulling back ${\mathbb{G}}$ to ${\mathbb{G}'}$, we need to solve that
\bea
\omega_5^{\mathbb{G}'}=\delta \beta_4^{\mathbb{G}'}
\eea 
with
 the split cochain solution
\bea  \label{eq:beta4G'}
\beta_4^{\mathbb{G}'}
= (B_2+ (1+K_1)w_1(TM)^2 + w_2(TM)) \cup \gamma_2^{\mathbb{G}'}. \quad \quad
\eea
Here we define that
$\gamma_2^{\mathbb{G}'}$ satisfies 
\be \label{eq:Sq1B=dgamma}
\Sq^1 B_2^{\mathbb{G}'}  =\delta \gamma_2^{\mathbb{G}'} \quad
\ee
with a solution
\bea
 \gamma_2^{\mathbb{G}'}(g') \equiv \frac{g'^2-g'}{2}, \quad \gamma_2\in C^2(\B^2\mathbb{Z}_4,\mathbb{Z}_2).
\eea
with $g' \in \Z_4$.
For a 2-simplex/2-plaquette $ijk$, let $g'=g'_{ijk}$, so 
$
\gamma_2(g') = (g'^2 -g')/2 \mod 2,
$
where $g'\in \Z_4$ assigned on a 2-simplex, while $\gamma_2(g')$ maps the input $g'\in \Z_4$ to the output $\Z_2$-valued cochain in $C^2(\B^2\mathbb{Z}_4,\mathbb{Z}_2)$.
This boils down to simply show \eqn{eq:Sq1B=dgamma}
$\Sq^1 B_2 = B_2 \hcup{1} B_2$
on a 3-simplex say with vertices 0-1-2-3 can be \emph{split} into 2-cochains $ \gamma_2$ in the following way:
\bea
&&\delta (\gamma_2)(g') = - \gamma_2(g'_{0,1,2})+\gamma_2(g'_{0,1,3})-\gamma_2(g'_{0,2,3})+\gamma_2(g'_{1,2,3})\nn\\
&&= - \gamma_2(g'_{a})+\gamma_2(g'_{b})-\gamma_2(g'_{c})+\gamma_2({g'_a-g'_b+g'_c})\nn\\
&&=(g'_b + g'_c) (g'_a + g'_b)\mod 2\nn\\
&&=\Sq^1 B_2 (r(g')).
\eea
Therefore, we can also show \eqn{eq:beta4G'}
that $\beta_4^{\mathbb{G}'}
= (B_2+ (1+K_1)w_1(TM)^2 + w_2(TM)) \cup \gamma_2^{\mathbb{G}'}$ via the above given $\gamma_2$.

Using the data,
$\omega_5^{\mathbb{G}'}$, we can construct a \emph{${\mathbb{G}'}$-symmetry extended gapped phase} (namely the phase \ref{phase-II}).
Using the pair of the above data,
$\omega_5^{\mathbb{G}}$ and $\beta_4^{\mathbb{G}'}$,
we also hope to construct the 4d fully symmetry-preserving TQFT with an emergent 2-form 
$\Z_{2}$ gauge field
(given by $\beta_4^{\mathbb{G}'}$ and via gauging the 1-form $\Z_{2,[1]}$-symmetry) 
living on the boundary of 5d SPT (given by $\omega_5^{\mathbb{G}'}$).
However, it turns out that gauging $K$ results in 
\emph{${\mathbb{G}}$-spontaneously symmetry breaking} (SSB in 1-form $\Z_{2,[1]}^e$, namely the phase \ref{phase-i}).
The SSB phase agrees with the analysis in Sec.~8 of \Ref{Wan2019oyr1904.00994} and \cite{CordovaCO2019, KO-Strings-2019-talk}
}

\subsection{$\Omega_3^{\tO}(\B\tO(3))$: $\Z_4^T$-symmetry-extended but $\Z_2^T$-spontaneously symmetry breaking}
\label{sec:sym-ext-H-2}

We consider 
$w_1(E)^3+w_1(TN^3)w_2(V_{\SO(3)})+w_1(E) w_2(V_{\SO(3)})
+ K_1(w_1(E)^3+w_1(E) w_1(TN^3)^2) $ of \eq{eq:2d-anomaly-CP1} and 
\eq{eq:CP1-3dSPT} for 2d $\CP^{\rN-1}_{\theta=\pi}$-model's anomaly  at $\rN=2$.

Since $w_2(V_{\SO(3)})$ can be trivialized in $\SU(2)=\Spin(3)$.
Also $w_1(E)^3$ can be trivialized by 
$$
\Z_4^C\to \Z_2^C,
$$ 
and
since $\Sq^2w_1(E)=(w_2(TN^3)+w_1(TN^3)^2)w_1(E)=0$, $w_1(E)w_1(TN^3)^2=w_1(E)w_2(TN^3)$ can be trivialized by 
$$
\Pin^+(d)\to\tO(d).
$$

In summary, 
$w_1(E)^3+w_1(TN^3)w_2(V_{\SO(3)})+w_1(E) w_2(V_{\SO(3)})
+ K_1(w_1(E)^3+w_1(E) w_1(TN^3)^2)$
can be trivialized via an induced fiber sequence:
\begin{multline} 
[\B(\Z_2)^3] \to \B \Pin^+(d) \times \B \SU(2) \times \B\Z_4^C   \\
\to \B \tO(d) \times \B\PSU(2) \times \B\Z_2^C.
\end{multline}
The above shows that the anomaly can be trivialized in ${{G}'}=\Pin^+ \times \SU(2) \times \Z_4^C$,
we can construct a \emph{${{G}'}$-symmetry extended gapped phase} (namely the phase \ref{phase-II}). 
However, it turns out that gauging $K$ results in 
\emph{${{G}}$-spontaneously symmetry breaking} (SSB in 0-form symmetry here, namely the phase \ref{phase-i}).
The SSB phase agrees with the analysis in Appendix A.2.4 of \Ref{Wang2017locWWW1705.06728} and \cite{CordovaCO2019, KO-Strings-2019-talk}

\subsection{$\Omega_5^{\tO}(\B\Z_2\ltimes\B^2\Z_4)$: $\Z_{8,[1]}$-symmetry-extended but $\Z_{4,[1]}$-spontaneously symmetry breaking}
\label{sec:sym-ext-H-3}

We consider $\tilde B_2 \beta_{(2,4)}  B_2+K_{C,1}Aw_1(TM)^2 B_2 $ of 
\eq{eq:SU4YM-5dSPT-0} and \eq{eq:SU4YM-5dSPT}
for 4d SU($\rN$)$_{\theta=\pi}$-YM's anomaly at $\rN=4$.

Notice
$\beta_{(2,4)}B_2$ can be trivialized by $\B^2\Z_8\to\B^2\Z_4$, and notice that $B_2=B_2'\mod4$, $B_2':M\to\B^2\Z_8$, 
$\beta_{(2,4)}B_2=2\beta_{(2,8)}B_2'=0$ (see Appendix \ref{Bockstein}).

Since $w_1(TM^5)^2$ is trivialized in the group $\tE(d)\subset \tO(d)\times\Z_4$ defined in \cite{Freed2016}
which consists of the pairs $(A,j)$ with $\det A=j^2$. 

So
$\tilde B_2 \beta_{(2,4)}  B_2 +K_{C,1}Aw_1(TM)^2 B_2 $
can be trivialized via an induced fiber sequence:
\begin{multline} 
[ \B\Z_2\times \B^2\Z_{2,[1]}] \to \B \tE(d) \times \B\Z_2^C\ltimes \B^2\Z_{8,[1]}^e \\
\to \B\tO(d) \times \B\Z_2^C\ltimes\B^2\Z_{4,[1]}^e.
\end{multline} 
The above shows that the anomaly can be trivialized in ${\mathbb{G}'}$,
we can construct a \emph{${\mathbb{G}'}$-symmetry extended gapped phase} (namely the phase \ref{phase-II}). 
However, it turns out that gauging $K$ results in 
\emph{${\mathbb{G}}$-spontaneously symmetry breaking} (SSB in 1-form $\Z_{4,[1]}^e$ symmetry here, namely the phase \ref{phase-i}).
The SSB phase agrees with the analysis in Appendix A.2.4 of \Ref{Wang2017locWWW1705.06728} and \cite{CordovaCO2019, KO-Strings-2019-talk}
It also agrees with the fact found in  \Ref{Wang2017locWWW1705.06728} that the 1+1D symmetry-preserving bosonic TQFT is not robust against local perturbation,
thus  this TQFT flows to the SSB phase.

\subsection{$\Omega_3^{\tO}(\B(\Z_2\ltimes\PSU(4)))$: $\Z_4^T \times {\SU(4)}$-symmetry-extended but $\Z_2^T  \times {\PSU(4)}$-spontaneously symmetry breaking}
\label{sec:sym-ext-H-4}
We consider the 3d term \eq{eq:2d-anomaly-CP3} and \eq{eq:CP3-3dSPT} for 2d $\CP^{\rN-1}_{\theta=\pi}$-model's anomaly at $\rN=4$: 
$ \beta_{(2,4)} w_2(E) +K_{C,1}w_1(E)w_1(TN^3)^2 =\frac{1}{2} \tilde w_1(TN^3)w_2(V_{\PSU(4)}) +K_{C,1}w_1(E)w_1(TN^3)^2$.

Since there is a short exact sequence of groups:
$1\to\Z_4 \to \Z_2^C\ltimes\SU(4)\to \Z_2^C\ltimes\PSU(4)\to1$,
we have an induced fiber sequence:
$\B\Z_4\to \B(\Z_2^C\ltimes\SU(4) )\to \B(\Z_2^C\ltimes\PSU(4) )\stackrel{w_2}{\to}\B^2\Z_4$,
so $w_2(V_{\PSU(4)})$ can be trivialized by 
$$\B(\Z_2^C\ltimes\SU(4) )\to \B(\Z_2^C\ltimes\PSU(4) ).$$

Also since $\Sq^2w_1(E)=(w_2(TN^3)+w_1(TN^3)^2)w_1(E)=0$, $w_1(E)w_1(TN^3)^2=w_1(E)w_2(TN^3)$ can be trivialized by 
$$\Pin^+(d)\to\tO(d).$$

So
$ \beta_{(2,4)} w_2(E) +K_{C,1}w_1(E)w_1(TN^3)^2$
can be trivialized via
an induced fiber sequence:
\begin{multline} 
[ \B\Z_2\times\B \Z_4 ] \to \B \Pin^+(d) \times \B (\Z_2^C\ltimes\SU(4) )   \\
 \to \B \tO(d)  \times \B(\Z_2^C\ltimes\PSU(4) ).
\end{multline} 
The above shows that the anomaly can be trivialized in ${{G}'}=\Pin^+ \times (\Z_2^C\ltimes\SU(4) )$,
we can construct a \emph{${{G}'}$-symmetry extended gapped phase} (namely the phase \ref{phase-II}). 
However, it turns out that gauging $K$ results in 
\emph{${{G}}$-spontaneously symmetry breaking} (SSB in 0-form symmetry here, namely the phase \ref{phase-i}).
The SSB phase agrees with the analysis in Appendix A.2.4 of \Ref{Wang2017locWWW1705.06728} and \cite{CordovaCO2019, KO-Strings-2019-talk}.
It also agrees with the fact found in  \Ref{Wang2017locWWW1705.06728} that the 1+1D symmetry-preserving bosonic TQFT is not robust against local perturbation,
thus  this TQFT flows to the SSB phase.

\subsection{Summary on the fate of dynamics}

{
In summary, in this section, 
for all examples \Sec{sec:sym-ext-H-1}, \Sec{sec:sym-ext-H-2}, \Sec{sec:sym-ext-H-3} and \Sec{sec:sym-ext-H-4},
we have found that there exists such a finite $K$ extension such that the  $\mathbb{G}$-anomaly becomes $\mathbb{G}'$-anomaly free, via the pull back 
procedure of \eqn{eq:KGG-extend}.\\
}
%

{
Namely, we can obtain
various symmetry-$\mathbb{G}$ extended TQFTs (namely the phase \ref{phase-II}) to saturate (higher) 't Hooft anomalies of YM theories and $\mathbb{CP}^{\mathrm{N}-1}$-model,
via the $[\B K]$ extension to a higher-symmetry $\mathbb{G}'$ or a higher-classifying space $B\mathbb{G}'$. 
}\\
%

{
However, when $K$ is dynamically gauged to obtain a dynamical $K$ gauge topologically ordered TQFT, 
thanks to a caveat in footnote \ref{footnote-sTQFT}, we find that
the above particular examples of 0-form-symmetric 2d TQFT and 1-form-symmetric 4d TQFT
become 
\emph{${\mathbb{G}}$-spontaneously symmetry breaking} (SSB in 0-form symmetry for 2d and
SSB in 1-form symmetry for 4d). Namely, dynamically gauging $K$ result in the symmetry-breaking phase \ref{phase-i}
in our examples. We do not obtain symmetry-preserving gapped phase \ref{phase-ii} in the end.
}


%
%
%
%
%
\begin{widetext}

\section{Main Results summarized in Figures}
\label{sec:tables}

\subsection{SU($\rN$)-YM and $\mathbb{CP}^{\rN-1}_{\theta=\pi}$-model at $\rN=2$}

\label{sec:tables-N=2}

In \Fig{Fig-reduce-1}, we organize the $\rN=2$ case of 
4d anomalies and 5d topological terms of 4d SU(2)$_{\theta=\pi}$ YM theory (these 5d terms are abbreviated as ``5d YM terms''), as well as the 2d anomalies and 3d topological terms of 2d $\mathbb{CP}^{1}_{\theta=\pi}$-model
(these 3d topological terms are abbreviated as ``3d $\mathbb{CP}^{1}_{\theta=\pi}$ terms'').


{
Based on the discussions around \eq{eq:select-N=2YM} in \Sec{sec:B} and the \ref{Rule 1} in \Sec{sec:rule}, we proposed that the 5d Yang-Mills term for $\rN=2$ is at most
\be
(B\Sq^1B+\Sq^2\Sq^1B)+K_1w_1(TM^5)^2\Sq^1B
{= \frac{1}{2} \tilde w_1(TM) \cP_2(B_2)+K_1w_1(TM^5)^2\Sq^1B}, 
\ee
where $K_1 \in \Z_2 = \{0,1\}$.
}

{
Amusingly, recently \Ref{Wan2019oyr1904.00994} derived this precise anomaly based on a different method: putting 4d YM on unorientable manifolds, 
and then turning on background $B$ fields.
\Ref{Wan2019oyr1904.00994} also gives mathematical and physical interpretations of the $K_1$ term,
based on the gauge bundle constraint,
\bea \label{eq:GBC-N2}
w_2(V_{\PSU(2)})=w_2(V_{\SO(3)})=B+K_1w_1(TM^5)^2+K_2w_2(TM^5).
\eea
Following \Ref{Wan2019oyr1904.00994}, $K_1$ and $K_2$ in \eqn{eq:GBC-N2} are the choices of the gauge bundle constraint, with $K_1 \in \{0,1\}= \Z_2$ and $K_2 \in  \{0,1\}= \Z_2$.
The $K_1$ is associated with Kramers singlet ($\cT^2=+1$) or Kramers doublet ($\cT^2=-1$) of Wilson line under time-reversal symmetry.
The $K_2$ is related to bosonic or fermionic properties of Wilson line under quantum statistics.
}\\

\begin{figure}[!h]
\centering
\includegraphics[scale=1.18]{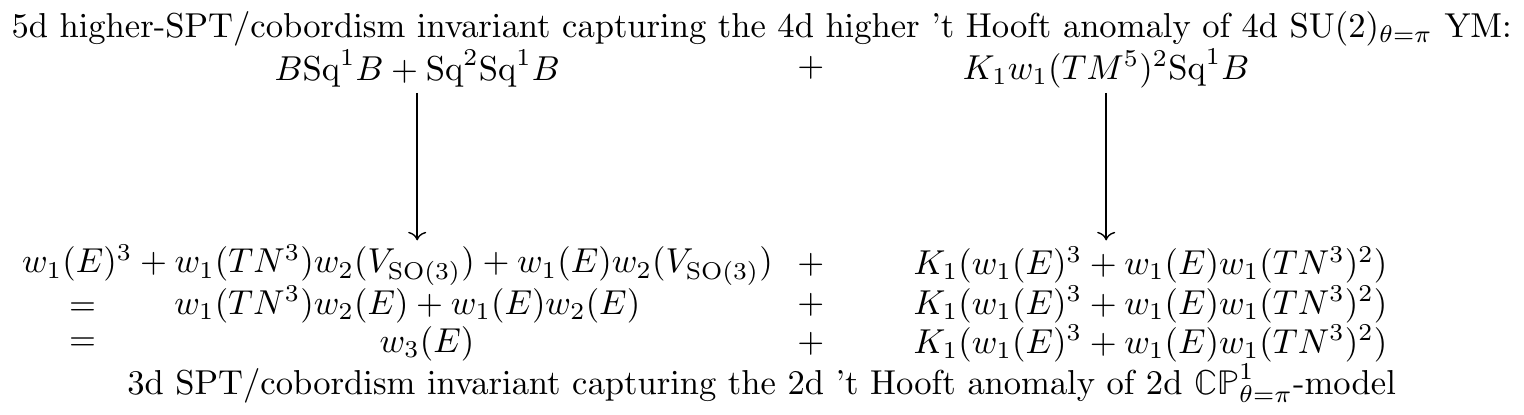}
\caption{{
A main result of our work obtains the higher 't Hooft anomalies of 4d SU(N)$_{\theta=\pi}$ YM and 't Hooft anomalies of 2d $\mathbb{CP}^{{\mathrm{N}}-1}_{\theta=\pi}$ at ${\mathrm{N}}=2$. 
We express the 4d and 2d anomalies in terms of 5d and 3d SPT/cobordism invariants respectively.  
We related the 4d and 2d anomalies by compactifying a 2-torus with twisted boundary conditions.
Some useful formulas are proved in Sec.~\ref{sec:newSUn}, \ref{sec:newCPn}  and \Ref{W2} ---
Note that on a closed 5-manifold, we rewrite:
$B_2{\mathrm{Sq}}^1B_2+{\mathrm{Sq}}^2{\mathrm{Sq}}^1B_2=\frac{1}{4} \delta(\mathcal{P}_2(B_2))=\beta_{(2,4)}\mathcal{P}_2(B_2)= \frac{1}{2} \tilde w_1(TM) \mathcal{P}_2(B_2)$,
$\Sq^2\Sq^1B_2
= (w_2(TM) +w_1(TM)^2) \Sq^1 B_2
=(w_3(TM) +w_1(TM)^3) B_2
$,
and
$B_2{\mathrm{Sq}}^1B_2=(\frac{1}{2} \tilde w_1(TM) \mathcal{P}_2(B_2)-(w_3(TM) +w_1(TM)^3) 
B_2)$.
}}
\label{Fig-reduce-1}
\end{figure}

\subsection{SU($\rN$)-YM and $\mathbb{CP}^{\rN-1}_{\theta=\pi}$-model at $\rN=4$}
\label{sec:tables-N=4}


In \Fig{Fig-reduce-2}, 
we organize the $\rN=4$ case of 
4d anomalies and 5d topological terms of 4d SU(4)$_{\theta=\pi}$ YM theory (these 5d terms are abbreviated as ``5d YM terms''):
{
\bea
&&B\beta_{(2,4)}B+K_1' w_1(TM^5)^2\beta_{(2,4)}B+K_C' A^2\beta_{(2,4)}B+K_{C,1}' Aw_1(TM^5)^2B, 
\eea
where $K_1',K_C',K_{C,1}' \in \Z_2 = \{0,1\}$. 
$K_1',K_C',K_{C,1}'$ are distinct couplings different from the gauge bundle constraint couplings $K_1,K_C,K_{C,1}$ in \Eq{eq:gauge-bundle-constraint-N=4}.
In \Sec{sec:Set-UpforN=4}, we also mentioned the appearances of $w_1(TM^5)^2\beta_{(2,4)}B$ and $A^2\beta_{(2,4)}B$ for the anomaly of 4d SU(4) YM are unlikely and the full discussion is left for the future work \cite{WWZ2019-2}.}
Thus we focus on: 
\bea
&&B\beta_{(2,4)}B+K_{C,1}Aw_1(TM^5)^2B.
\eea

%
{Amusingly, similar to the discussion in \Ref{Wan2019oyr1904.00994},
we find that the 5d Yang-Mills term for $\rN=4$ is 
\be
B\beta_{(2,4)}B+K_{C,1}Aw_1(TM^5)^2B, 
\ee
where $K_{C,1}$ is from the gauge bundle constraint similar to the generalization in \Ref{Wan2019oyr1904.00994},
\bea
w_2(V_{\PSU(4)})=B+2(K_1w_1(TM^5)^2+K_2w_2(TM^5)
+K_CA^2+K_{C,1}Aw_1(TM^5))\mod4.
\eea
}
The full discussion will be left in a future work \cite{WWZ2019-2}.
We also organize the $\rN=4$ case
 of 2d anomalies and 3d topological terms of 2d $\mathbb{CP}^{3}_{\theta=\pi}$-model
(these 3d topological terms are abbreviated as ``3d $\mathbb{CP}^{3}_{\theta=\pi}$ terms'') in the bottom part of \Fig{Fig-reduce-2}.

\begin{figure}[!h]
\includegraphics[scale=1.18]{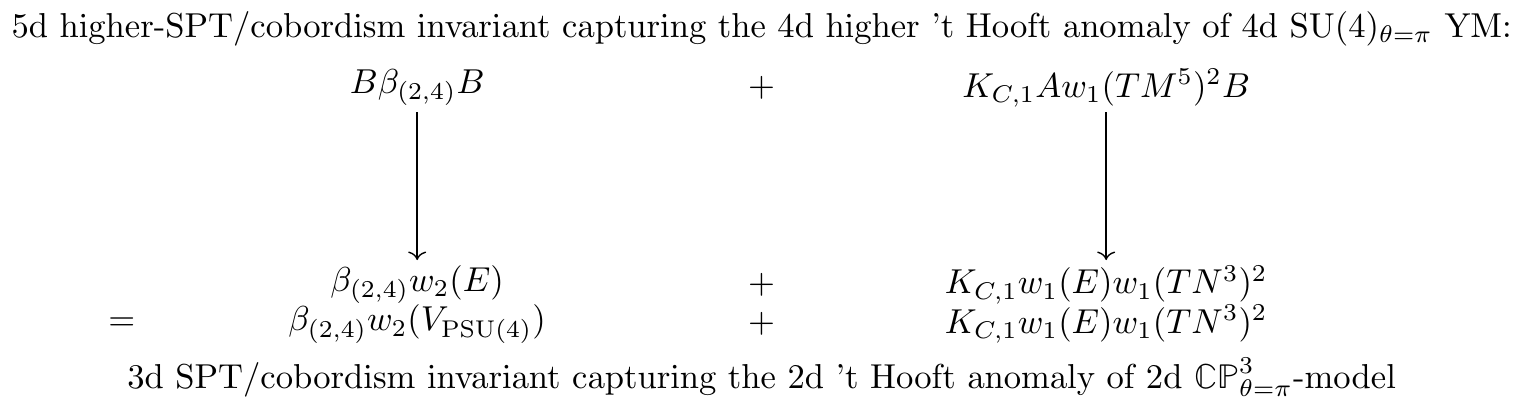}
\caption{{A main result of our work obtains the higher 't Hooft anomalies of 4d SU(N)$_{\theta=\pi}$ YM and 't Hooft anomalies of 2d $\mathbb{CP}^{\rN-1}_{\theta=\pi}$ at $\rN=4$. We express the 4d and 2d anomalies in terms of 5d and 3d SPT/cobordism invariants respectively.  
We related the 4d and 2d anomalies by compactifying a 2-torus with twisted boundary conditions.
Some useful formulas are proved in Sec.~\ref{sec:newSUn}, \ref{sec:newCPn} and \Ref{W2} ---
Note that on a closed 5-manifold, we rewrite:
$B_2\beta_{(2,4)}B_2
=\frac{1}{4}\tilde w_1(TM)\mathcal{P}_2(B_2)$.
}
}
\label{Fig-reduce-2}
\end{figure}

\pagebreak
\end{widetext}

\section{Conclusion and More Comments: Anomalies for the general N}
\label{sec:con} 

In this work, we 
propose a new and more complete set
of 't Hooft anomalies of certain 
quantum field theories (QFTs): 
time-reversal symmetric 
4d SU(N)-Yang-Mills (YM) and 2d-$\CP^{\rN-1}$ models with a topological term $\theta=\pi$,
and then give an eclectic ``proof'' of the existence of these full anomalies (of ordinary 0-form global symmetries or higher symmetries) to match these QFTs.
Our  ``proof'' is formed by a set of analyses and arguments, combining 
algebraic/geometric topology, QFT analysis, condensed matter inputs and additional physical criteria

We mainly focus on N = 2 and N = 4 cases.
As known in the literature, we actually know that N = 3 case is absent from the strict 't Hooft anomaly. 
The absence of obvious 't Hooft anomalies also apply to the more general odd integer N case (although one needs to be careful about the global consistency or 
global inconsistency, see \cite{Gaiotto2017yupZoharTTT}).
For a general even N integer, it has not been clear in the literature what are the complete 't Hoot anomalies for these QFTs.

Physically we follow the idea that coupling the global symmetry of $d$d QFTs to background fields, we can detect the higher dimensional ($d+1$d) SPTs/counter term as \eq{eq:QFT-SPT}:
\begin{multline} 
\left.
 \bZ^{\text{$d$d}}_{{\text{QFT}}}   \right|_{\text{bgd.field}=0} \nonumber\\
 { \xrightarrow{\hspace*{1cm}}} \quad
 \bZ^{\text{$(d+1)$d}}_{{\text{SPTs}}}(\text{bgd.field}) \cdot 
 \left. \bZ^{\text{$d$d}}_{{\text{QFT}}}
   \right|_{\text{bgd.field}\neq 0},  \nonumber
\end{multline}
that cannot be absorbed by $d$d SPTs.
(Here, for condensed matter oriented terminology, we follow the conventions of \cite{1405.7689}.)
This underlying $d+1$d SPTs means that the $d$d QFTs have an obstruction to be regularized with all the relevant (higher) global symmetries strictly local or onsite.
Thus this indicates the obstruction of gauging, which indicates the $d$d 't Hooft anomalies (See [\onlinecite{Wen2013oza1303.1803, 1405.7689, Wang2017locWWW1705.06728}] for QFT-oriented discussion and references therein).

We comment that the above idea \eq{eq:QFT-SPT} is \emph{distinct} from another idea also relating to coupling QFTs to SPTs,
for example used in \cite{2017arXiv171111587GPW}:
There one couples $d$d QFTs to $d$d SPTs/topological terms,
\begin{multline} \label{eq:QFT-gauge-SPT}
\left.
 \bZ^{\text{$d$d}}_{{\text{QFT}}} (A_1, B_2, .)   \right|_{\text{bgd.field}}  \overset{\text{dynamical gauging + $d$d SPTs}}{ \xrightarrow{\hspace*{3cm}}} 
  \\
  \int [{\cal D} A_1][{\cal D} B_2] \dots\,
  \bZ^{\text{$d$d}}_{{\text{QFT}}} (A_1, B_2, .)
    \cdot
   \bZ^{\text{$d$d}}_{{\text{SPTs}}}(A_1, B_2, .), 
\end{multline}
with the allowed global symmetries, and then dynamically gauging some of global symmetries.
A similar framework outlining the above two ideas, on coupling QFTs to SPTs and gauging, is also explored in \cite{Kapustin:2014gua}.

Follow the idea of \eq{eq:QFT-SPT} 
and the QFT and global symmetries information given in \Sec{sec:remark}, we 
classify all the possible anomalies enumerated by the cobordism theory computed in \Sec{sec:cobo-top}.
Then constrained by the known anomalies in the literature \Sec{sec:rev},
we follow the rules for the anomaly constraint we set in \Sec{sec:rule}
and a dimensional reduction method in \Sec{sec:5d3d},
we deduce the new anomalies of 2d-$\CP^{\rN-1}$ models in \Sec{sec:newCPn} and of 4d SU(N)-Yang-Mills (YM)  in \Sec{sec:newSUn}.

\subsection{Anomaly of 2d ${\CP^1}$ model}
To summarize the $d$d anomalies
and the $(d+1)$ cobordism/SPTs invariants of the above QFTs,
\begin{widetext}
we propose that a general anomaly formula (3d cobordism/SPT invariant) for 2d ${\CP^{\rN-1}_{\theta=\pi}}$ model at $\rN=2$ as:
\begin{multline}
 \label{eq:CP1-3dSPT}
\bZ^{\text{$2$d}}_{{{\CP^1}_{\theta=\pi}}}(w_j(TM),w_j(E),\dots)\bZ^{\text{$3$d}}_{{\text{SPTs}}} \\
\equiv
\bZ^{\text{$2$d}}_{{{\CP^1}_{\theta=\pi}}}(w_j(TM),w_j(E),\dots) \exp(\ii \pi \int_{M^3} 
\big(
w_1(E)^3+w_1(TN^3)w_2(V_{\SO(3)})+w_1(E) w_2(V_{\SO(3)}) \\
+ K_1(w_1(E)^3+w_1(E) w_1(TN^3)^2)\big))\\
=\bZ^{\text{$2$d}}_{{{\CP^1}_{\theta=\pi}}}(w_j(TM),w_j(E),\dots) \exp(\ii \pi \int_{M^3} 
\big(w_3(E)+ K_1( w_1(E)^3+w_1(E)w_1(TN^3)^2) \big)).
\end{multline}
Note that we used and derived that
$w_1(E)^3 + w_1(E) w_2(V_{\SO(3)}) = w_1(E) w_2(E)$,
 $ w_1(TN^3) w_2(V_{\SO(3)})=w_1(TN^3) w_2(E)$
 and
$\big( w_1(E)    
+ w_1(TN^3) \big) w_2(E)=w_3(E)$.

Schematically, the 2d anomaly of  ${\CP^{1}_{\theta=\pi}}$, written as a qualitative expression for its 3d SPT term, behaves as:
$$
\boxed{\sim (A_x)^3+ \cT w_2(V_{\SO(3)}) +A_x w_2(V_{\SO(3)})  + K_1 ((A_x)^3 + A_x \cT^2 ) \mod 2}.
$$
Here $\cT$ means the  dependence on 
the time-reversal background field $w_1(TN^3)$. 
Here $w_2(V_{\SO(3)})$ behaves as a topological term for the SO(3)-symmetric 1+1D Haldane chain.
From \eqn{eq:CP1-anom-A3},
the $w_1(E)$ behaves as a $\Z_2$-translation background gauge field as $w_1(\Z_2^x)$ or $A_x$.
The $w_1(TN^3)^2$ behaves as a topological term for the $\Z_2^T$-symmetric 1+1D Haldane chain.

\subsection{Anomaly of 2d ${\CP^3}$ model}

We propose that a general anomaly formula (3d cobordism/SPT invariant) for 2d ${\CP^{\rN-1}_{\theta=\pi}}$ model at $\rN=4$ as:
\begin{multline}
 \label{eq:CP3-3dSPT}
\bZ^{\text{$2$d}}_{{\CP^3_{\theta=\pi}}}(w_j(TM), \t w_j(E),\dots)\bZ^{\text{$3$d}}_{{\text{SPTs}}} \\
\equiv
{\bZ^{\text{$2$d}}_{{\CP^3_{\theta=\pi}}}(w_j(TM), \t w_j(E),\dots) \exp(\ii \pi \int_{M^3} \big(  
 \frac{1}{2} \tilde{w}_1(TN^3) w_2(V_{\PSU(4)})
{+K_{C,1}w_1(E)w_1(TN^3)^2\big) )}}\\
=\bZ^{\text{$2$d}}_{{\CP^3_{\theta=\pi}}}(w_j(TM), \t w_j(E),\dots) \exp(\ii \pi \int_{M^3} \big(  
 \frac{1}{2} \tilde{w}_1(TN^3) w_2(E)
+K_{C,1}w_1(E)w_1(TN^3)^2  \big))\\
=\bZ^{\text{$2$d}}_{{\CP^3_{\theta=\pi}}}(w_j(TM), \t w_j(E),\dots) \exp(\ii \pi \int_{M^3} \big(  
\beta_{(2,4)} w_2(E)
+K_{C,1}w_1(E)w_1(TN^3)^2  \big)).
\end{multline}
Using the fact that $w_1(TN^3)$ is the background field for $\cC\cT$, $w_1(TN^3)$ can be schematically written as $\cT+ A_x$. Hence the 2d anomaly of  ${\CP^{3}_{\theta=\pi}}$, written as a qualitative expression up to a normalization factor for its 3d SPT term, behaves as:
\begin{eqnarray}
\begin{split}
\sim&\boxed{ (\cT+A_x)  w_2(V_{\PSU(4)})    +K_{C,1} A_x (\cT+A_x)^2  \mod 2}\\
=&  \cT w_2(V_{\PSU(4)}) +  A_x w_2(V_{\PSU(4)})  +K_{C,1} (A_x^3+ A_x \cT^2)   \mod 2
\end{split}
\end{eqnarray}
Here $\cT$ means the  dependence on 
the background field for time-reversal symmetry $\Z_2^T$. 
Here $w_2(V_{\PSU(4)})$ behaves as a topological term for a ${\PSU(4)}$-symmetric 1+1D generalized spin chain.
From \eqn{eq:CP1-anom-A3},
the $w_1(E)$ behaves as a $\Z_2$-translation background gauge field written as $w_1(\Z_2^x)$ or $A_x$.
The $w_1(TN^3)^2$ behaves as a topological term for the $\Z_2^T$-symmetric 1+1D Haldane chain.
\subsection{Higher Anomaly of 4d SU(2) Yang-Mills theory}

We propose that a general anomaly formula (5d cobordism/higher SPT invariant) for 4d ${\SU(\rN)_{\theta=\pi}}$-YM theory at $\rN=2$ as:
\begin{multline}
 \label{eq:SU2YM-5dSPT}
\bZ^{\text{$4$d}}_{{\text{SU(2)YM}_{\theta=\pi}}}(w_j(TM),A,B_2,\dots)\bZ^{\text{$5$d}}_{{\text{higher-SPTs}}} \\
\equiv
\boxed{\bZ^{\text{$4$d}}_{{\text{SU(2)YM}_{\theta=\pi}}}(w_j(TM),A,B_2,\dots) \exp(\ii \pi \int_{M^5} \big(B_2\Sq^1B_2+\Sq^2\Sq^1B_2+K_1w_1(TM)^2\Sq^1B_2 \big))}\\
=
\boxed{\bZ^{\text{$4$d}}_{{\text{SU(2)YM}_{\theta=\pi}}}(w_j(TM),A,B_2,\dots) \exp(\ii \pi \int_{M^5} \big(  \frac{1}{2} \tilde w_1(TM) \cP_2(B_2) +K_1w_1(TM)^2\Sq^1B_2\big))}.
\end{multline}
Schematically, the 4d anomaly of ${\SU(2)_{\theta=\pi}}$-YM theory, written as a qualitative expression for its 5d SPT term, behaves as:
$$
\boxed{\sim (\cT B B + K_1\cT^3 B) \mod 2}.
$$
Here $\cT$ means the  dependence on 
the time-reversal background field $w_1(TM)$. 
Here $B$  means the dependence on the $\Z_{2,[1]}^e$ background field $B$.

\subsection{Higher Anomaly of 4d SU(4) Yang-Mills theory}

We propose that a general anomaly formula (5d cobordism/higher SPT invariant) for 4d ${\SU(\rN)_{\theta=\pi}}$-YM theory at $\rN=4$ as:
\begin{multline}
 \label{eq:SU4YM-5dSPT}
\bZ^{\text{$4$d}}_{{\text{SU(4)YM}_{\theta=\pi}}}(w_j(TM),A,B_2,\dots)\bZ^{\text{$5$d}}_{{\text{higher-SPTs}}} \\
\equiv
\boxed{
\bZ^{\text{$4$d}}_{{\text{SU(4)YM}_{\theta=\pi}}}(w_j(TM),A,B_2,\dots) \exp(\ii \pi \int_{M^5} \big(B_2 \beta_{(2,4)}B_2+K_{C,1}Aw_1(TM)^2B_2\big))}.
\end{multline}
Schematically the 4d anomaly of ${\SU(2)_{\theta=\pi}}$-YM theory, written as a qualitative expression for its 5d SPT term, behaves as:
$$
\boxed{\sim (\cT B B +  K_{C,1} A \cT^2 B) \mod 2}.
$$
Here $\cT$ means the  dependence on 
the time-reversal background field $w_1(TM)$. 
Here $B$  means the dependence on the $\Z_{4,[1]}^e$ background field $B$.
Here $A\equiv A_C$ means the dependence on the charge conjugation $C$ background field.
{We will leave the discussions for possible additional anomalies at N = 4 in an upcoming work \cite{WWZ2019-2}.}

\subsection{Higher Anomaly of 4d SU(N) Yang-Mills theory}

When $\rN$ is an even number, we propose that
a partial list of the 4d anomaly formula (5d cobordism/higher SPT invariant) the 4d ${\SU(\rN)_{\theta=\pi}}$-YM theory
\begin{multline}
 \label{eq:SUNYM-5dSPT}
\bZ^{\text{$4$d}}_{{\text{SU(\rN)YM}_{\theta=\pi}}}(w_j(TM),A,B_2,\dots)\bZ^{\text{$5$d}}_{{\text{higher-SPTs}}} \\
\equiv
\bZ^{\text{$4$d}}_{{\text{SU(\rN)YM}_{\theta=\pi}}}(w_j(TM),A,B_2,\dots) \exp(\ii \pi \int_{M^5} \big(B_2 \beta_{(2,\rN)}B_2+\frac{\rN}{2}\Sq^2\beta_{(2,\rN)}B_2+ \dots\big)\\
=
\boxed{\bZ^{\text{$4$d}}_{{\text{SU(\rN)YM}_{\theta=\pi}}}(w_j(TM),A,B_2,\dots) \exp(\ii \pi \int_{M^5} \big(
 \frac{1}{\rN} \tilde w_1(TM) \cP_2(B_2)
 + \dots\big)}.
\end{multline}
Note that we can derive $B_2 \beta_{(2,\rN=2^n)}B_2+\frac{\rN}{2}\Sq^2\beta_{(2,\rN)}B_2 = \frac{1}{\rN} \tilde w_1(TM) \cP_2(B_2)$, where 
Pontryagin square
$\mathcal{P}_2:\H^2(-,\Z_{2^n})\to\H^4(-,\Z_{2^{n+1}})$.
However, when charge conjugation $C$ is an additional $\Z_2$-discrete symmetry for SU(N) YM with $\rN >2$,
we foresee the additional new anomalies can happen, such as $Aw_1(TM)^2B_2$
 where $A\equiv A_C$ is the charge conjugation $C$ background field.
We will leave this additional anomalies in an upcoming work \cite{WWZ2019-2}.

\end{widetext}

We have commented about the higher symmetry analog of ``Lieb-Schultz-Mattis theorem'' 
in \Sec{sec:sTQFT},
for example, the consequences of low-energy dynamics due to the anomalies.
(For the early-history and the recent explorations on the emergent dynamical gauge fields and anomalous higher symmetries
in quantum mechanical and in condensed matter systems, see for example, \cite{Wilczek1984dhPRL} and \cite{Wen2018zux1812.02517} respectively, and references therein.)
In all examples of 4d ${\SU(\rN)_{\theta=\pi}}$-YM and 2d ${\CP^{\rN-1}_{\theta=\pi}}$ model in \Sec{sec:sTQFT},
we find the symmetry-extension method \cite{Wang2017locWWW1705.06728} or higher-symmetry-extension method \cite{Wan2018djlW2.1812.11955} can construct
their \emph{symmetry-extended gapped phases}.
However, the dynamical fates of the gauged topologically ordered gapped phases suggest them flow to become \emph{spontaneously symmetry breaking}
instead of \emph{symmetry preserving} \cite{Wan2019oyr1904.00994}.
The fact that {symmetry-preserving} gapped phase is not allowed is consistent with \Ref{CordovaCO2019, KO-Strings-2019-talk}.
We hope to address more about the dynamics in future work.

\section{Acknowledgments}

The authors are listed in the alphabetical order by a standard convention.
JW thanks the participants of Developments in Quantum Field Theory and Condensed Matter Physics (November 5-7, 2018) 
at Simons Center for Geometry and Physics at SUNY Stony Brook University
for giving valuable feedback where this work is publicly reported. 
JW thanks Harvard CMSA for the seminar invitation (March 11,  April 10 and September 10, 2019) on presenting the related work \cite{CMSA}.   
We thank Pavel Putrov and Edward Witten for helpful remarks or comments.
JW especially thanks Zohar Komargodski, and also Clay Cordova, for discussions on the issue of the potentially missing anomalies of YM theories \cite{CKW}.
JW also thanks Kantaro Ohmori, Nathan Seiberg, and Masahito Yamazaki  
for conversations. 
ZW acknowledges support from the Shuimu Tsinghua Scholar Program, NSFC grants 11431010 and 11571329. 
JW  acknowledges the Corning Glass Works Foundation Fellowship and NSF Grant PHY-1606531. 
YZ thanks the support from Physics Department of Princeton University.
This work is also supported by NSF Grant DMS-1607871 ``Analysis, Geometry and Mathematical
Physics'' and Center for Mathematical Sciences and Applications at Harvard University.

\appendix

\section{Bockstein Homomorphism}\label{Bockstein}

In general, given a chain complex $C_{\bullet}$ and a short exact sequence of abelian groups:
\bea
0\to A'\to A\to A''\to 0,
\eea
we have a short exact sequence of cochain complexes:
\bea
&&0\to\Hom(C_{\bullet},A')\to\Hom(C_{\bullet},A)\notag\\
&&\to\Hom(C_{\bullet},A'')\to0.
\eea
Hence we obtain a long exact sequence of cohomology groups:
\bea
&&\cdots\to\H^n(C_{\bullet},A')\to\H^n(C_{\bullet},A)\to\H^n(C_{\bullet},A'')\notag\\
&&\stackrel{\partial}{\to}\H^{n+1}(C_{\bullet},A')\to\cdots,
\eea
the connecting homomorphism $\partial$ is called Bockstein homomorphism.

For example,
$\beta_{(n,m)}:\H^*(-,\Z_{m})\to\H^{*+1}(-,\Z_{n})$ is the Bockstein homomorphism associated with the extension $\Z_n\stackrel{\cdot m}{\to}\Z_{nm}\to\Z_m$ where $\cdot m$ is the group homomorphism given by multiplication by $m$. In particular, $\beta_{(2,2^n)}=\frac{1}{2^n}\delta\mod2$.

Since there is a commutative diagram
\bea
\xymatrix{\Z_n\ar[r]^{\cdot m}\ar@{=}[d]&\Z_{nm}\ar[r]^{\mod m}\ar[d]^{\cdot k}&\Z_m\ar[d]^{\cdot k}\\
\Z_n\ar[r]^{\cdot km}&\Z_{knm}\ar[r]^{\mod km\quad}&\Z_{km},}
\eea
by the naturality of connecting homomorphism, we have 
the following commutative diagram:
\bea
\xymatrix{\H^*(-,\Z_m)\ar[r]^{\beta_{(n,m)}}\ar[d]^{\cdot k}&\H^{*+1}(-,\Z_n)\ar@{=}[d]\\
\H^*(-,\Z_{km})\ar[r]^{\beta_{(n,km)}}&\H^{*+1}(-,\Z_n).}
\eea

Hence we prove that
\bea\label{bsrel1}
\beta_{(n,m)}=\beta_{(n,km)}\cdot k.
\eea

In particular, since $\Sq^1=\beta_{(2,2)}$, we have $\Sq^1=\beta_{(2,4)}\cdot2$.
This formula is used in Sec. \ref{sec:sTQFT}.

\section{Poincar\'e Duality}\label{Poincare}

An orientable manifold is $R$-orientable for any ring $R$, while a non-orientable manifold is $R$-orientable if and only if $R$ contains a unit of order $2$, which is equivalent to having $2 = 0$ in $R$. Thus every manifold is $\Z_2$-orientable.

\textbf{Poincar\'e Duality}: Let $M$ be a closed connected $n$-dimensional
manifold,  $R$ is a ring, if $M$ is $R$-orientable, let $[M] \in \H_n(M,R)$ be the fundamental class for $M$ with coefficients in $R$, then
the map $\text{PD}: \H^k (M,R) \to \H_{n - k} (M,R)$  defined by $\text{PD}(\al) =[M]\cap \al$
is an isomorphism for all $k$.

%
%

\section{Cohomology of Klein bottle with coefficients $\Z_4$}\label{Klein}

In this Appendix, we derive the relation of $\beta_{(2,4)}x=z$, where $x$ is the generator of the $\Z_4$ factor of $\H^1(K,\Z_4)=\Z_4\times\Z_2$ and $z$ is the generator of $\H^2(K,\Z_2)=\Z_2$.

One $\Delta$-complex structure of Klein bottle is shown in Fig. \ref{Delta-Klein}.
Let $\alpha_i$ denote the dual cochain of the 1-simplex $a_i$ with coefficients $\Z_4$, $\lambda_i$ the dual cochain of the 2-simplex $u_i$ with coefficients $\Z_4$, let $\tilde{}$ denote its mod 2 reduction and let $\{\;\}$ denote the cohomology class.

\begin{figure}[!h]
\begin{center}
\begin{tikzpicture}[scale=2, decoration={
    markings,
    mark=at position 0.5 with {\arrow{>}}}]
\draw[postaction={decorate}] (0,0)--(1,0);
\draw[postaction={decorate}] (0,1)--(0,0);
\draw[postaction={decorate}] (1,1)--(0,1);
\draw[postaction={decorate}] (1,0)--(1,1);

\draw[postaction={decorate}] (1,0)--(0,1);

\node at (0.5,-0.1) {$a_1$};
\node at (0.5,1.1) {$a_2$};
\node at (-0.1,0.5) {$a_1$};
\node at (1.1,0.5) {$a_2$};
\node at (0.3,0.3) {$u_1$};
\node at (0.7,0.7) {$u_2$};
\node[right] at (0.5,0.5) {$a_3$};
\end{tikzpicture}
\end{center}
\caption{One $\Delta$-complex structure of Klein bottle}
\label{Delta-Klein}
\end{figure}
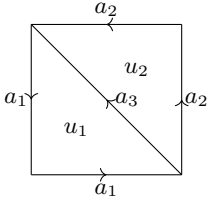

The 2-simplexes and 1-simplexes are related by the boundary differential $\partial$ of chains, namely
$\partial u_1=2a_1+a_3$, $\partial u_2=2a_2-a_3$, so 
we deduce that the boundary differential $\delta$ of cochains have the following relation:
$\delta \alpha_1=2\lambda_1$, $\delta\alpha_2=2\lambda_2$, $\delta\alpha_3=\lambda_1-\lambda_2$. 
So we deduce that the cohomology classes $\{\lambda_1\}=\{\lambda_2\}$ are the same.

Since $\delta(\alpha_1-\alpha_2-2\alpha_3)=0$, $\delta(2\alpha_1)=0$, $\H^1(K,\Z_4)=\Z_4\times\Z_2$. Let $x=\{\alpha_1-\alpha_2-2\alpha_3\}$, $y=\{2\alpha_1\}$, then $x$ generates $\Z_4$, $y$ generates $\Z_2$, $x\mod2=\{\tilde{\alpha}_1+\tilde{\alpha}_2\}$, $y\mod2=0$.

By the definition of cup product,
$\alpha_1^2(u_1)=\alpha_1(a_1)\cdot\alpha_1(a_1)=1$, $\alpha_1^2(u_2)=\alpha_1(a_2)\cdot\alpha_1(a_2)=0$, so $\alpha_1^2=\lambda_1$, similarly $\alpha_2^2=\lambda_2$. 

$\{\tilde{\alpha}_1+\tilde{\alpha}_2\}^2=\{\tilde{\alpha}_1\}^2+\{\tilde{\alpha}_2\}^2=2z=0$ where $z=\{\tilde{\lambda}_1\}=\{\tilde{\lambda}_2\}$ is the generator of $\H^2(K,\Z_2)=\Z_2$, so $\beta_{(2,4)}x=z$.

\section{Twisted cohomology $\H^n(\B\Z_2,\Z_{4,\rho})$}

Here $\rho: \Z_2\to\text{Aut}(\Z_4)$ is a nontrivial homomorphism.

Since $\B\Z_2=S^{\infty}/\Z_2$ whose universal covering space is $S^{\infty}$, by the definition of twisted cohomology \cite{Davis-Kirk}, $\H^n(\B\Z_2,\Z_{4,\rho})$ is the $n$-th cohomology group of the cochain complex
$\Hom_{\Z\Z_2}(C_{\bullet}(S^{\infty}),\Z_4)$ where both $C_{\bullet}(S^{\infty})$ and $\Z_4$ are left $\Z\Z_2$-modules, $C_{\bullet}(S^{\infty})$ is the cellular chain complex of $S^{\infty}$ with two cells in each dimension. Denote the two cells in dimension $n$ by $e^n_1$ and $e^n_2$, and denote the action of $\Z_2$ on $C_{\bullet}(S^{\infty})$ by $\rho'$.

Then 
\bea\label{rho'1}
\rho'(e^n_1)=(-1)^ne^n_2,\text{   and   }\rho'(e^n_2)=(-1)^ne^n_1.
\eea

We have 
\bea\label{partial}
\partial_n(e^n_1)=\partial_n(e^n_2)=e^{n-1}_1-e^{n-1}_2.
\eea

$f_n\in\Hom_{\Z\Z_2}(C_{n}(S^{\infty}),\Z_4)$ satisfies 

\bea\label{rho'2}
f_n(\rho'(e^n_i))=\rho(f_n(e^n_i))=(f_n(e^n_i))^{-1}.
\eea

By \eqref{partial},
if $f_n\in\text{Ker}\delta_n$, then $f_n(e^n_1)=f_n(e^n_2)$.
While by \eqref{partial}, if $f_n\in\text{Im}\delta_{n-1}$, then there exists $f_{n-1}\in \Hom_{\Z\Z_2}(C_{n-1}(S^{\infty}),\Z_4)$ such that $f_n(e^n_1)=f_n(e^n_2)=\frac{f_{n-1}(e^{n-1}_1)}{f_{n-1}(e^{n-1}_2)}$.

If $n$ is odd, then by \eqref{rho'1} and \eqref{rho'2}, for any $f_n\in\Hom_{\Z\Z_2}(C_{n}(S^{\infty}),\Z_4)$, we have $f_n(e^n_1)=f_n(e^n_2)$.

If $n$ is even, then by \eqref{rho'1} and \eqref{rho'2}, for any $f_n\in\Hom_{\Z\Z_2}(C_{n}(S^{\infty}),\Z_4)$, we have $f_n(e^n_1)=(f_n(e^n_2))^{-1}$.

So if $n$ is odd, then $\text{Ker}\delta_n=\Z_4$, and $\text{Im}\delta_{n-1}=2\Z_4$, so $\H^n(\B\Z_2,\Z_{4,\rho})=\Z_4/2\Z_4=\Z_2$.
While if $n$ is even, then $\text{Ker}\delta_n=2\Z_4$, and $\text{Im}\delta_{n-1}=0$, so $\H^n(\B\Z_2,\Z_{4,\rho})=2\Z_4=\Z_2$.

\section{Cohomology of $\B\Z_2\ltimes\B^2\Z_4$}\label{BZ2B2Z4}
The reference for this appendix is the appendix of \cite{2013arXiv1309.4721K}.

In order to compute $\Omega_5^{\tO}(\B\Z_2\ltimes\B^2\Z_4)$, we need the data of $\H^n(\B\Z_2\ltimes\B^2\Z_4,\Z_2)$ for $n\le5$.

Let $\mathbb{G}$ be a 2-group with $\B\mathbb{G}=\B\Z_2\ltimes\B^2\Z_4$.
By the Universal Coefficient Theorem, 
\bea
\H^n(\B \mathbb{G},\Z_2)&=&\H^n(\B \mathbb{G},\Z)\otimes\Z_2\oplus\notag\\
&&\text{Tor}(\H^{n+1}(\B \mathbb{G},\Z),\Z_2).
\eea
So we need only compute $\H^n(\B\Z_2\ltimes\B^2\Z_4,\Z)$ for $n\le6$.

$\H^n(\B ^2\Z_4,\Z)$ is computed in Appendix C of \cite{Clement}.
\bea
\H^n(\B ^2\Z_4,\Z)=\left\{\begin{array}{lllllll}\Z&n=0\\
0&n=1\\
0&n=2\\
\Z_4&n=3\\
0&n=4\\
\Z_8&n=5\\
\Z_2&n=6\end{array}\right.
\eea

For the 2-group $\mathbb{G}$ defined by the nontrivial action $\rho$ of $\Z_2$ on $\Z_4$ and nontrivial fibration
\bea
\xymatrix{
\B ^2\Z_4  \ar[r] &\B \mathbb{G}\ar[d]\\
         &\B \Z_2}
\eea
classified by the nonzero Postnikov class $\pi\in\H^3(\B\Z_2,\Z_4)$. {Here we consider the fiber sequence
$\B ^2 \Z_{4,[1]} \to \B \mathbb{G} \to \B \Z_2 \to \B ^3 \Z_{4,[1]} \to \dots$ 
induced from a short exact sequence $1 \to \Z_{4,[1]} \to \mathbb{G} \to \Z_2 \to 1$.}
We have the Serre spectral sequence
\bea
\H^p(\B\Z_2,\H^q(\B^2\Z_4,\Z))\Rightarrow\H^{p+q}(\B\mathbb{G},\Z),
\eea
the $E_2$ page of the Serre spectral sequence is the $\rho$-equivariant cohomology 
$\H^p(\B \Z_2,\H^q(\B ^2\Z_4,\mathbb{Z}))$. The shape of the relevant piece is shown in Fig. \ref{fig:SSSfor(BZ2B2Z4)}.

\begin{figure}[!h]
\center
\begin{sseq}[grid=none,labelstep=1,entrysize=1cm]{0...7}{0...6}
\ssdrop{\Z}
\ssmoveto 1 0 
\ssdrop{0}
\ssmoveto 2 0
\ssdrop{\Z_2}
\ssmoveto 3 0
\ssdrop{0}
\ssmoveto 4 0
\ssdrop{\Z_2}
\ssmoveto 5 0
\ssdrop{0}
\ssmoveto 6 0
\ssdrop{\Z_2}
\ssmoveto 7 0
\ssdrop{0}
\ssmoveto 0 1
\ssdrop{0}
\ssmoveto 1 1
\ssdrop{0}
\ssmoveto 2 1
\ssdrop{0}
\ssmoveto 3 1
\ssdrop{0}
\ssmoveto 4 1
\ssdrop{0}
\ssmoveto 5 1
\ssdrop{0}
\ssmoveto 6 1
\ssdrop{0}
\ssmoveto 7 1
\ssdrop{0}
\ssmoveto 0 2
\ssdrop{0}
\ssmoveto 1 2
\ssdrop{0}
\ssmoveto 2 2
\ssdrop{0}
\ssmoveto 3 2
\ssdrop{0}
\ssmoveto 4 2
\ssdrop{0}
\ssmoveto 5 2
\ssdrop{0}
\ssmoveto 6 2
\ssdrop{0}
\ssmoveto 7 2
\ssdrop{0}
\ssmoveto 0 3
\ssdrop{\Z_2}
\ssmoveto 1 3
\ssdrop{\Z_2}
\ssmoveto 2 3
\ssdrop{\Z_2}
\ssmoveto 3 3
\ssdrop{\Z_2}
\ssmoveto 4 3
\ssdrop{\Z_2}
\ssmoveto 5 3
\ssdrop{\Z_2}
\ssmoveto 6 3
\ssdrop{\Z_2}
\ssmoveto 7 3
\ssdrop{\Z_2}
\ssmoveto 0 4
\ssdrop{0}
\ssmoveto 1 4
\ssdrop{0}
\ssmoveto 2 4
\ssdrop{0}
\ssmoveto 3 4
\ssdrop{0}
\ssmoveto 4 4
\ssdrop{0}
\ssmoveto 5 4
\ssdrop{0}
\ssmoveto 6 4
\ssdrop{0}
\ssmoveto 7 4
\ssdrop{0}
\ssmoveto 0 5
\ssdrop{\Z_8}
\ssmoveto 1 5
\ssdrop{\Z_2}
\ssmoveto 2 5
\ssdrop{\Z_2}
\ssmoveto 3 5
\ssdrop{\Z_2}
\ssmoveto 4 5
\ssdrop{\Z_2}
\ssmoveto 5 5
\ssdrop{\Z_2}
\ssmoveto 6 5
\ssdrop{\Z_2}
\ssmoveto 7 5
\ssdrop{\Z_2}
\ssmoveto 0 6
\ssdrop{\Z_2}
\ssmoveto 1 6
\ssdrop{\Z_2}
\ssmoveto 2 6
\ssdrop{\Z_2}
\ssmoveto 3 6
\ssdrop{\Z_2}
\ssmoveto 4 6
\ssdrop{\Z_2}
\ssmoveto 5 6
\ssdrop{\Z_2}
\ssmoveto 6 6
\ssdrop{\Z_2}
\ssmoveto 7 6
\ssdrop{\Z_2}
\ssmoveto 0 3
\ssarrow[color=black] 4 {-3}

\ssmoveto 2 3
\ssarrow[color=black] 4 {-3}

\ssmoveto 0 5
\ssarrow[color=black] 3 {-2}

\end{sseq}
\caption{Serre spectral sequence for $(\B \Z_2,\B^2\Z_4)$}
\label{fig:SSSfor(BZ2B2Z4)}
\end{figure}

Note that $p$ labels the columns and $q$ labels the rows.

The bottom row is $\H^p(\B \Z_2,\Z)$.

The universal coefficient theorem tells us that 
$\H^3(\B ^2\Z_4,\Z)$ $=\H^2(\B ^2\Z_4,\R/\Z)$ 
$=$ $\Hom(\H_2(\B ^2\Z_4,\Z),\R/\Z)$ 
$=$ $\Hom(\pi_2(\B ^2\Z_4),\R/\Z)$ 
$=$ $\Hom(\Z_4,\R/\Z)$
$=\hat \Z_4$, 
so the $q=3$ row is $\H^p(\B \Z_2,\hat \Z_4)$, where $\Z_2$ acts on $\Z_4$ via $\rho$. For example, $\H^0(\B\Z_2,\hat\Z_4)$ is the subgroup of $\Z_2$-invariant characters in 
$\hat\Z_4$. 

$f\in\hat\Z_4$ is $\Z_2$-invariant if and only if $f(x)=f(x^{-1})=f(x)^{-1}$, so $f(x)^2=1$, we have $\H^0(\B\Z_2,\hat\Z_4)=\Z_2$.

It is also known that $\H^5(\B ^2\Z_4,\Z)=\H^4(\B ^2\Z_4,\R/\Z)$ is the group of quadratic functions $q:\Z_4\to\R/\Z$. The group at $(p,q)=(0,5)$ is then the subgroup of $\Z_2$-invariant quadratic forms.

Since for any quadratic function $q$, $q(x^{-1})=q(x)$, so $q$ is always $\Z_2$-invariant.

The first possibly non-zero differential is on the $E_3$ page:
\bea
\H^0(\B \Z_2,\H^5(\B ^2\Z_4,\Z))\to \H^3(\B \Z_2,\hat \Z_4).
\eea

Following the appendix of \cite{2013arXiv1309.4721K}, 
this map sends a $\Z_2$-invariant 
quadratic form $q:\Z_4\to\R/\Z$ to $\langle \pi, - \rangle_q$, where the bracket denotes the bilinear pairing $\langle x, y \rangle_q = q(x+y)-q(x)-q(y)$.

More precisely, the value of $\langle \pi, - \rangle_q$ on the simplex $(v_0,\dots,v_3)$ is $\langle \pi(v_0,\dots,v_3), - \rangle_q$ which is in $\hat\Z_4$.

The next possibly non-zero differentials are on the $E_4$ page:
\bea
\H^j(\B \Z_2,\hat \Z_4)\to \H^{j+3}(\B \Z_2,\R/\Z)\stackrel{\sim}{\to} \H^{j+4}(\B \Z_2,\Z).\quad\quad
\eea
The first map is contraction with $\pi$.

The last relevant possibly non-zero differential is on the $E_6$ page:
\bea
\H^0(\B \Z_2,\H^5(\B ^2\Z_4,\Z))\to \H^6(\B \Z_2,\Z).
\eea

Following the appendix of \cite{2013arXiv1309.4721K}, 
this differential is actually zero.

So the only possible differentials in Fig. \ref{fig:SSSfor(BZ2B2Z4)} below degree 5 are $d_3$ from $(0,5)$ to $(3,3)$ and $d_4$ from the third row to the zeroth row.

Since $q(k)=\e^{\frac{2\pi\ii k^2}{8}}$, if $\pi(v_0,\dots,v_3)=k$, then $\langle \pi(v_0,\dots,v_3),\pi(v_0,\dots,v_3)\rangle_q=q(\pi(v_0,\dots,v_3))^2=\e^{\frac{2\pi\ii k^2}{4}}=1$ if $k=0\mod2$, so $\text{Ker}d_3^{(0,5)}=2\Z_8=\Z_4$.


If we identify $\hat\Z_4$ with $\Z_4$, then the nonzero element in the image of $q\to\langle\pi,-\rangle_q$ is just $\pi$.
So the differential $d_3^{(0,5)}$ is nontrivial.

The differential $d_4^{(0,3)}:\H^0(\B \Z_2,\hat\Z_4)\to \H^3(\B \Z_2,\R/\Z)$ is defined by $$d_4^{(2,3)}(\lambda)(v_0,\dots,v_3)=\lambda(\pi(v_0,\dots,v_3))$$
which is actually zero since $\pi(v_0,\dots,v_3)\in2\Z_4$, if we identify $\hat\Z_4$ with $\Z_4$, then this is just the cup product of $\lambda$ and $\pi$, and $\lambda\in2\Z_4$.

The differential $d_4^{(2,3)}:\H^2(\B \Z_2,\hat\Z_4)\to \H^5(\B \Z_2,\R/\Z)$ is defined by $$d_4^{(2,3)}(\chi)(v_0,\dots,v_5)=(\chi(v_0,\dots,v_2))(\pi(v_2,\dots,v_5))$$
which is also actually zero since $\pi(v_2,\dots,v_5)\in2\Z_4$, if we identify $\hat\Z_4$ with $\Z_4$, then this is just the cup product of $\chi$ and $\pi$, and $\chi(v_0,\dots,v_2)\in2\Z_4$.

The position (3,3) corresponds to the term $A^3B_2$ where $A$ and $B_2$ are explained in \Sec{sec:F}.
So only the $A^3B_2$ vanishes in $\H^5(\B\Z_2\ltimes\B^2\Z_4,\Z_2)$, hence in $\Omega_5^{\tO}(\B\Z_2\ltimes\B^2\Z_4)$.

\onecolumngrid
\pagebreak
\newpage

\bibliography{Yang-Mills.bib,Yang-Mills-JW.bib}

\end{document}